\documentclass[a4paper,onecolumn,10pt, accepted=2026-02-02]{quantumarticle}
\pdfoutput=1
\usepackage[a4paper, total={7in, 9in}]{geometry}
\usepackage[utf8]{inputenc}
\usepackage[english]{babel}
\usepackage[T1]{fontenc}
\usepackage{etoolbox}

\usepackage{xcolor}

\usepackage{mathtools}
\usepackage{physics}
\usepackage{listings}

\usepackage{relsize}
\usepackage{caption}
\usepackage{subcaption}

\usepackage{amsmath,amssymb,amsfonts,amsthm,bm}

\usepackage{dsfont}
\makeatletter
\newcommand{\opnorm}{\@ifstar\@opnorms\@opnorm}
\newcommand{\@opnorms}[1]{%
  \left|\mkern-1.5mu\left|\mkern-1.5mu\left|
   #1
  \right|\mkern-1.5mu\right|\mkern-1.5mu\right|
}
\newcommand{\@opnorm}[2][]{%
  \mathopen{#1|\mkern-1.5mu#1|\mkern-1.5mu#1|}
  #2
  \mathclose{#1|\mkern-1.5mu#1|\mkern-1.5mu#1|}
}
\makeatother

\usepackage{array}
\usepackage{multirow}
\usepackage{adjustbox}

\usepackage{csquotes}

\usepackage{booktabs}

\definecolor{codegreen}{rgb}{0,0.6,0}
\definecolor{codegray}{rgb}{0.5,0.5,0.5}
\definecolor{codepurple}{rgb}{0.58,0,0.82}
\definecolor{backcolour}{rgb}{0.95,0.95,0.92}

\lstdefinestyle{mystyle}{
    backgroundcolor=\color{backcolour},   
    commentstyle=\color{codegreen},
    keywordstyle=\color{magenta},
    numberstyle=\tiny\color{codegray},
    stringstyle=\color{codepurple},
    basicstyle=\ttfamily\footnotesize,
    breakatwhitespace=false,         
    breaklines=true,                 
    captionpos=b,                    
    keepspaces=true,                 
    numbers=left,                    
    numbersep=5pt,                  
    showspaces=false,                
    showstringspaces=false,
    showtabs=false,                  
    tabsize=2
}

\usepackage{thm-restate}

\newtheorem*{theo*}{Theorem}
\newtheorem{theo}{Theorem}[section]
\newtheorem{corr}{Corollary}[theo]

\usepackage{xurl}
\usepackage[colorlinks=true]{hyperref}

\definecolor{darkblue} {rgb}{0,0,0.5}
\definecolor{darkgreen}{rgb}{0,0.5,0}
\hypersetup{
  urlcolor   = blue,         
  linkcolor  = darkblue,     
  citecolor  = darkgreen     
}

\newcommand{\R}{\mathbb{R}}
\newcommand{\Z}{\mathbb{Z}}
\newcommand{\N}{\mathbb{N}}
\newcommand{\C}{\mathbb{C}}
\DeclareMathOperator{\E}{\mathbb{E}}
\newcommand{\Lag}{\mathcal{L}}
\newcommand{\Hil}{\mathcal{H}}

\DeclareMathOperator{\Prob}{\mathbb{P}}

\DeclareMathOperator{\Supp}{\mathrm{supp}}
\DeclareMathOperator{\ind}{\mathds{1}}

\DeclareMathOperator{\pr}{\mathrm{pr}}
\DeclarePairedDelimiter{\set}{\lbrace}{\rbrace}
\DeclarePairedDelimiter{\of}{\lparen}{\rparen}

\newcommand{\glob}{\mathrm{global}}
\newcommand{\local}{\mathrm{local}}
\newcommand{\fix}{\mathrm{fix}}
\newcommand{\Langle}{\Big\langle}
\newcommand{\Rangle}{\Big\rangle}
\newcommand{\SYK}{\mathrm{SYK}}
\newcommand{\SSYK}{\mathrm{SSYK}}

\usepackage[usestackEOL]{stackengine}
\usepackage{scalerel}
\usepackage{graphicx,amsmath}

\makeatletter
\newcommand{\subalign}[1]{%
  \vcenter{%
  
    \Let@ \restore@math@cr \default@tag
    \baselineskip\fontdimen10 \scriptfont\tw@
    \advance\baselineskip\fontdimen12 \scriptfont\tw@
    \lineskip\thr@@\fontdimen8 \scriptfont\thr@@
    \lineskiplimit\lineskip
    \ialign{\hfil$\m@th\scriptstyle##$&$\m@th\scriptstyle{}##$\hfil\crcr
      #1\crcr
    }%
  }%
}
\makeatother

\begin{document}
\title{Trotter error and gate complexity of the SYK and sparse SYK models}

\author{Yiyuan Chen}
\affiliation{Korteweg-de Vries Institute for Mathematics, University of Amsterdam, The Netherlands}

\author{Jonas Helsen}
\affiliation{QuSoft, Amsterdam, The Netherlands}
\affiliation{CWI, Amsterdam, The Netherlands}

\author{Maris Ozols}
\affiliation{Institute for Logic, Language, and Computation, University of Amsterdam, The Netherlands}
\affiliation{Korteweg-de Vries Institute for Mathematics, University of Amsterdam, The Netherlands}
\affiliation{QuSoft, Amsterdam, The Netherlands}

\begin{abstract}
The Sachdev--Ye--Kitaev (SYK) model is a prominent model of strongly interacting fermions that serves as a toy model of quantum gravity and black hole physics. In this work, we study the Trotter error and gate complexity of the quantum simulation of the SYK model using Lie--Trotter--Suzuki formulas. Building on recent results by Chen and Brand\~{a}o \cite{Anthony} -- in particular their uniform smoothing technique for random matrix polynomials -- we derive bounds on the first- and higher-order Trotter error of the SYK model, and subsequently find near-optimal gate complexities for simulating these models using Lie--Trotter--Suzuki formulas. For the $k$-local SYK model on $n$ Majorana fermions, at time $t$, our gate complexity estimates for the first-order Lie--Trotter--Suzuki formula scales with $\tilde{\mathcal{O}}(n^{k+\frac{5}{2}}t^2)$ for even $k$ and $\tilde{\mathcal{O}}(n^{k+3}t^2)$ for odd $k$, and the gate complexity of simulations using higher-order formulas scales with $\tilde{\mathcal{O}}(n^{k+\frac{1}{2}}t)$ for even $k$ and $\tilde{\mathcal{O}}(n^{k+1}t)$ for odd $k$. Given that the SYK model has $\Theta(n^k)$ terms, these estimates are close to optimal. These gate complexities can be further improved upon in the context of simulating the time evolution of an arbitrary fixed input state $\ket{\psi}$, leading to a $\mathcal{O}(n^2)$-reduction in gate complexity for first-order formulas and $\mathcal{O}(\sqrt{n})$-reduction for higher-order formulas.

We also apply our techniques to the sparse SYK model, which is a simplified variant of the SYK model obtained by deleting all but a $\Theta(n)$ fraction of the terms in a uniformly i.i.d.\ manner. We find the average (over the random term removal) gate complexity for simulating this model using higher-order formulas scales with $\tilde{\mathcal{O}}(n^{1+\frac{1}{2}} t)$ for even $k$ and $\tilde{\mathcal{O}}(n^{2} t)$ for odd $k$. Similar to the full SYK model, we obtain a $\mathcal{O}(\sqrt{n})$-reduction simulating the time evolution of an arbitrary fixed input state $\ket{\psi}$.

Our results highlight the potential of Lie--Trotter--Suzuki formulas for efficiently simulating the SYK and sparse SYK models, and our analytical methods can be naturally extended to other Gaussian random Hamiltonians.
\end{abstract}

\maketitle

\tableofcontents 

\parskip \baselineskip
\def\svmybf#1{\rotatebox{90}{\stretchto{\{}{#1}}}
\def\svnobf#1{}
\def\rlwd{.5pt}
\newcommand\notate[4][B]{%
  \if B#1\let\myupbracefill\svmybf\else\let\myupbracefill\svnobf\fi%
  \def\useanchorwidth{T}%
  \setbox0=\hbox{$\displaystyle#2$}%
  \def\stackalignment{c}\stackunder[-6pt]{%
    \def\stackalignment{c}\stackunder[-1.5pt]{%
      \stackunder[2pt]{\strut $\displaystyle#2$}{\myupbracefill{\wd0}}}{%
    \rule{\rlwd}{#3\baselineskip}}}{%
  \strut\kern9pt$\rightarrow$\smash{\rlap{$~\displaystyle#4$}}}%
}

\pagebreak

\section{Introduction}
The \emph{Sachdev--Ye--Kitaev} (SYK) model is a randomized model of $n$ strongly interacting fermions that has recently gained significant interest as a toy model for holography \cite{kitaev1, kitaev2} in the field of quantum gravity. Notably, it can capture certain characteristics of black holes such as Hawking radiation \cite{hawking_radiation}.  Although this model is solvable in certain limits \cite{IntroSYK, Maldacena_2016, Jian_2017}, its behavior at finite $n$ remains nontrivial. Simulating the SYK model on a quantum computer would help us gain insights into its properties~\cite{jafferis2022traversable}. \emph{Out-of-time-order correlators} (OTOC), which are an important measurable quantity that serve as a measure for quantum chaos \cite{Hashimoto_2017, Garc_a_Mata_2023}, can be measured for the SYK model at finite $n$ on a quantum computer \cite{Garc_a_lvarez_2017}.  Such measurements have already been performed for small numbers of $n$ on noisy superconducting quantum computers~\cite{PhysRevD.109.105002}.

The main goal of this work is to study the quantum simulation of the SYK model using \emph{Lie--Trotter--Suzuki} formulas, also referred to as \emph{product formulas}. These formulas approximate the time evolution of an arbitrary Hamiltonian, and the error of these approximations is referred to as \emph{Trotter error}. Recently Chen and Brand\~{a}o \cite{Anthony} developed sophisticated tools for bounding the Trotter error of \emph{random Hamiltonians}. Using their toolkit, adapted to fermionic Hamiltonians, we derive error bounds for the first-order and higher-order Trotter errors of the SYK model. Based on these error bounds, we subsequently estimate the gate complexity of the quantum simulation of the SYK model using product formulas.

Additionally, we study quantum simulation of the \emph{sparse} SYK model proposed by Xu, Susskind, Su, and Swingle \cite{susskind}. This is a modified version of the SYK model that requires fewer resources to simulate, but still captures several essential properties. In particular, it exhibits a maximally chaotic gravitational sector at low temperature \cite{susskind,preskill}. Moreover, \cite{Anschuetz_2025} demonstrates that there exists a certain correspondence between the sparse and dense SYK models in terms of limiting maximum energy, and similar correspondence can be found across spin glasses \cite{Crawford_2007,Anschuetz_2025}. The sparse SYK model can be obtained by probabilistically removing each local interaction in the SYK model with a given probability, such that on average there are only $\mathcal{O}(n)$ terms left. Due to the probabilistic nature of the removal, we study the Trotter error (and the corresponding gate complexity) in terms of average-case statements.

\subsection{Previous work}

The notion of Trotter error dates back to the late 1800s (in the context of Lie algebra theory) and has been widely studied in the context of quantum computation \cite{Childs}. In recent years, significant progress has been made in bounding the Trotter error \cite{Anthony,Childs,Berry_2006, Somma_2016, PhysRevLett.123.050503}. In 2019, Childs et al.\ showed that the $l$-th order Trotter error for any Hamiltonian in the form $H=\sum_{i=1}^{\Gamma}H_{i}$ scales with the sum of (spectral) norms of nested commutators \cite[Theorem~6]{Childs}:\footnote{Here, $e^{iHt}$ denotes the time-evolution operator where we hide the minus-sign in the Hamiltonian. $S_l(t)$ represents the Trotterized unitary of order $l$.}
\begin{equation}
    \norm{e^{iHt} - S_l(t)} = \mathcal{O}\left(\sum_{i_1,\dotsc,i_{l+1}=1}^{\Gamma}\norm{[H_{i_{l+1}},\cdots [H_{i_2},H_{i_1}]\cdots ]}\; t^{l+1}\right).
\end{equation}
With this, they showed that the gate complexity $G$ of simulating the time-evolution operator using higher-order product formulas scales as
\begin{equation}
    G \approx \Gamma \norm{H}_{\glob,1}t,
\end{equation}
where $\Gamma$ is the number of terms in $H=\sum_{i=1}^{\Gamma}H_{\gamma}$ \cite{Childs}, and the global 1-norm is defined as 
\begin{equation}
    \norm{H}_{\glob,1} = \sum_{i=1}^{\Gamma}\norm{H_{i}}.
\end{equation}

In 2023, Chen and Brand\~{a}o developed an improved bound on the Trotter error of \emph{random} $k$-local qubit Hamiltonians \cite{Anthony}. They studied two cases extensively. First, for a random qubit Hamiltonian $H = \sum_{i=1}^{\Gamma}H_{i}$ such that each term is independent, bounded, and has zero mean, i.e.
\begin{equation}
    \E(H_{i}) = 0, \quad \text{and} \quad \norm{H_{i}} \leq b_{i} < \infty,
\end{equation}
the gate complexity $G$ of implementing the time-evolution operator using first-order product formula scales as 
\begin{equation}
    G\approx \Gamma \norm{H}_{\glob,2}\norm{H}_{\local,2}t^2,
\end{equation}
where the global and local 2-norms are defined as 
\begin{equation}
    \norm{H}_{\glob,2} = \sqrt{\sum_{i=1}^{\Gamma} b_{i}^2}\qquad\text{ and }\qquad\norm{H}_{\local,2} = \max_{\text{site $s$}}\sqrt{\sum_{i: H_i\text{ acts on site }s}b_{i}^2}.
\end{equation}
This result applies even when the Hamiltonian is non-local, but does require each term $H_{i}$ to be bounded.

The second case applies to random $k$-local qubit Hamiltonians with Gaussian coefficients
\begin{equation}
    H = \sum_{i=1}^{\Gamma} g_{i}K_{i}
\end{equation}
where $g_{i}$'s are i.i.d.\@ standard Gaussian coefficients and $K_{i}$ deterministic matrices bounded by $\norm{K_{i}} \leq b_{i}$. Chen and Brand\~{a}o found that the gate complexity of implementing the time-evolution operator using higher-order product formulas in this case enjoys an even better scaling:
\begin{equation}
    G\approx \Gamma \norm{H}_{\local,2}t.
\end{equation}
Contrary to the previous case, this result does not require each local term $g_{i}
K_{i}$ to be bounded, as the Gaussian coefficient $g_{i}$ is an unbounded random variable. However, this result does require the Hamiltonian to be local. Both cases provide a mean square root improvement over the results of Childs et al.~\cite{Anthony}. 

However, Chen and Brand\~{a}o's results \cite{Anthony} are not directly applicable to the SYK model, because their theorems are derived in the context of local qubit Hamiltonians while the SYK models is a fermionic Hamiltonian. In particular, the Majorana fermions obey a different algebra than qubits. Contrary to qubits, Majorana fermions anti-commute when they do not collide. Mapping Majorana fermions to qubit operators (through e.g.\ the Jordan--Wigner transform) incurs inevitable locality overheads that make the bounds in~\cite{Anthony} suboptimal. Still, \cite{Anthony} contains useful theorems and tools, such as uniform smoothness and Rademacher expansions, which apply in a broader context for operators as opposed to scalars. Using these results directly in the fermionic setting we can study the Trotter error of the SYK and the sparse SYK model. 

There are other works studying the Trotter error of the SYK model. García-Álvarez et al.~\cite{Garc_a_lvarez_2017} gave a rough estimate of the gate complexity for simulating the SYK model with $k=4$ using the first-order product formula. Their result estimated that the gate complexity scales as $\mathcal{O}(n^{10}t^2)$. Asaduzzaman et al.~\cite{PhysRevD.109.105002} found an improved gate complexity of $\mathcal{O}(n^5 t^2)$ by using a graph-coloring algorithm to minimize the number of commuting clusters of terms in the qubitized Hamiltonian, although this result seems to be formulated under a norm assumption we could not verify. Both of these results only consider the first-order product formula specifically applied to the SYK model with $k=4$. For the SYK model with general $k$, Xu, Susskind, Su, and Swingle \cite{susskind} give a gate complexity estimate of $\mathcal{O}(n^{2k-1 + (k+1)/l}\log(n)t^{1+1/l})$ for simulations using higher-order product formulas (where $l$ is the degree of the Trotter formula). They have also given a gate complexity estimate of $\mathcal{O}(n^{1+1/l}\log(n)t^{1+1/l})$ for simulating the sparse SYK model using higher-order product formulas. All the works above evaluate the complexity scaling under an average case statement, i.e.\ they compute the gate complexity for which the average Trotter error of the corresponding simulation is sufficiently small. In this work, we present an improved result in gate complexity using higher-order product formulas for both the SYK model and the sparse SYK model with general $k$, formulated under a more precise probability statement to characterize sufficiently small Trotter error.

Product formulas are not the only way to simulate the SYK model. For instance,  asymmetric qubitization allows one to reach a gate complexity of $\mathcal{O}(n^{7/2}t + n^{5/2}t\cdot \mathrm{polylog}(n/\epsilon))$ for simulating\footnote{Here $\epsilon$ denotes the accuracy.} the SYK model with $k=4$ \cite{Babbush}, which is lower than the theoretical limit of product formulas. However, unlike product formulas, this algorithm is less suitable for current noisy quantum devices \cite{Anthony, PhysRevD.109.105002}. 

\subsection{Summary of results for the SYK model}
\subsubsection{Error bounds}
We present a summary of our main results here. First, we have obtained bounds for the first- and higher-order Trotter error of the SYK model, established in Theorems \ref{thm:FOE} and \ref{thm:SYKRPG} respectively (we state simplified versions below).
In the statements below, $n$ denotes the number of Majorana fermions, $k$ is the locality of the SYK model, $p$ refers to the Schatten $p$-norm, while $r$ and $l$ denote the number of rounds (also referred to as \emph{Trotter number}) and the order of the Lie--Trotter--Suzuki formula, respectively.

\begin{theo*}[Simplified version of Theorem \ref{thm:FOE}]
    The first-order Trotter error of the SYK model for fixed locality $k$ and $p\geq 2$ is bounded by 
    \begin{equation}
        \frac{\opnorm{e^{iHt} - S_1(t/r)^r}_{*,p}}{\opnorm{I}_{*,p}} \leq \Delta^{\mathrm{dense}}_1,
    \end{equation}
    where 
    \begin{equation}
        \Delta^{\mathrm{dense}}_1 = 
        \begin{cases}
           \mathcal{O}\left( p^2 \sqrt{n}t^2\left[\frac{1}{r} 
    + \frac{t}{r^2} \right]\right) & \text{if $k$ is even},\\[12pt]
     \mathcal{O}\left(p^2 n t^2\left[\frac{1}{r} 
    + \sqrt{n} \frac{t}{r^2} \right]\right) & \text{if $k$ is odd.}
        \end{cases}
    \end{equation}
\end{theo*}

\begin{theo*}[Simplified version of Theorem \ref{thm:SYKRPG}]
     The (even) $l$-th order Trotter error of the SYK model for fixed $k$ and $l,p\geq 2$ is bounded by
     \begin{equation}
         \frac{\opnorm{e^{iHt} - S_l(t/r)^r}_{*,p}}{\opnorm{I}_{*,p}} \leq \Delta^{\mathrm{dense}}_l,
     \end{equation}
     where
     \begin{equation}
         \Delta^{\mathrm{dense}}_l = \begin{cases}
             \mathcal{O}\left(\sqrt{p} n t\left( \left[\sqrt{p} \frac{t}{r}\right]^{l}  +  n^{k}\left[\sqrt{p} \frac{t}{r}\right]^{l+1} \right) \right) & \text{if $k$ is even},\\[12pt]
             \mathcal{O}\left(\sqrt{pn} t\left(\left[\sqrt{pn} \frac{t}{r}\right]^{l}  +  n^{k}\left[\sqrt{pn} \frac{t}{r}\right]^{l+1} \right) \right) & \text{if $k$ is odd.}
         \end{cases}
     \end{equation}
\end{theo*}
 
In these theorems the norm $\opnorm{A}_{*,p}$ for a random operator $A$ on a given Hilbert space $\Hil$ is either the expected Schatten $p$-norm 
\begin{equation}
    \opnorm{A}_p = \of[\big]{\E(\norm{A}_p^p)}^{\frac{1}{p}},
\end{equation}
or the $L^p$-expected operator norm
\begin{equation}
    \opnorm{A}_{\mathrm{fix},p} = \sup\left\{ \of[\Big]{\E\of[\big]{\norm{A\ket{\psi}}_{l^2}^p}}^{\frac{1}{p}} : \ket{\psi} \in \Hil, \norm{\ket{\psi}}_{l^2} = 1\right\}.
\end{equation}

\subsubsection{Gate complexity}

\begin{table}[!hbt]
\centering
\begin{tabular}{|c|c|c|}
 \hline
 \multirow{2}*{Orders of product formulas} & \multicolumn{2}{c|}{Gate complexity} \\
 \cline{2-3}
     & Even $k$ & Odd $k$ \\
 \hline \rule{0pt}{1.5\normalbaselineskip}
First-order & $n^{k + \frac{5}{2}} t^2$ & $n^{k + 3}t^2$ \\[2ex]
 \hline \rule{0pt}{1.5\normalbaselineskip}
Higher-order & $n^{k + \frac{1}{2}} t$ & $n^{k + 1} t$\\[2ex]
 \hline 
\end{tabular}
\caption{The gate complexity of simulating the SYK model on a quantum computer such that $\norm{e^{iHt} - S_l(t/r)^r} < \epsilon$. We omitted the big-$\mathcal{O}$ notation in the cells. The $\epsilon$-scaling of these bounds go with $\sim \frac{1}{\epsilon}$ for first-order formulas and $\sim \left(\frac{1}{\epsilon}\right)^{\frac{1}{l}}$ for $l$-th order formulas.}
\label{tab:simpleresults}
\end{table}

Based on the above error bounds, we can estimate the corresponding gate complexity of the quantum simulation of the $k$-local SYK model of $n$ Majorana fermions using product formulas. We define the gate complexity as the number of gates required to approximate the time-evolution operator of the SYK model using product formulas such that the Trotter error is sufficiently small. We quantify small Trotter error from two different perspectives. The first is to characterize it by the spectral norm of the difference between the time-evolution operator and the product formula approximation, which we can informally write down as 
\begin{equation}
    \norm{e^{iHt} - S_l(t/r)^r} < \epsilon.
\end{equation}
A precise statement of this is given in equation \eqref{eq:concineq}. Under this characterization of small Trotter error, we get the gate complexities shown in Table~\ref{tab:simpleresults}. Notice the different scaling behaviors for the gate complexities depending on the parity of the locality $k$ of the SYK model, which we indeed expect as the behavior of the SYK model is different for even $k$ and odd $k$.

We find that for the SYK model with $k=4$, the gate complexity using first-order product formula scales as $n^{6.5}t^2$, and the gate complexity using higher-order product formulas scales as $n^{4.5} t$. This offers an improvement of $\mathcal{O}(n^{3.5})$ and $\mathcal{O}(n^{5.5})$ in $n$ respectively compared to the previous result $\mathcal{O}(n^{10}t^2)$ obtained by \cite{Garc_a_lvarez_2017}. Our higher-order result also offers an improvement of $\mathcal{O}(n^{2.5})$ compared to the $\mathcal{O}(n^7 t)$ result obtained by \cite{susskind}\footnote{The original statement is $\mathcal{O}(n^{2k-1 + (k+1)\delta}\log(n)t^{1+\delta})$, but we have dropped $\log(n)$ and taken $\delta$ to be sufficiently small for simplicity.}, and a slight improvement of $\mathcal{O}(\sqrt{n})$ compared to the $\mathcal{O}(n^{5}t^2)$ result obtained by \cite{PhysRevD.109.105002}, although the latter result seems to assume that the Trotter number $r$ is constant in $n$, which in general is not the case when simulating the SYK model. 

Moreover, as simulating the SYK model with product formulas of any order requires at least $\Omega(n^k)$ gates to capture all Hamiltonian terms, our estimated gate complexities $\mathcal{O}(n^{k+\frac{1}{2}}t)$ and $\mathcal{O}(n^{k+1}t)$ for implementing the SYK model using higher-order product formulas with even and odd $k$ respectively, are close to optimal. 

The other way to characterize small Trotter error is by the $l^2$-norm deviation of approximating the time evolution of an arbitrary fixed input state $\ket{\psi}$ compared to its real time evolution, which we can informally write down as 
\begin{equation}
    \norm{(e^{iHt} - S_l(t/r)^r)\ket{\psi}}_{l^2} <\epsilon.
\end{equation}
A precise statement of this is given in equation \eqref{eq:concineqpsi}. Under this characterization of small Trotter error, we get the gate complexities shown in Table~\ref{tab:simpleresults_fixed}. As we can see, simulating the SYK model in this case requires fewer gates, as the gate complexity of using first-order formulas gets reduced by a factor of $\mathcal{O}(n^2)$ and the gate complexity using higher-order formulas gets reduced by a factor of $\mathcal{O}(\sqrt{n})$ compared to Table~\ref{tab:simpleresults}. Moreover, given that simulating the SYK model requires at least $\Omega(n^k)$ gates for any order or product formula, we find that higher-order formulas are optimal for approximating the time evolution of an arbitrary fixed input state $\ket{\psi}$ under the even-$k$ SYK Hamiltonians.

\begin{table}[!t]
\centering
\setlength\tabcolsep{10pt}
\begin{tabular}{|c|c|c|}
 \hline
 \multirow{2}*{Orders of product formulas} & \multicolumn{2}{c|}{Gate complexity} \\
 \cline{2-3}
     & Even $k$ & Odd $k$ \\
 \hline \rule{0pt}{1.5\normalbaselineskip}
First-order & $n^{k + \frac{1}{2}} t^2$ & $n^{k + 1}t^2$ \\[2ex]
 \hline \rule{0pt}{1.5\normalbaselineskip}
Higher-order & $n^{k} t$ & $n^{k + \frac{1}{2}} t$\\[2ex]
 \hline 
\end{tabular}
\caption{The gate complexity of simulating the SYK model on a quantum computer such that for an arbitrary fixed input state $\ket{\psi}$ we have $\norm{(e^{iHt} - S_l(t/r)^r)\ket{\psi}}_{l^2} <\epsilon$ with high probability. We omitted the big-$\mathcal{O}$ notation in the cells. The $\epsilon$-scaling of these bounds go with $\sim \frac{1}{\epsilon}$ for first-order formulas and $\sim \left(\frac{1}{\epsilon}\right)^{\frac{1}{l}}$ for $l$-th order formulas.}
\label{tab:simpleresults_fixed}
\end{table}

\subsection{Summary of results for the sparse SYK model}

We also extend our analysis to the sparse SYK model. The probabilistic removal of terms for obtaining the sparse SYK model can be equivalently achieved by attaching a Bernoulli variable $B_{\bm{i}} \in \{0,1\}$ to each term of the SYK model, such that every $B_{\bm{i}}$ takes the value 1 with a certain probability $p_B$ and the value 0 with probability $1-p_B$. The value of $p_B$ is tuned in such a way that on average, there are $\mathcal{O}(n)$ terms in the sparse SYK model.

Due to the probabilistic removal, the number of remaining terms in the sparse SYK model depends on the specific sample of the Bernoulli variables. Hence, both the Trotter error and the corresponding gate complexity will also depend on the specific sample of the Bernoulli variables. Therefore, we study the sparse SYK model in the following setting. First, we treat the Bernoulli variables as coefficients and compute the Trotter error and the gate complexity as functions depending on these coefficients. Then, by averaging over these Bernoulli coefficients, we end up with average-case statements for both the Trotter error and gate complexity of the sparse SYK model. 

We first state the following theorem bounding the average Trotter error: 
\begin{equation}
    \Langle\opnorm{\mathcal{D}(\bm{B})}_{*,p}\Rangle := \Langle\opnorm{e^{iH_{\mathrm{SSYK}}(\bm{B})t} - S_l(t/r)^r}_{*,p}\Rangle,
\end{equation}
where we use $\langle \cdot \rangle$ to denote averaging over the Bernoulli variables.
\begin{theo*}[Simplified version of Theorem \ref{thm:avgSSYKerror}]
    Consider the sparse SYK model $H_{\mathrm{SSYK}}(\bm{B})$. The $l$-th order average Trotter error~\eqref{eq:avD(B)} is bounded by
    \begin{equation}
        \frac{\Langle\opnorm{\mathcal{D}(\bm{B})}_{*,p}\Rangle }{\opnorm{I}_{*,p}} \leq \Delta^{\mathrm{sparse}}_l 
\end{equation}
where the error bound $\Delta^{\mathrm{sparse}}_l$ is given by
\begin{equation}
    \Delta^{\mathrm{sparse}}_l :=
        \begin{cases}
        n^{\frac{3-k}{2}}\sqrt{p} t\cdot\Bigg( \left[\sqrt{p} n^{\frac{1}{2}-\frac{k}{2}}\frac{t}{r}\right]^{l} 
        + n^k\left[\sqrt{p} n^{\frac{1}{2}-\frac{k}{2}}\frac{t}{r}\right]^{l+1}\Bigg),& \text{if $k$ is even},\\[12pt]
        n^{\frac{3-k}{2}}\sqrt{p} t\cdot\Bigg( \left[\sqrt{p}  n^{1-\frac{k}{2}}\frac{t}{r}\right]^{l} 
        + n^k \left[\sqrt{p} n^{1-\frac{k}{2}}\frac{t}{r}\right]^{l+1}\Bigg),& \text{if $k$ is odd}.
        \end{cases}
\end{equation}
\end{theo*}

Similarly, we find the following \emph{average} gate complexity of the simulation of the sparse SYK model using higher-order product formulas:
\begin{equation}
    \overline{G} \sim \begin{cases}
        n^{1+\frac{1}{2}} \mathcal{J}t & \text{if $k$ is even},\\
        n^2\mathcal{J}t & \text{if $k$ is odd}.
    \end{cases}
\end{equation}
such that the spectral norm difference is sufficiently small. This gate complexity, compared to the (regular) SYK model, is significantly smaller. Similar to the argument made before for the SYK model, our result for the sparse SYK model is near optimal, since simulating the sparse SYK model using any order of product formulas requires on average $\Omega(n)$ gates. 

Similar to the full SYK Hamiltonian, if we only aim for a slight deviation when approximating the time-evolution of an arbitrary fixed input state $\ket{\psi}$, we get the following improved average gate complexity
\begin{equation}
    \overline{G} \sim \begin{cases}
        n \mathcal{J}t & \text{if $k$ is even},\\
        n^{1+\frac{1}{2}}\mathcal{J}t & \text{if $k$ is odd}.
    \end{cases}
\end{equation}
which shows that higher-order formulas are also near optimal when it comes to approximating the time evolution of an arbitrary fixed input state $\ket{\psi}$ under even $k$ sparse SYK Hamiltonians. The above two results also hold for a general class of sparse Gaussian random Hamiltonians where the sparsification is regulated by the Bernoulli variables similar to the sparse SYK model.\footnote{See equation \eqref{eq:generalsparseH}.}

\subsection{Comparison with numerical simulations}
We also study the performance of our bounds relative to actual numerical computations of the normalized Trotter error, 
\begin{equation}
    \opnorm{e^{iH_{\SYK}t} - S_l(t/r)^r}_{\overline{p}} := \frac{\opnorm{e^{iH_{\SYK}t} - S_l(t/r)^r}_p}{\opnorm{I}_p}
\end{equation} 
for the dense SYK model, and the average normalized Trotter error, 
\begin{equation}
    \left\langle \opnorm{e^{iH_{\SSYK}t} - S_l(t/r)^r}_{\overline{p}} \right\rangle := \frac{\left \langle \opnorm{e^{iH_{\SYK}t} - S_l(t/r)^r}_p\right\rangle}{\opnorm{I}_p}
\end{equation} 
for the sparse SYK model. The $\langle\cdots\rangle$ averages are taken over the Bernoulli variables. We perform two types of test of our error bounds by comparing their scaling in $n$, the system size, and in $t$, the evolution time, with the numerical values. The code we used for the numerics here can be found in GitHub repository \cite{trotter2025}. 

\begin{figure}[!t]
  \begin{subfigure}[b]{0.49\textwidth}
    \includegraphics[width=\textwidth]{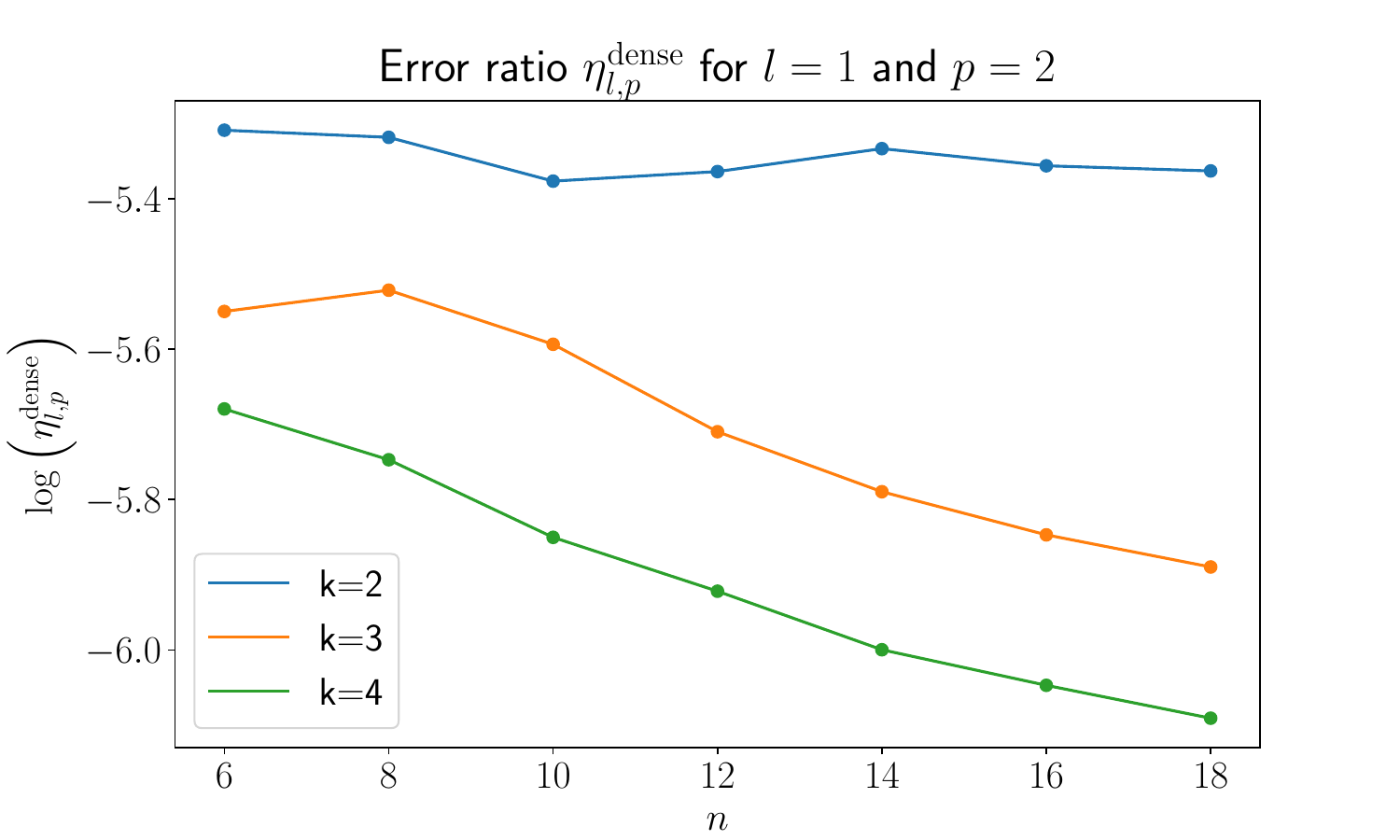}
    \caption{Error ratio of the first-order Trotter error.}
    \label{fig:error_ratio_dense1}
  \end{subfigure}
  \hfill
  \begin{subfigure}[b]{0.49\textwidth}
    \includegraphics[width=\textwidth]{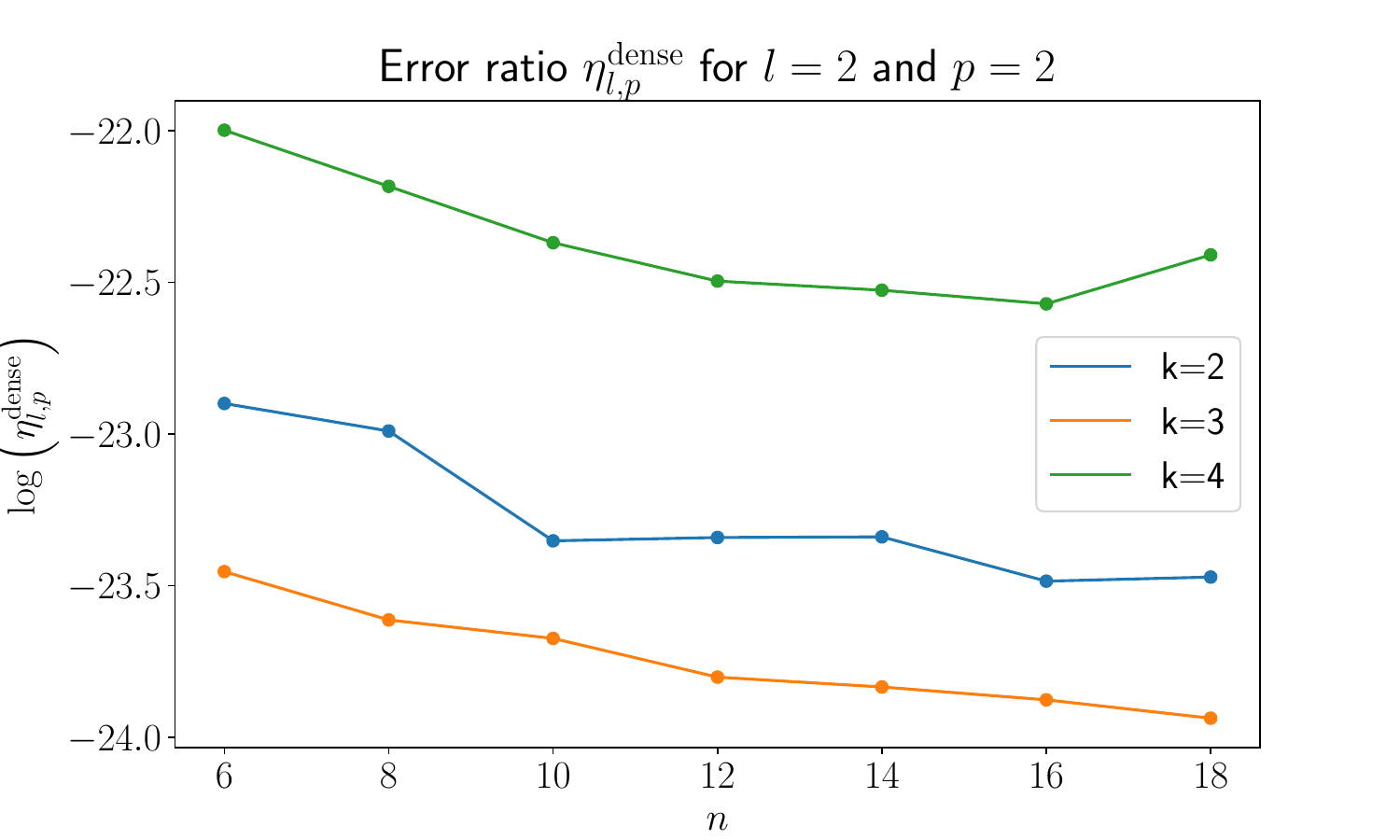}
    \caption{Error ratio of the second-order Trotter error.}
    \label{fig:error_ratio_dense2}
  \end{subfigure}
  \caption{Error ratios of the dense SYK models of different localities for even $n$ from $n=6$ to $n=18$.}
  \label{fig:error_ratio_dense}
\end{figure}
\begin{figure}[!t]
  \begin{subfigure}[b]{0.49\textwidth}
    \includegraphics[width=\textwidth]{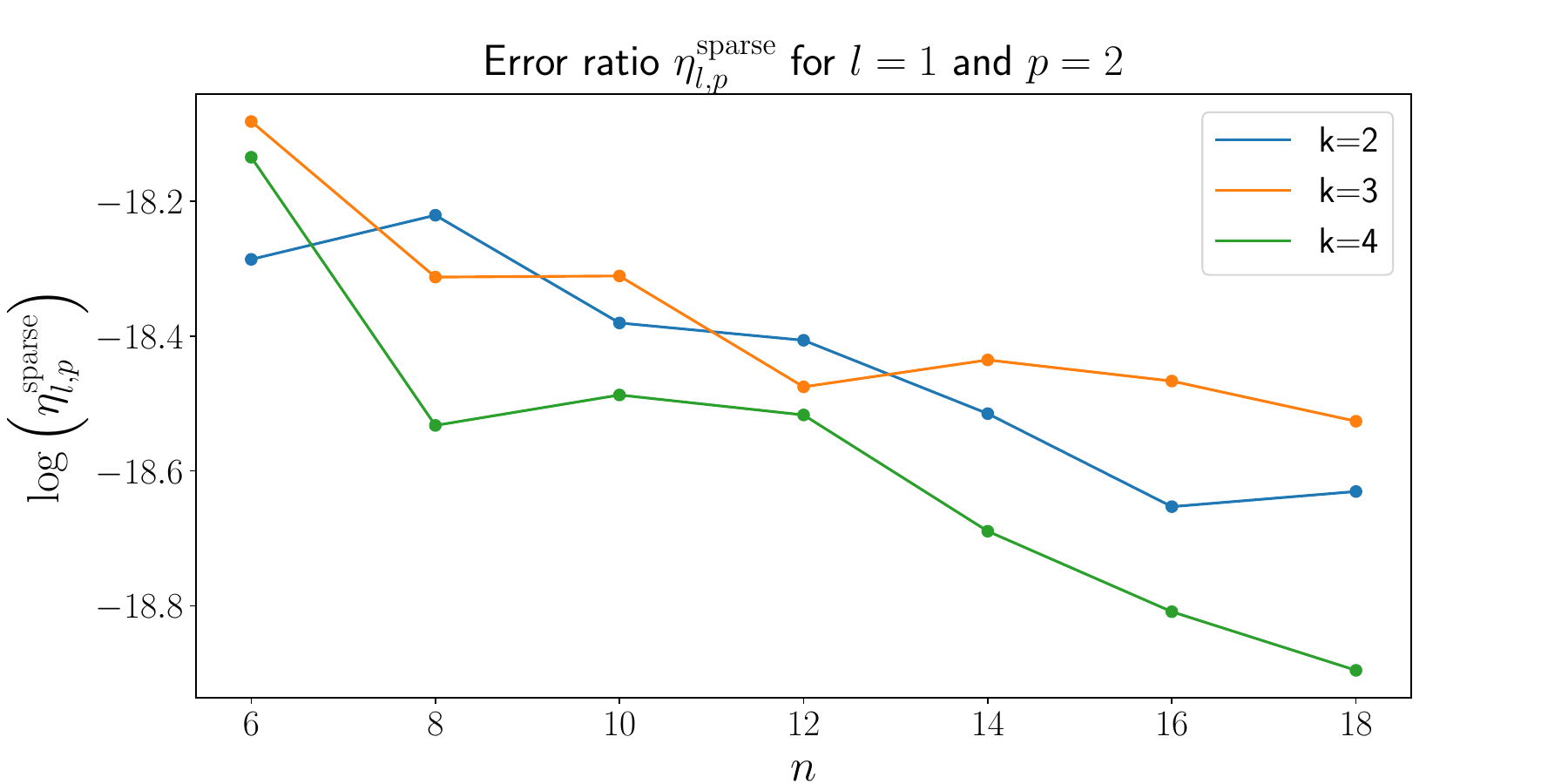}
    \caption{Error ratio of the first-order Trotter error.}
    \label{fig:error_ratio_sparse1}
  \end{subfigure}
  \hfill
  \begin{subfigure}[b]{0.49\textwidth}
    \includegraphics[width=\textwidth]{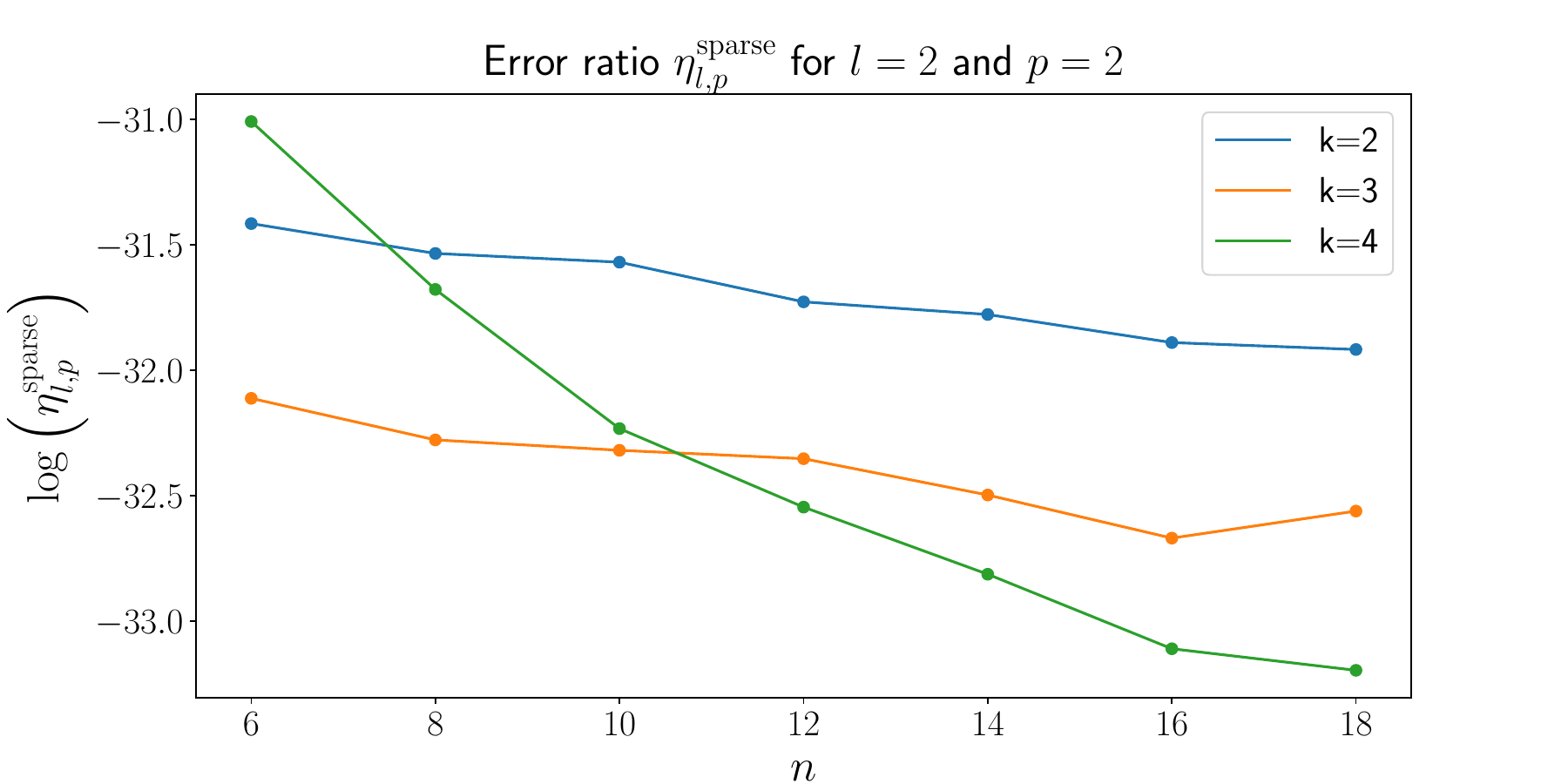}
    \caption{Error ratio of the second-order Trotter error.}
    \label{fig:error_ratio_sparse2}
  \end{subfigure}
  \caption{Error ratios of the sparse SYK models of different localities for even $n$ from $n=6$ to $n=18$.}
  \label{fig:error_ratio_sparse}
\end{figure}

\subsubsection{Scaling in \texorpdfstring{$n$}{n}}
Due to hardware limitations, our simulations could not be performed beyond $n = 18$. Consequently, it is not meaningful to compare the asymptotic scaling exponents of the analytical error bounds with the numerical results directly. To enable a more reliable comparison at small system sizes, we introduce the following \emph{error ratios}:
\begin{equation}
    \eta^{\mathrm{dense}}_{l,p}
    := 
    \frac{1}{\Delta^{\mathrm{dense}}_l}
    \bigl\| e^{i H_{\SYK} t} - S_l(t/r)^r \bigr\|_{\overline{p}},
\end{equation}
for the dense SYK model, and
\begin{equation}
    \eta^{\mathrm{sparse}}_{l,p}
    := 
    \frac{1}{\Delta^{\mathrm{sparse}}_l}
    \left\langle 
    \bigl\| e^{i H_{\SSYK} t} - S_l(t/r)^r \bigr\|_{\overline{p}}
    \right\rangle,
\end{equation}
for the sparse SYK model. 

If the analytical bounds correctly capture the finite-$n$ scaling behavior of the actual Trotter error, these ratios should remain approximately constant as $n$ varies. In particular, a constant $\eta_{l,p}$ indicates that the bound differs from the true Trotter error only by a fixed multiplicative factor. Throughout our numerical analysis, we set $p = 2$, use $r = 10^5$ Trotter steps, and focus exclusively on the first- and second-order Trotter errors.

For the dense SYK model with $k = 2, 3, 4$, we plot the error ratio $\eta^{\mathrm{dense}}_{l,p}$ for both the first- and second-order Trotter errors in Figure~\ref{fig:error_ratio_dense}. As shown, the scaling behavior predicted by our analytical bounds $\Delta^{\mathrm{dense}}_l$ agrees well with the numerical Trotter errors at small $n$. Overall, the plots display a decreasing trend, meaning that we are slightly overestimating the actual Trotter error. We do not know if this remains the case for larger $n$, as computing the Trotter error for $n > 18$ lies beyond our current computational capabilities. We nevertheless encourage future work to explore higher $n$-regime further.

We performed a similar numerical analysis for the sparse SYK model and plotted the corresponding error ratios $\eta^{\mathrm{sparse}}_{l,p}$ for the first- and second-order Trotter errors in Figure~\ref{fig:error_ratio_sparse}. In this case, the analytical bounds $\Delta^{\mathrm{sparse}}_l$ tend to overestimate the actual Trotter errors more than the dense case, particularly for the second-order cases with $k = 4$. It remains uncertain whether these ratios would approach a constant value at larger $n$, as our computational resources do not allow us to probe beyond this regime. We encourage interested readers to explore this question further.

\subsubsection{Scaling in \texorpdfstring{$t$}{t}}
To examine the scaling behavior with respect to $t$, we fix the system size to $n = 10$. In this setting, the operator dimension remains constant as $t$ varies, allowing us to compute the Trotter error for times up to $t \sim 1000$, where an asymptotic scaling exponent becomes meaningful.

For the dense SYK model, we plot the first- and second-order normalized Trotter errors together with the analytical bounds $\Delta^{\mathrm{dense}}_l$ in Figure~\ref{fig:t_dense}. The scaling exponent in $t$ is obtained via a linear fit $y = a x + b$, where $y$ corresponds to either $\log\!\bigl(\| e^{i H_{\SYK} t} - S_l(t/r)^r \|_{\overline{p}}\bigr)$ or $\log\!\bigl(\Delta^{\mathrm{dense}}_l\bigr)$, and $x = \log(t)$. Once again, we find that our error bounds accurately capture the scaling behavior in $t$.
\begin{figure}[!b]
  \begin{subfigure}[b]{0.49\textwidth}
    \includegraphics[width=\textwidth]{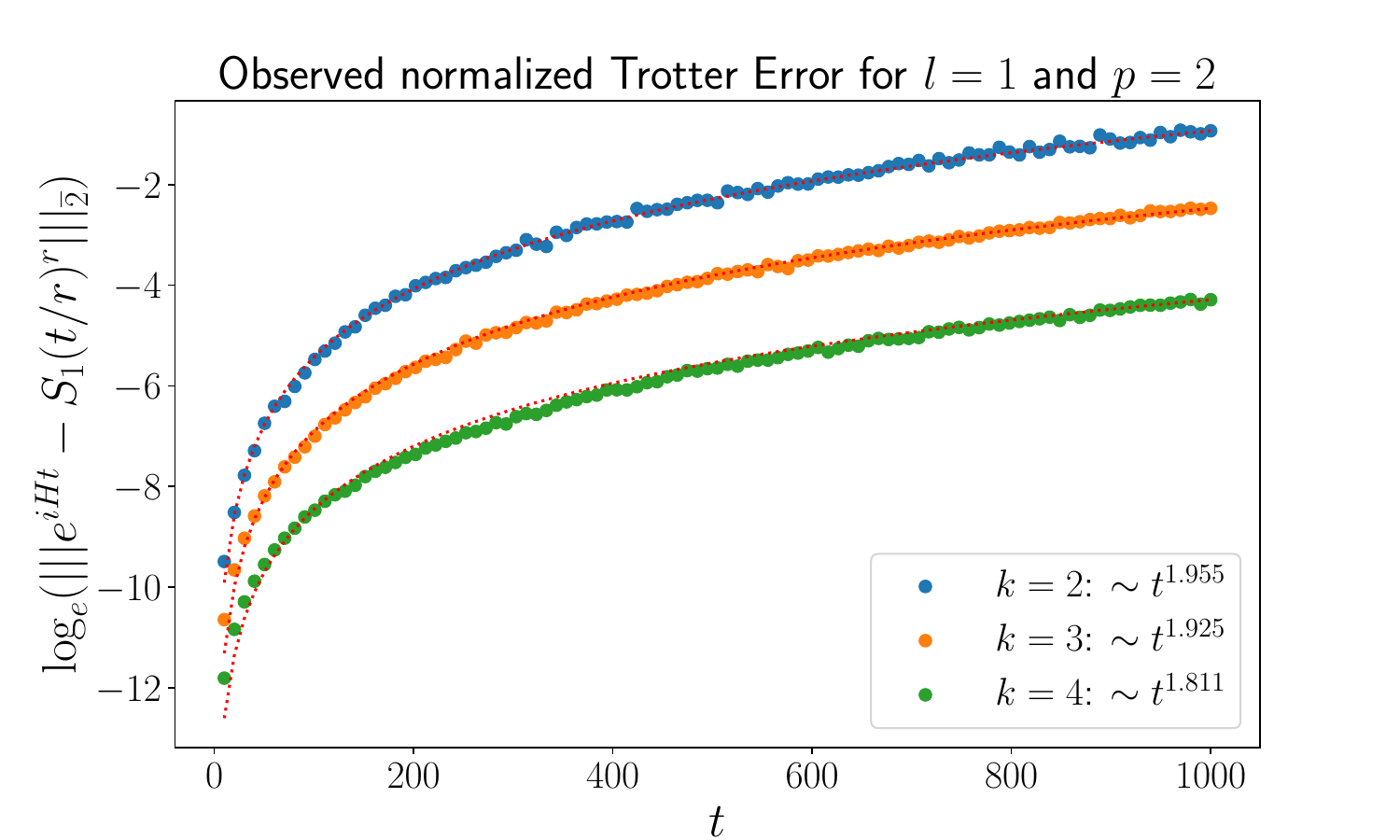}
    \caption{Observed first-order Trotter error.}
    \label{fig:t_dense_obs1}
  \end{subfigure}
  \hfill
  \begin{subfigure}[b]{0.49\textwidth}
    \includegraphics[width=\textwidth]{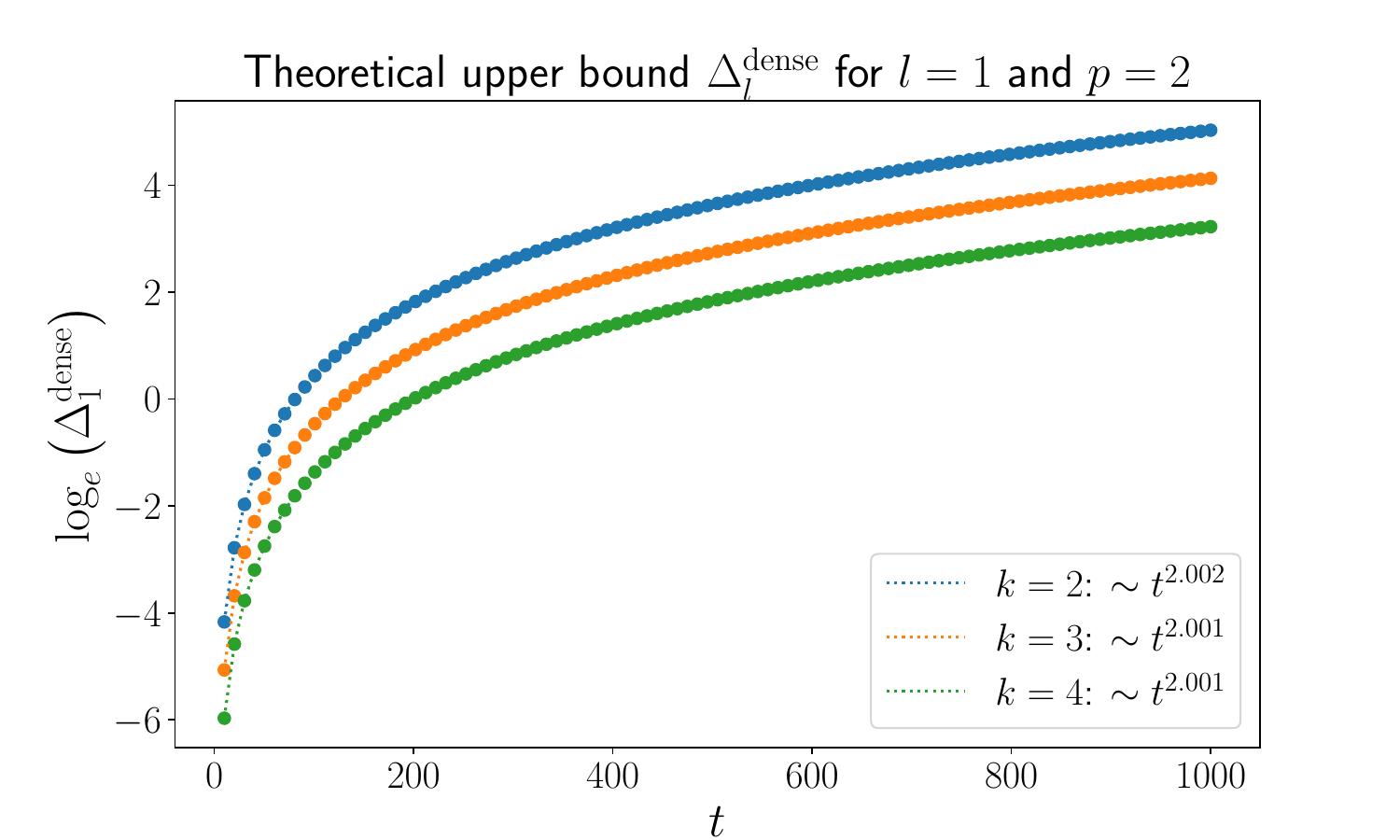}
    \caption{First-order error bound $\Delta_1^{\mathrm{dense}}$.}
    \label{fig:t_dense_theory1}
  \end{subfigure}
  \\
  \begin{subfigure}[b]{0.49\textwidth}
    \includegraphics[width=\textwidth]{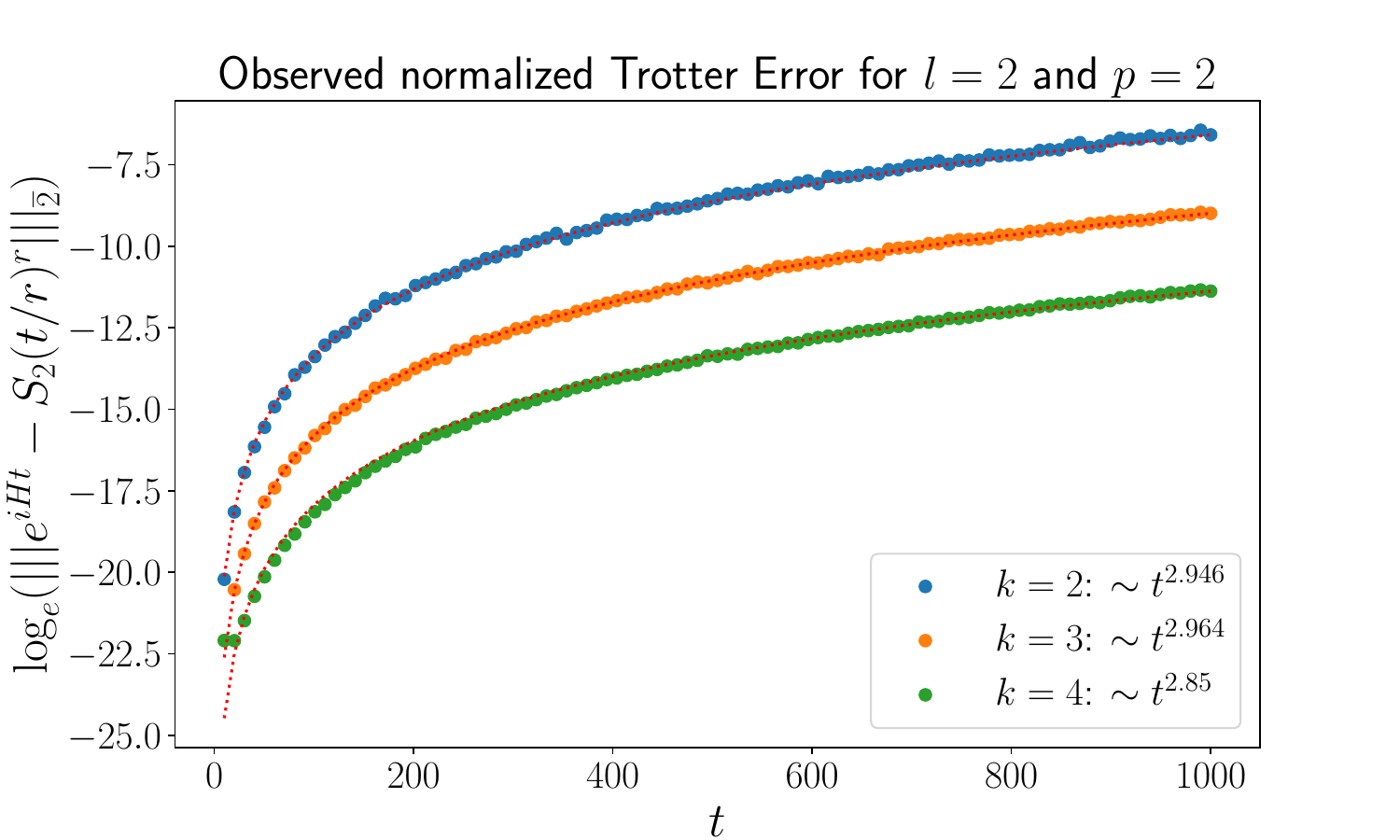}
    \caption{Observed second-order Trotter error.}
    \label{fig:t_dense_obs2}
  \end{subfigure}
  \hfill
  \begin{subfigure}[b]{0.49\textwidth}
    \includegraphics[width=\textwidth]{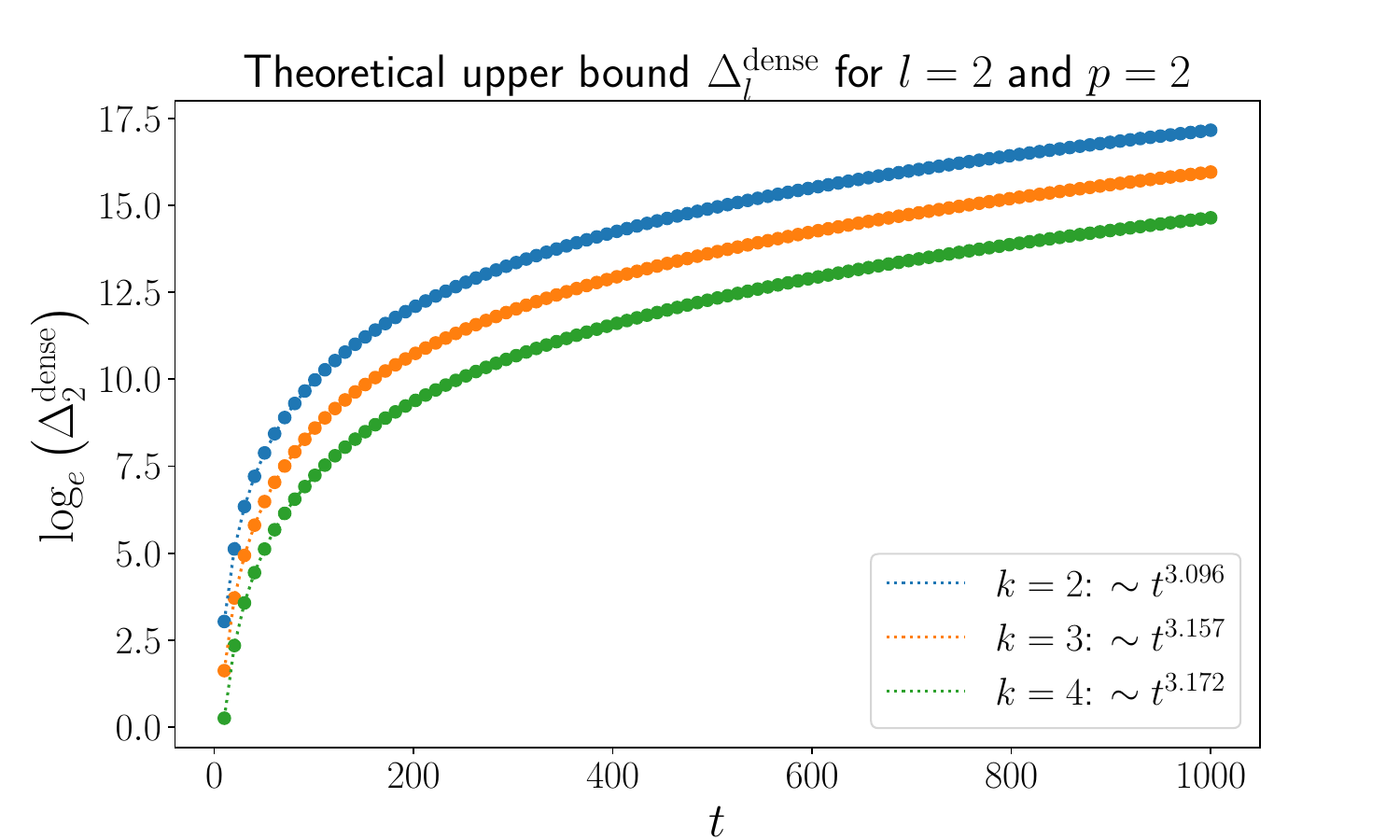}
    \caption{Second-order error bound $\Delta_2^{\mathrm{dense}}$.}
    \label{fig:t_dense_theory2}
  \end{subfigure}
  \caption{Observed normalized Trotter error and the theoretical upper bounds $\Delta_l^{\mathrm{dense}}$ for $l=1,2$. We fix $n=10$, and study the dense SYK models of localities $k=2,3,4$ with $t\in [0,1000].$}
  \label{fig:t_dense}
\end{figure}

Similarly, for the sparse SYK model, we plot the first- and second-order average normalized Trotter errors together with the analytical bounds $\Delta^{\mathrm{sparse}}_l$ in Figure~\ref{fig:t_sparse}. As in the dense case, the bounds $\Delta^{\mathrm{sparse}}_l$ provide accurate predictions for $k = 2, 3, 4$.

\begin{figure}[!t]
  \begin{subfigure}[b]{0.49\textwidth}
    \includegraphics[width=\textwidth]{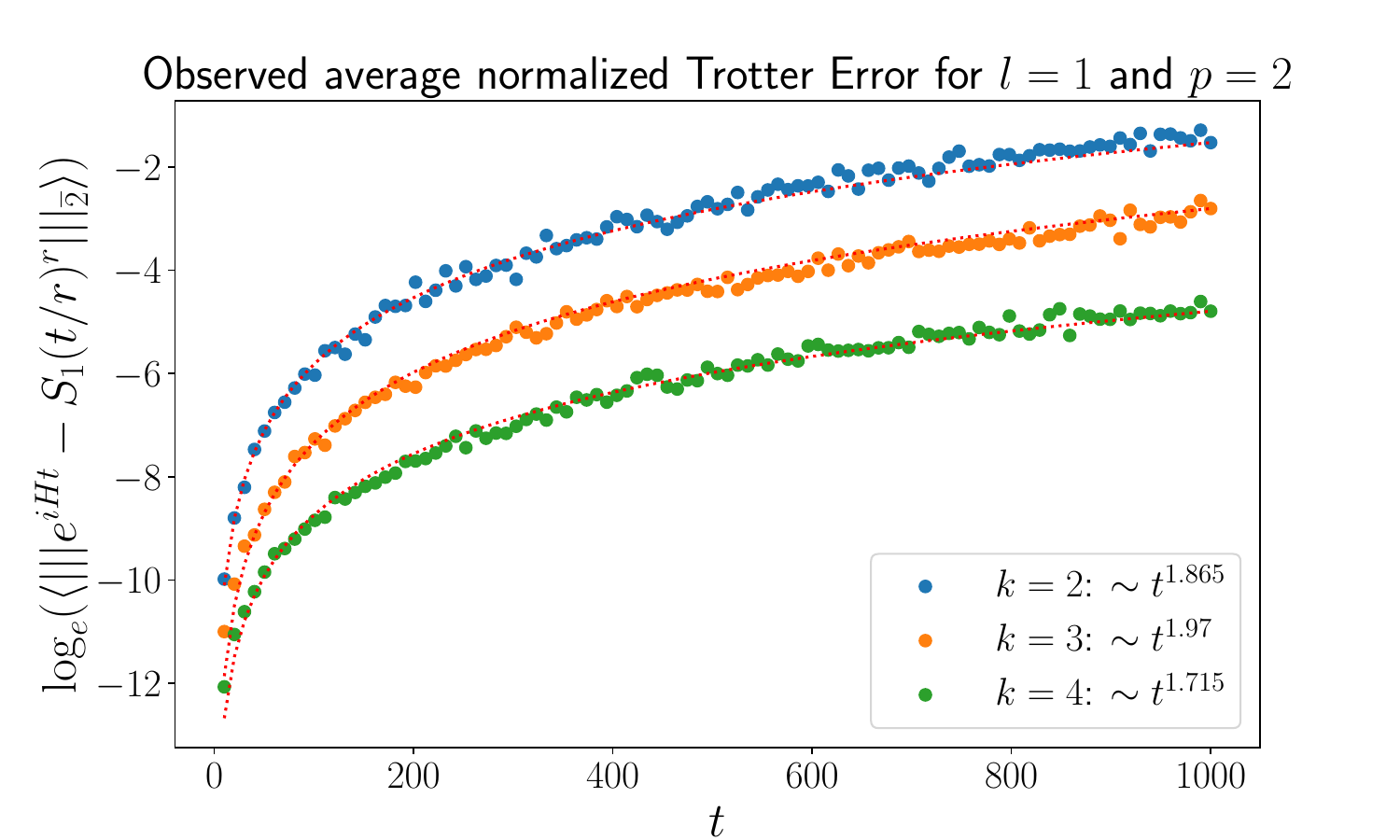}
    \caption{Observed average first-order Trotter error.}
    \label{fig:t_sparse_obs1}
  \end{subfigure}
  \hfill
  \begin{subfigure}[b]{0.49\textwidth}
    \includegraphics[width=\textwidth]{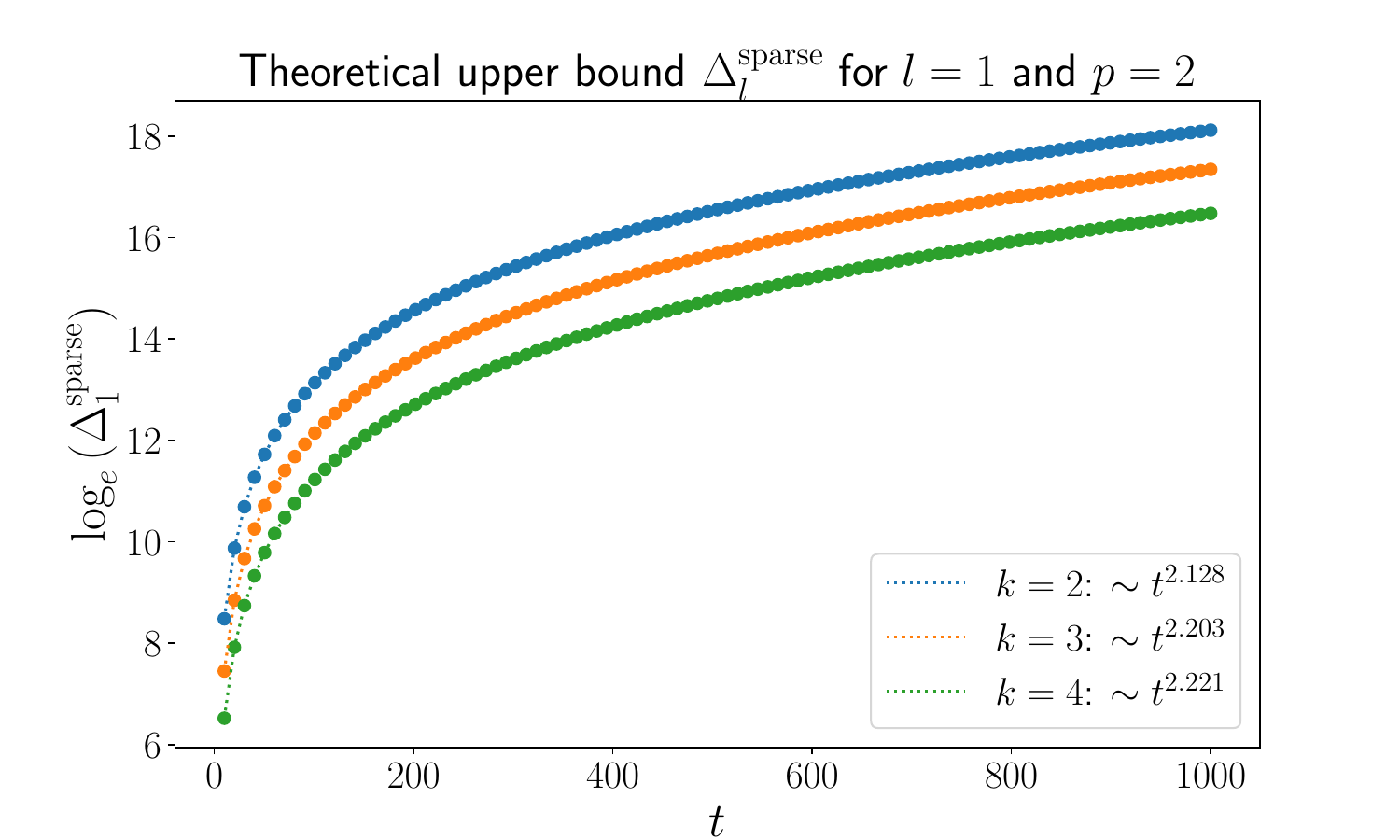}
    \caption{First-order error bound $\Delta_1^{\mathrm{sparse}}$.}
    \label{fig:t_sparse_theory1}
  \end{subfigure}
  \\
  \begin{subfigure}[b]{0.49\textwidth}
    \includegraphics[width=\textwidth]{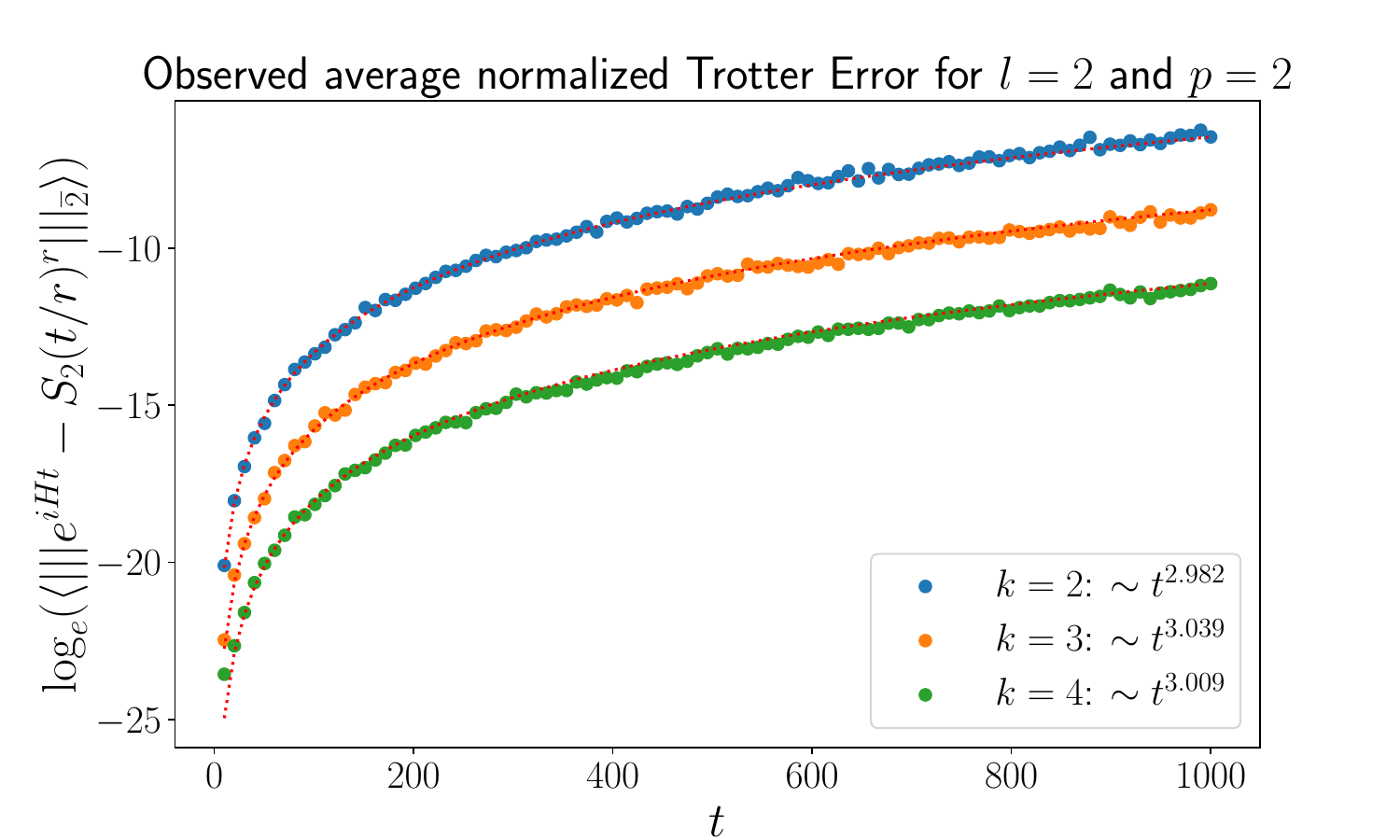}
    \caption{Observed average second-order Trotter error.}
    \label{fig:t_sparse_obs2}
  \end{subfigure}
  \hfill
  \begin{subfigure}[b]{0.49\textwidth}
    \includegraphics[width=\textwidth]{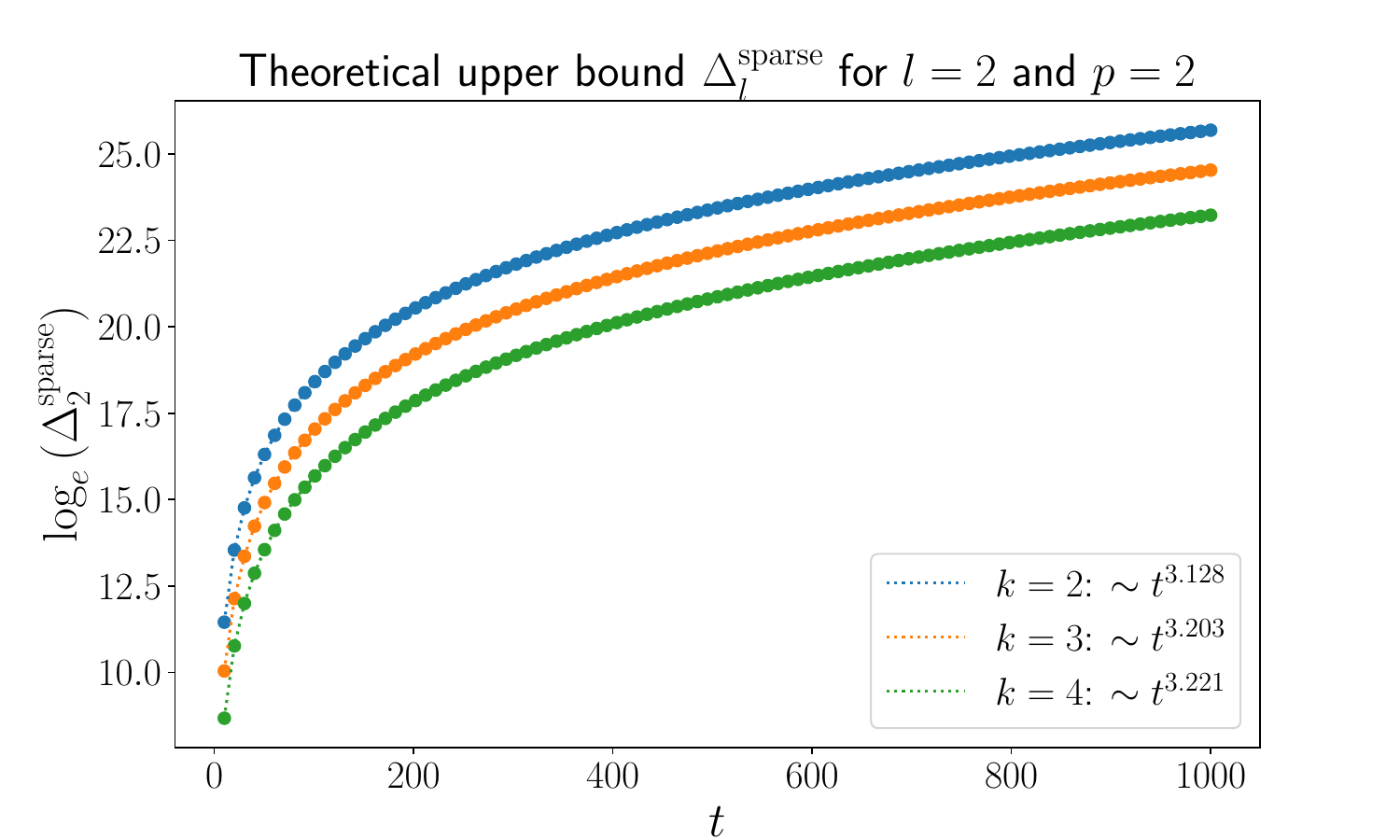}
    \caption{Second-order error bound $\Delta_2^{\mathrm{sparse}}$.}
    \label{fig:t_sparse_theory2}
  \end{subfigure}
  \caption{Observed average normalized Trotter error and the theoretical upper bounds $\Delta_l^{\mathrm{sparse}}$ for $l=1,2$. We fix $n=10$, and study the sparse SYK models of localities $k=2,3,4$ with $t\in [0,1000].$}
  \label{fig:t_sparse}
\end{figure}
\break

\subsection{Structure of the paper}

The paper is organized as follows. In Section~\ref{sec:Prelim}, we establish notation and introduce several mathematical tools to study the Trotter error and the gate complexity of the associated simulation. After that, we first tackle the regular SYK model -- we study the first-order and higher-order Trotter errors and their gate complexity in Sections~\ref{sec:FOE} and~\ref{sec:HOSYK} respectively. In Section~\ref{sec:SSYK}, we study the higher-order Trotter error and gate complexity of the sparse SYK model. We conclude with a summary and open questions in Section~\ref{sec:Discussion}.

\subsection{Acknowledgements}
We thank Chi-Fang Chen for explaining the proof methods of Ref.~\cite{Anthony}.
JH acknowledges funding from the Dutch Research Council (NWO) through a Veni grant (No.VI.Veni.222.331) and from the Quantum Software Consortium (NWO Zwaartekracht Grant No.024.003.037).
MO was supported by an NWO Vidi grant (No.VI.Vidi.192.109) and by a National Growth Fund grant (NGF.1623.23.025).
This work was primarily performed in the context of YC's bachelor thesis at the University of Amsterdam.

\section{Preliminaries}
\label{sec:Prelim}
In this section, we formally describe the problem, provide sufficient mathematical background and introduce some key definitions and notation that will be used throughout. 

\subsection{The SYK and sparse SYK models}

The \emph{SYK model} is defined as a random $k$-local Majorana fermionic Hamiltonian\footnote{For qubit Hamiltonians the underlying $n$-qubit Hilbert space $(\C^2)^{\otimes n}$ has tensor product structure, and a term is ``$k$-local'' if it acts as identity on all but $k$ qubits. For fermionic Hamiltonians the underlying Hilbert space need not have such structure, and the notion of locality is different -- it refers to the number of Majoranas involved in a given term.}
\begin{equation}
\label{eq:SYKdef}
    H_{\text{SYK}} = i^{k(k-1)/2}\sum_{1\leq i_1<\cdots<i_k\leq n}J_{i_1,\dots,i_k}\chi_{i_1}\cdots\chi_{i_k},
\end{equation}
consisting of (even) $n$ Majorana fermions. The coefficients $J_{i_1,\dots,i_k} \in \R$ are i.i.d.\@ Gaussian variables such that 
\begin{equation}
    \label{eq:mean and variance of J}
    \E(J_{i_1,\dots,i_k}) = 0, \quad \text{and} \quad
    \E(J_{i_1,\dots,i_k}^2) = \frac{(k-1)!\mathcal{J}^2}{k n^{k-1}} =: \sigma^2
\end{equation}
with a certain energy constant $\mathcal{J}< \infty$.
The $\chi_i$'s are the Majorana fermions that obey the following anti-commutation relation:
\begin{equation}
    \label{eq:Majoranas}
    \{\chi_i,\chi_j\} = 2\delta_{ij}.
\end{equation}

The \emph{sparse SYK model} is a modified version of the SYK model proposed by Xu, Susskind, Su, and Swingle \cite{susskind}, that requires fewer resources to simulate, but still captures several essential properties \cite{susskind,preskill}. The model is obtained by removing each local interaction in the SYK model with probability $1-p_B$, where $p_B$ is usually defined as 
\begin{equation}
    p_B = \frac{\kappa n}{\binom{n}{k}}
\end{equation}
with a constant $\kappa \geq 0$ (not to be confused with locality $k$) controlling the sparsity of the model.  According to \cite{susskind, preskill}, gravitational physics emerges for $\kappa$ between $\frac{1}{4}$ and 4 at low temperatures.

Removing the local terms in the SYK model is equivalent to attaching a Bernoulli variable $B_{i_1,\dotsc, i_k} \in \set{0,1}$ to each term such that $\Prob(B_{i_1,\dotsc, i_k} = 1) = p_B$. Hence, the sparse SYK model can be written as
\begin{equation}
    H_{\mathrm{SSYK}} = i^{k(k-1)/2}\sum_{1\leq i_1 < \cdots < i_k \leq n} B_{i_1,\dotsc, i_k} J_{i_1,\dotsc, i_k} \chi_{i_1} \cdots \chi_{i_k}.
\end{equation}
To ensure that this model still has extensive energy, we renormalize the variance of the Gaussian couplings $J_{i_1,\dots,i_k}$ as
\begin{equation}
    \sigma^2 \longmapsto \frac{\sigma^2}{p_B},
\end{equation}
to keep the model extensive. 

\subsection{Representation and norm of Majoranas}\label{sec:rep and norm}

Evaluating the norm of the Trotter error of the SYK model inevitably requires evaluating the norm of the Majorana fermions $\chi_i$ appearing in equation~\eqref{eq:SYKdef}. For even $n$, the Majoranas $\chi_1,\dotsc,\chi_{n}$ can be understood as generators of the Clifford algebra $\mathrm{Cl}(\R^{n/2}).$ However, $\mathrm{Cl}(\R^{n/2})$ is inherently not a normed space. Hence, when computing the norm of any element of this algebra, we implicitly compute the norm of its representation in $\mathrm{GL}(2^{n/2},\C).$ A concrete representation we have in mind here is the Jordan--Wigner transformation\footnote{One can choose another representation, as long as it satisfies the two properties in equation \eqref{eq:rho}.} \cite{JW}
\begin{equation}
    \rho_{\mathrm{JW}}: \mathrm{Cl}(\R^{n/2}) \longrightarrow \mathrm{GL}(2^{n/2},\C).
\end{equation}
It is a representation of $\mathrm{Cl}(\R^{n/2})$ because it preserves anti-commutativity and normalization of the generators $\chi_i$:
\begin{align}
    \label{eq:rho}
    \rho(\{\chi_i,\chi_j\}) &= \{\rho(\chi_i),\rho(\chi_j)\}, &
    \rho(\chi_i)\rho(\chi_i) &= I_{2^{n/2}}
\end{align}
for all $1\leq i,j\leq n$. From now on, when we write the norm $\norm{\chi}$ for any $\chi \in \mathrm{Cl}(\R^{n/2})$, we implicitly mean the norm $\norm{\rho_{\mathrm{JW}}(\chi)}$ of its Jordan--Wigner representation. In other words, we have introduced a norm on $\mathrm{Cl}(\R^{n/2})$ through a norm on $\mathrm{GL}(2^{n/2},\C).$

\subsection{Notation for local Hamiltonians}
The summation indices of the SYK and sparse SYK models can be cumbersome when computing product formula approximations of these models. Hence, in this section we introduce an alternative way of indexing terms in a local Hamiltonian by assigning an arbitrary ordering to hyperedges so that we can index terms by integers instead of subsets.

Let $H$ be a $k$-local Hamiltonian acting on $n$ particles labeled by $V=\{1,\dotsc,n\}$. We can write $H$ as 
\begin{equation}
    H = \sum_{\eta\in E} H_{\eta},
\end{equation}
where $E \subseteq 2^V$ is some set of size-$k$ subsets of $V$, and each $H_{\eta}$ is a local interaction on particles $\eta \subseteq V$.
This notation naturally induces a hypergraph $(V,E)$ representing the local terms of $H.$ The particles $V$ are the vertices of the hypergraph, while $E$ are the hyperedges. Each hyperedge denotes the particles on which the corresponding local interaction acts on. As we are using the hyperedges as labels for the local terms in the Hamiltonian, we denote the total number of local terms as $\Gamma := |E|.$

Throughout this paper, we will encounter situations where we need to partially sum over hyperedges in $E$, or sum over pairs of hyperedges in certain ordering. In those cases, the above notation can be cumbersome. To get a clearer overview for bookkeeping, we choose an arbitrary bijection 
\begin{equation}
    \phi: \{1,\dotsc,\Gamma\}\to E
\end{equation}
to identify each hyperedge with a number (the explicit definition of $\phi$ does not matter). Observe that $\phi$ induces an ordering on $E$, such that for $\eta_1,\eta_2\in E$,
\begin{equation}
    \eta_1 > \eta_2 \iff \phi^{-1}(\eta_1) > \phi^{-1}(\eta_2).
\end{equation}

In certain situations, we need to repeatedly sum over $E$. Hence, we can periodically extend $\phi$ to $\N$ by considering the map
\begin{equation}
\label{eq:labeling}
    \gamma:\N \to E,
    \qquad
    \gamma(r + q\cdot \Gamma) := \phi(r)
\end{equation}
with $1\leq r \leq \Gamma$ and $q\in \Z_{\geq 0}$. We will refer to this map $\gamma$ as the \emph{ordering map}.
For example, if $f: E^4 \to \C$ and $g:\N\to \N$, the ordering map allows us to write down sums such as 
\begin{equation}
    \sum_{i = 1}^{3 \Gamma} \sum_{j = i+1}^{4 \Gamma}\sum_{\substack{k,l=1\\k + l \leq \Gamma}}^{\Gamma} f(H_{\gamma(g(i))}, H_{\gamma(j)}, H_{\gamma(g(k))}, H_{\gamma(l)}),
\end{equation}
which would be quite cumbersome if we used summation over hyperedges. More importantly, we can write the SYK model in this notation as
\begin{equation}
    \label{eq:HSYK}
    H_{\mathrm{SYK}} = \sum_{i=1}^{\Gamma} H_{\gamma(i)},
\end{equation}
where $\gamma$ is any ordering map of our choice.  The hyperedge $\gamma(i) = \{\gamma(i)_1,\dotsc, \gamma(i)_k\} \in E$ denotes the Majorana fermions acted upon by the local term
\begin{equation}
    H_{\gamma(i)} = i^{k(k-1)/2}J_{\gamma(i)} \left(\chi_{\gamma(i)_1}\cdots \chi_{\gamma(i)_k}\right)_{\text{lex}},
    \label{eq:H gamma i}
\end{equation}
where the product of Majorana fermions is ordered lexicographically to match the ordering of indices in the original definition of the SYK model in equation~\eqref{eq:SYKdef}.

\subsection{Trotter error and product formulas}
\label{subsec:TEPFA}
Simulating the SYK model on a quantum computer means implementing its time-evolution operator $e^{itH_{\mathrm{SYK}}}$ for any choice of the random couplings.
In general, it is unclear how to exactly implement this matrix exponential, but there are several efficient approximation methods, one of them being product formulas. In essence, product formulas allow us to decompose an exponential of a sum of matrices into products of exponentials. Suppose $H$ is a local qubit Hamiltonian defined as 
\begin{equation}
    H = \sum_{i=1}^{\Gamma}H_{\gamma(i)},
\end{equation}
with a certain ordering map $\gamma$ defined in equation~\eqref{eq:labeling}.
Then we can decompose the unitary time evolution as a product plus some error
\begin{equation}
    e^{iHt} = e^{iH_{\gamma(1)}t}e^{iH_{\gamma(2)}t}\cdots e^{iH_{\gamma(\Gamma)}t} + \mathcal{O}(t^2).
\end{equation}
The product 
\begin{equation}
    S_1(t) = \prod_{i=1}^{\Gamma} e^{itH_{\gamma(i)}}
\end{equation}
is called the \emph{first-order product formula}. The error term $\mathcal{O}(t^2)$ associated with this formula is called the \emph{first-order Trotter error}.

We can also write down product formulas of higher orders -- the most prominent family of examples being the \emph{Lie--Trotter--Suzuki formulas} \cite{Suzuki}. The \emph{second-order formula} is defined as 
\begin{equation}
    S_2(t) = \prod_{i = \Gamma}^1 e^{i H_{\gamma(i)} t/2} \prod_{j = 1}^{\Gamma}e^{i H_{\gamma(j)} t/2},
\end{equation}
and \emph{higher-order formulas} are defined recursively:
\begin{equation}
\label{eq:recursive_def}
    S_{2p}(t) = S_{2p - 2}(q_pt)^2\cdot S_{2p-2}((1-4q_p)t)\cdot S_{2p-2}(q_pt)^2
\end{equation}
with $q_p = 1/ (4 - 4^{1/(2p-1)})$. This recursive structure only defines product formulas of even order. Therefore, when we refer to higher-order product formulas in this work, we only consider even orders. In general, an $l$-th order product formula has Trotter error of leading order $\mathcal{O}(t^{l+1}).$ 

Another way to write down the $l$-th order product formula is to factor the recursive equation~\eqref{eq:recursive_def} into a product of $\Upsilon$ \emph{stages}:
\begin{equation}
\label{eq:stages}
    S_l(t) = \prod_{a = 1}^{\Upsilon}\prod_{b=1}^{\Gamma}e^{ic(a,b) H_{\gamma(d(a,b))}t}
\end{equation}
where $c(a,b)\in \R$ and $d(a,b)$ are chosen such that the above equals equation~\eqref{eq:recursive_def}. By the recursive definition of $l$-th order product formulas, the number of stages is $\Upsilon = 2\cdot 5^{\frac{l}{2}-1}$. We can write this even more compactly with $J= \Upsilon \cdot \Gamma$:
\begin{equation}
\label{eq:compact_def}
    S_l(t) = \prod_{j = 1}^{J} e^{ia_j H_{\gamma(b_j)}t }
\end{equation}
where $a_j\in \R$ and $b_j\in \N$ are chosen such that the above expression is the same as equation~\eqref{eq:recursive_def}.

In general, one uses $S_l(t/r)^r$ to approximate the Hamiltonian, instead of just $S_l(t).$ We refer to the \emph{number of rounds} $r$ of applying the product formulas as \emph{Trotter number} for short. For further reading on this topic, we refer the reader to \cite{Anthony, Childs, Suzuki}.

\subsection{Probability statements and concentration inequalities}
Because the SYK and sparse SYK models are random Hamiltonians, there are many ways to characterize what sufficiently small Trotter error means. The first one, appearing in \cite{PhysRevD.109.105002, susskind, Garc_a_lvarez_2017}, is the average case statement
\begin{equation}
\label{eq: averagestatement}
    \E\left(\norm{e^{iHt} - S_l(t/r)^r}\right) < \epsilon,
\end{equation}
where the expectation is over the randomness in $H$. Because this formulation of small Trotter error only covers the average case, we use the following two more precise statements given by \cite{Anthony}. The first one is stated using the operator norm:
\begin{equation}
\label{eq:concineq}
    \Prob\left(\norm{e^{iHt} - S_l(t/r)^r} \geq \epsilon \right) < \delta.
\end{equation}
The other one quantifies small Trotter error by looking at the difference between the actual time-evolution of an arbitrary fixed input state $\ket{\psi}$ and its approximated time evolution in $l^2$-norm:
\begin{equation}
\label{eq:concineqpsi}
    \Prob\left(\norm{(e^{iHt} - S_l(t/r)^r)\ket{\psi}}_{l^2} \geq \epsilon \right) < \delta.
\end{equation}

In addition to yielding more precise characterization of small Trotter error, these two probability statements also connect to the following two norms for random-valued operators $A$ on a given finite-dimensional Hilbert space $\Hil$:
\begin{align}
\label{eq: ESN}
    \opnorm{A}_p &:= \of[\Big]{\E\of[\big]{\norm{A}_p^p}}^{\frac{1}{p}}, \\
\label{eq: fixedESN}
    \opnorm{A}_{\fix,p} &:= \sup\left\{ \of[\Big]{\E\of[\big]{\norm{A\ket{\phi}}_{l^2}^p}}^{\frac{1}{p}} : \ket{\phi} \in \Hil, \norm{\ket{\phi}}_{l^2} = 1\right\}.
\end{align}
These norms obey certain uniform smoothness properties, which lead to better scaling in the Trotter error bounds and the corresponding gate complexities. Structuring our proofs around these properties ensures that the results that we derive in terms of $\opnorm{\cdot}_p$ also hold in terms of $\opnorm{\cdot}_{\fix,p}$. Hence, following the notation in \cite{Anthony}, we will introduce the notation $\opnorm{\cdot}_{*,p}$\footnote{The notation used in \cite{Anthony} is $\opnorm{}_{*}$ which hides the $p$-dependence. Here, we will denote this dependence explicitly.} as a placeholder for both of these norms so that we can formulate our results more concisely. Whenever we state a result in the form 
\begin{equation}
    \opnorm{A}_{*,p} \leq C\opnorm{B}_{*,p}
\end{equation}
for $A,B$ operators on our Hilbert space and $C$ a positive real constant, what we mean is that $\opnorm{\cdot}_{*,p}$ can be replaced with either $\opnorm{\cdot}_p$ or $\opnorm{\cdot}_{\fix,p}$ throughout. In other words, both
\begin{equation}
     \opnorm{A}_p \leq C\opnorm{B}_p,\quad\text{and}\quad  \opnorm{A}_{\fix,p} \leq C\opnorm{B}_{\fix,p}
\end{equation}
are true. For the scope of this work, we shall not discuss these uniform smoothness properties in detail. We refer interested readers to \cite{Tomczak1974, Ball1994SharpUC, ESPN, Anthony} for more details.

To evaluate the probability statements in equations~\eqref{eq:concineq} and~\eqref{eq:concineqpsi}, we use the following concentration inequalities:

\begin{restatable}[Concentration inequality for operator norm]{lemma}{lemmaconcineq}
\label{lem:concineq}
    Let $H$ be a Hamiltonian on a Hilbert space. Let $l\in \{1\} \cup 2\N$ be the order of the product formula approximating the Hamiltonian. 
    If the Trotter error of this product formula is bounded by a positively valued function $\lambda(p,r)$, in the sense that
    \begin{equation}
    \label{eq: actual_concineq}
        \opnorm{e^{iHt} - S_l(t/r)^r}_p \leq p\lambda(p, r) \opnorm{I}_p
    \end{equation}
    for all $p\geq 2$, then choosing the Trotter number $r$ so that it obeys the inequality
    \begin{equation}
    \label{eq:solveforr}
        \frac{\lambda\left(\log(e^2 \opnorm{I}_p^p/\delta),r\right)}{\epsilon} \leq \frac{1}{e\cdot \log(e^2 \opnorm{I}_p^p/\delta)}
    \end{equation}
   ensures that
   \begin{equation}
       \Prob\left(\norm{e^{iHt} - S_l(t/r)^r} \geq \epsilon \right) < \delta.
   \end{equation}
\end{restatable}

\begin{restatable}[Concentration inequality for fixed input state]{lemma}{lemmaconcineqpsi}
\label{lem:concineqpsi}
    Let $H$ be a Hamiltonian on a Hilbert space of dimension $D$ and $\ket{\psi}$ an arbitrary fixed input state in this space. Let $l\in \{1\} \cup 2\N$ be the order of the product formula approximating the Hamiltonian.
    If the Trotter error of this product formula is bounded by a positively valued function $\lambda(p,r)$, in the sense that
    \begin{equation}
    \label{eq: actual_concineqpsi}
        \opnorm{e^{iHt} - S_l(t/r)^r}_{\fix,p} \leq p\lambda(p, r) 
    \end{equation}
    for all $p\geq 2$, then choosing the Trotter number $r$ so that it obeys the inequality
    \begin{equation}
    \label{eq:solveforrpsi}
        \frac{\lambda\left(\log(e^2/\delta),r\right)}{\epsilon} \leq \frac{1}{e\cdot \log(e^2/\delta)}
    \end{equation}
   ensures that
   \begin{equation}
       \Prob\left(\norm{(e^{iHt} - S_l(t/r)^r)\ket{\psi}}_{l^2} \geq \epsilon \right) < \delta.
   \end{equation}
\end{restatable}

The proof of these lemmas can be found in Appendix~\ref{apx:concineq}. We will refer to equations~\eqref{eq: actual_concineq} and~\eqref{eq: actual_concineqpsi} as concentration inequalities. The positively-valued function $\lambda(p,r)$ typically takes the form of a positive rational power function of $p$ and a negative rational power function of $r$, obtained from explicitly evaluating the representations of Trotter error defined in Section~\ref{subsec:Error Representation}. 

\subsection{Gate complexity}\label{sec:gate complexity}

After finding the Trotter number $r$ through the concentration inequalities (depending on the probability statement we choose), we can compute the number of local exponentials in the corresponding approximation $S_l(t/r)^r$ as
\begin{equation}
    C = \Upsilon\Gamma r,
\end{equation}
where $\Upsilon$ is the number of stages and $\Gamma$ the number of local terms, see equation~\eqref{eq:stages}.
However, $C$ is not the actual gate complexity\footnote{We defined gate complexity as the number of gates required to simulate the Hamiltonian such that the Trotter error is sufficiently small, as characterized by either equation~\eqref{eq:concineq} or~\eqref{eq:concineqpsi}.} of simulating a fermionic Hamiltonian.  As we mentioned earlier, mapping fermions to qubits inevitably introduces extra overhead for the implementation of each local exponential. For instance, one could consider the standard Jordan--Wigner transform that gives each local exponential an $\mathcal{O}(n)$ overhead. In that case, the actual gate complexity should scale with 
\begin{equation}
    G = \mathcal{O}(n\cdot \Upsilon \Gamma r).
\end{equation}

Fortunately, using the ternary tree mapping \cite{Jiang_2020} incurs only a $\log(n)$ overhead for each local exponential, giving a total gate complexity that scales with
\begin{equation}
    G = \mathcal{O}(\log(n)\cdot \Upsilon \Gamma r).
\end{equation}
Because $\log(n)$ is a slow-growing function, we opt to omit it when writing down the gate complexity. To distinguish from the actual gate complexity, our result will be reported as 
\begin{equation}
    G = \tilde{\mathcal{O}}(\Upsilon \Gamma r),
\end{equation}
but the reader should always keep in mind the overhead from fermion-to-qubit mapping. 
 
\subsection{Error representation}
\label{subsec:Error Representation}

We will need a more nuanced error representation to obtain better bounds on the Trotter error $\norm{e^{iHt} - S_l(t/r)^r}$.
Instead of measuring the error by a single number, we will consider an operator $\mathcal{E}(t)$ that captures the deviation from the exact evolution $e^{iHt}$ at any given time $t$.
More specifically, we will use the exponential form of the error from \cite{Childs, Anthony} to capture the deviation of the product formula~\eqref{eq:compact_def} at any given $t$:
\begin{equation}
    \label{eq:expT}
    S_l(t) = \prod_{j=1}^J e^{ia_jH_{\gamma(b_j)}t} = \mathrm{exp}_{\mathcal{T}}\left(i\int(\mathcal{E}(t) + H)dt\right),
\end{equation}
where the right-hand side involves an ordered exponential. Hence, we can think of every product formula as the time evolution of our Hamiltonian $H$ with some time-dependent error $\mathcal{E}(t)$ added to it. We call $\mathcal{E}(t)$ the \emph{error generator}, and it takes the following form \cite{Anthony, Childs}:
\begin{equation}
    \label{eq:Et def}
    \mathcal{E}(t) := \sum_{j=1}^J\left(\prod_{k=j+1}^J e^{a_k\Lag_{\gamma(b_k)}t}[a_jH_{\gamma(b_j)}] - a_jH_{\gamma(b_j)}\right),
\end{equation}
where $\Lag_{\gamma(j)}$ sends an operator $O$ to the commutator with $H_{\gamma(j)}$:
\begin{equation}
    \label{eq:Lj}
    \Lag_{\gamma(j)}[O] := i[H_{\gamma(j)}, O].
\end{equation}
This error generator will be our primary target when deriving bounds for the Trotter error.

\subsection{Random matrix polynomials}
\label{subsec:randommatpol}

A very useful method for the study of higher-order error generator is to view it as \emph{random matrix polynomials} of local terms in the Hamiltonian. Hence, we introduce some tools for treating random matrix polynomials in this subsection.

We define a degree-$g$ Gaussian random matrix polynomial of $m$ matrix variables $X_{1},\dotsc,X_{m}$ as 
\begin{equation}
\label{eq: def_F}
    F(X_{m},\dotsc,X_{1}) = \sum_{\bm{i}}T_{\bm{i}} X_{i_1}\cdots X_{i_g}
\end{equation}
where $\sum_{\bm{i}}$  is the abbreviated notation for the summation 
\begin{equation}
    \sum_{\bm{i}} := \sum_{i_1=1}^m\dotsc \sum_{i_g = 1}^m ,
\end{equation}
and $T_{\bm{i}} := T_{(i_1,\dotsc, i_g)}$ are scalar coefficients. For notation here, we will consistently use bold text to denote vectors. In this case $\bm{j}$ is the vector of the summation indices $(i_1,\dotsc, i_g).$ For simplicity, we pose that the variables $X_{1},\dotsc, X_{m}$ are independent from each other, and each $X_i$ can be written as 
\begin{equation}
    X_i = g_i K_i
\end{equation}
with $g_i$ a standard Gaussian variable and $K_i$ a deterministic matrix with bounded norm $\norm{K_i} \leq \sigma_i$. We refer to the set that collects the subscripts of the random variables $F$ depends on as the index set $\mathcal{I}$. In this case, we have
\begin{equation}
    \mathcal{I} := \{1,2,3,\dotsc, m\} = [m].
\end{equation}

To study the concentration bounds on these polynomials, \cite[Proposition VII.2.1]{Anthony} proposed the following (sum of squares) uniform smoothing inequality:
\begin{theo}
\label{thm:matpol}
Let $F(X_{i_m},\dotsc,X_{i_1})$ be a random matrix polynomial, with index set $\mathcal{I}=\{i_1,\dotsc,i_m\}$ and each $X_i$ is independent. Then, for $p\geq 2$ and $C_p = p-1$,
    \begin{equation}
        \opnorm{F}_p^2 \leq \sum_{S\subseteq \mathcal{I}}C_p^{|S|}\opnorm{[F]_S}_p^2
    \end{equation}
    where
    \begin{equation}
        [F]_S = \prod_{s\in S}(1 - \mathbb{E}_s)\prod_{s'\in S^c}\mathbb{E}_{s'}[F],
    \end{equation}
    with $\mathbb{E}_s$ being the entry-wise expectation applied only on $X_s.$ The notation $S^c$ is the complement $\mathcal{I}\backslash S.$
\end{theo}
We refer to $[\cdots]_S$ as the smoothing operator. This operator has a central part given by 
\begin{equation}
    (1-\mathbb{E})_{S}[\;\cdot\;] := \prod_{s\in S}(1 - \mathbb{E}_s)[\;\cdot\;]
\end{equation}
and an outer part given by 
\begin{equation}
    \mathbb{E}_{S^c}[\;\cdot\;] := \prod_{s'\in S^c}\mathbb{E}_{s'}[\;\cdot\;].
\end{equation}
Using our notation, the smoothing operator could be written as 
\begin{equation}
    [\;\cdot\;]_S = (1-\mathbb{E})_S\circ \mathbb{E}_{S^c}[\;\cdot\;].
\end{equation}

We would like to apply Theorem \ref{thm:matpol} to the random matrix polynomial $F$ defined in equation \eqref{eq: def_F}. This requires us to identify, for each subset $S \subseteq \mathcal{I}$, all terms $T_{\bm{i}} X_{i_1}\cdots X_{i_g}$ in the polynomial $F$ such that the corresponding smoothing operator is non-vanishing, i.e.
\begin{equation}
    \left[T_{\bm{i}} X_{i_1}\cdots X_{i_g}\right]_S \neq 0.
\end{equation}
Notice that the Gaussian matrices $X_{i}$ can appear with powers up to $g$ in $F$, which makes them difficult to handle when applying the smoothing operator. To address this, we first expand each $X_i$ into a product of Rademacher variables.

\subsubsection{Rademacher expansion}
Recall that the central limit theorem tells us that a standard Gaussian variable $g_i$ can be written as the limit of a sum of i.i.d.\@ Rademacher variables:
\begin{equation}
    g_i = \lim_{N\to \infty} \sum_{j = 1}^N \frac{\epsilon_{i,j}}{\sqrt{N}},
\end{equation}
where $\epsilon_{i,j}\in \{\pm 1\}$ are Rademacher variables. Taking $N$ to be a sufficiently large integer motivates us to express each $X_i$ as a sum of Rademacher variables $Y_{i,j}$:
\begin{equation}
    X_i \approx \sum_{j=1}^N \frac{\epsilon_{i,j}}{\sqrt{N}}K_i =: \sum_{j=1}^N Y_{i,j}.
\end{equation}
Under this expansion, we regard $F$ as a polynomial in the variables
\begin{equation}
\begin{aligned}
   F &=  F(Y_{1,1}, Y_{1,2},\dotsc,Y_{1,N},\dotsc, Y_{m,1}, Y_{m,2},\dotsc, Y_{m,N})\\
   &:= \sum_{\bm{i}}\sum_{\bm{j}} T_{(i_1,\dotsc,i_g)}Y_{i_1,j_1} Y_{i_2,j_2}\dotsc Y_{i_g,j_g}.
\end{aligned}
\end{equation}
Although this representation makes $F$ appear more complicated, it offers a key simplification: the non-subleading terms of $F$ in the large-$N$ limit contain Rademacher variables with exponents only in $\{0,1,2\}$, rather than arbitrary powers up to $g$. This significantly simplifies the application of the smoothing operator. To illustrate this subleading argument, we consider the following example.

For the first non-trivial case $m=3$ (the analysis for $m>3$ is identical), we have
\begin{equation}
    \sum_{j_1 = 1}^N \sum_{j_2 = 1}^N \sum_{j_3 = 1}^N Y_{1, j_1}Y_{1,j_2}Y_{1,j_3}. 
\end{equation}
How many terms of the form $Y_{1, j_i}^3$ are present in this sum? The answer is $N$, corresponding to $j_1 = j_2 = j_3$. Summing over these terms gives
\begin{equation}
    \sum_{j_1 = 1}^N Y_{1,j_1}^3 = \sum_{j_1 = 1}^N \frac{\epsilon_{1,j_1}^3}{N^{3/2}} K_1^3.
\end{equation}
Since the Rademacher variable $\epsilon^3_{1,j_1}$ only takes values $\pm 1$, we obtain the inequality
\begin{equation}
\bigg|\sum_{j_1 = 1}^N \frac{\epsilon_{1,j_1}^3}{N^{3/2}}\bigg| \leq \sum_{j_1 = 1}^N \frac{1}{N^{3/2}}.
\end{equation}
Simplifying yields
\begin{equation}
\bigg|  \sum_{j_1 = 1}^N \frac{\epsilon_{1,j_1}^3}{N^{3/2}}\bigg| \leq \frac{1}{N^{1/2}}.
\end{equation}
Clearly, this expression vanishes as $N\to \infty$. Hence, each Rademacher variable in a non-subleading term of $F$, say $Y_{i,j_i}$, may only have exponent $0$, $1$, or $2$.

\subsubsection{Resummation by subsets}
While this provides useful restrictions on the exponents of the Rademacher variables, it also introduces some complications: the index set of $F$ becomes $\mathcal{I} = [m]\times [N]$, rather than simply $[m]$, because we now regard $F$ as a polynomial depending on the variables $Y_{1,1},\dotsc, Y_{m,N}$. Consequently, the subset $S$ in the smoothing operator $[\cdots]_S$ must now be taken from the Cartesian product $[m]\times [N]$. Given a subset $S\subseteq [m]\times [N]$, we can rewrite $F$ as follows:
\begin{equation}
\begin{aligned}
    F =\sum_{\substack{x + y = g\\ x,y \in \{0,\dotsc, g\}}} \sum_{s_1,\dotsc, s_x \in S}\sum_{s'_1,\dotsc,s'_y \in S^c} \sum_{\pi \in \mathcal{S}_g\big/\mathcal{S}_x\times \mathcal{S}_y}T_{\pi\left[\bm{\pr}_{1}  (s_1,\dotsc, s_x, s'_1,\dotsc, s'_y)\right]}\pi\left[Y_{s_1}\cdots Y_{s_x} Y_{s'_1}\dotsc Y_{s'_y}\right].
\end{aligned}
\end{equation}

Let us unpack the notation. Each term in $F$ takes the form 
\begin{equation}
    \pi\left[Y_{s_1}\cdots Y_{s_x} Y_{s'_1}\dotsc Y_{s'_y}\right],
\end{equation}
where $s_1,\dotsc,s_x \in S$ and $s'_1,\dotsc, s'_y \in S^c$. The sums over $x$ and $y$ specify how many terms in the product $Y_{s_1}\cdots Y_{s_x} Y_{s'_1}\dotsc Y_{s'_y}$ have subscripts in $S$ and $S^c$, respectively. The permutation $\pi \in \mathcal{S}_g$ reorders the positions of the Rademacher variables in the product to account for non-commutativity of the random variables. To avoid double counting, we exclude permutations that only reorder $Y_{s_1},\dotsc, Y_{s_x}$ among themselves and/or $Y_{s'_1},\dotsc, Y_{s'_y}$ among themselves, since such symmetries are already present in the summations over $s_1,\dotsc, s_x$ and $s'_1,\dotsc, s'_y$. Thus, we sum over the quotient set $\mathcal{S}_g\big/\mathcal{S}_x\times \mathcal{S}_y$, which removes redundant symmetries from $\mathcal{S}_g$. We also extend the definition of the coefficients $T$ so that they vanish whenever their subscript $\pi\left[\bm{\pr}_{1} (s_1,\dotsc, s_x, s'_1,\dotsc, s'_y)\right]$ does not correspond to a nonzero term in $F$.

To match the subscripts of coefficients $T_{(\cdots)}$, we introduce the projection from a direct product to the first set $\pr_{1}: A\times B \to A$ defined by 
\begin{equation}
    \pr_1(i,j) = i \in A.
\end{equation}
The bold symbol $\bm{\pr}_{1}$ denotes a vector of projectors $(\pr_1, \dotsc, \pr_1)$ with $g$ entries. Hence,
\begin{equation}
   \bm{\pr}_{1}(s_1,\dotsc, s_x, s'_1,\dotsc, s'_y) = \left(\pr_1s_1,\dotsc, \pr_1s_x,\pr_1s'_1,\dotsc, \pr_1s'_y\right).
\end{equation}
Because we have separated the subscripts in parts that belong to $S$ and parts that belong to $S^c$, the action of the smoothing operator is simply
\begin{equation}
\label{eq: big_sum}
    [F]_S  = \sum_{x + y = g} \sum_{s_1,\dotsc, s_x \in S}\sum_{s'_1,\dotsc,s'_y \in S^c} \sum_{\pi \in \mathcal{S}_g\big/\mathcal{S}_x\times \mathcal{S}_y}T_{\pi\left[\bm{\pr}_{1}  (s_1,\dotsc, s_x, s'_1,\dotsc, s'_y)\right]}\pi\left[(1-\mathbb{E})_S\Big[Y_{s_1}\cdots Y_{s_x}\Big] \mathbb{E}_{S^c}\Big[Y_{s'_1}\dotsc Y_{s'_y}\Big]\right]
\end{equation}
Evaluating the above summations and combining it with Theorem \ref{thm:matpol}, we get the following result
\begin{restatable}{theo}{thmRPG}
\label{thm:RPG}
    For $p\geq 2$ and $C_p =p-1$, 
    \begin{equation}
        \opnorm{F}_{*,p}^2 \leq g^{g}C_p^g\sum_{\bm{w}}
        \Bigg( \sum_{\bm{v}}\sum_{\pi \in \mathcal{S}_g}\left|T_{\pi\left[\bm{\eta}(\bm{w} , 2\bm{v})\right]}\right|
        \prod_{i=1}^m \sigma_i^{w_i + 2v_i}\Bigg)^2\opnorm{I}_{*,p}^2,
    \end{equation}
    where $\bm{w},\bm{v}\in \{0,\dotsc,g\}^m$ are such that $|\bm{w}|+2|\bm{v}|=g$.\footnote{If $\bm{u} = (u_1, \dotsc, u_m) \in \Z_{\geq 0}^m$ is a vector then $|\bm{u}| := \sum_{i=1}^m |u_i|$.} The unary encoding $\bm{\eta}$ is defined as 
    \begin{equation}
    \label{eq:eta}
        \bm{\eta}(\bm{a}, \bm{b}) := (\underbrace{1,\dotsc,1}_{a_1\text{ times}}, \underbrace{2,\dotsc,2}_{a_2\text{ times}}, \dotsc, \underbrace{m,\dotsc,m}_{a_m\text{ times}},\underbrace{1,\dotsc,1}_{b_1\text{ times}}, \underbrace{2,\dotsc,2}_{b_2\text{ times}}, \dotsc, \underbrace{m,\dotsc,m}_{b_m\text{ times}})
    \end{equation}
    for $\bm{a},\bm{b}\in \{0,\dotsc,g\}^m.$
\end{restatable}

This is a significant result originally from the work of \cite[Theorem VII.3]{Anthony} with a slight variation of prefactors compared to their statement. However, we prove this result via a slightly different route by considering the resummation in equation \eqref{eq: big_sum}. For technical details, we refer the reader to Appendix~\ref{sec:RMPGC}. Intuitively, Theorem \ref{thm:RPG} upper bounds the square of the norm of a Gaussian random matrix polynomial $F$ by a sum of squares of its individual terms given by
\begin{equation}
   \sum_{\pi\in \mathcal{S}_g}\sum_{ \bm{v}: |\bm{w}|+2|\bm{v}|=g}\left|T_{\pi[\bm{\eta}(\bm{w},2\bm{v})]}\right| \prod_{i=1}^m \sigma_i^{w_i + 2v_i}.
\end{equation}
Therefore, this theorem provides us a square root improvement compared to triangle inequality for upper bounding the norm of a random matrix polynomial. Since the SYK model is Gaussian random, we will use this result to evaluate the expected Schatten norm of higher-order Trotter errors as we are going to see in the next section. 

\subsubsection{Remark for notations}
Most of our notation follows closely to that of \cite{Anthony}. In particular, the norms and error representations follow their conventions. Since Theorem \ref{thm:RPG} is a slight variant of Theorem VII.3 in \cite{Anthony}, we follow the same notation for the coefficients $T_{\pi(\cdots)}$ and the vector summations over $\bm{w},\bm{v}$ take over the role of $\bm{u},\bm{v}$ in the original statement. 

We also use our own notation apart from writing down local Hamiltonians. For instance, we use $[N]$ to denote the set $\{1,\dotsc, N\}$ for any positive integer $N$. Unless explicitly stated otherwise, all the vectors inside summations should be interpreted as non-negatively valued integer vectors. We shall introduce more notations as we progress through the following chapters. 

\section{First-order Trotter error of the SYK model}
\label{sec:FOE}
In this section, we will derive a bound for the first-order Trotter error of the SYK model, and estimate the corresponding gate complexity through the concentration inequality \eqref{eq: actual_concineq}. To derive the error bound, we use the following result in \cite[Lemma VI.2]{Anthony}: for $p\geq 2$ and $C_p = p-1$,
\begin{equation}
\label{eq:fogenerator}
    \opnorm{\mathcal{E}(t)}_{*,p}^2 \leq 2C_p \sum_{i=2}^{\Gamma}\left[4C_pt^2\sum_{j=1}^{i-1}\opnorm{[H_{\gamma(j)}, H_{\gamma(i)}]}_{*,p}^2 
    +\left(\sum_{j=1}^{i-1}\frac{t^2}{2}\opnorm{[H_{\gamma(j)},[H_{\gamma(j)},H_{\gamma(i)}]]}_{*,p}\right)^2\right],
\end{equation}
with an arbitrary choice of ordering map $\gamma$, see equation~\eqref{eq:labeling}. 

Although this result is derived in the context of qubit Hamiltonians, its proof generalizes to fermionic Hamiltonians such as the SYK model, with a slightly modified interpretation of the expected Schatten $p$-norm. Because the Majorana fermions are not normed objects, one could first transform them to qubit operators using, for instance, the Jordan--Wigner transformation. The Schatten norm is then taken over these qubit operators transformed from the Majoranas.

\subsection{Error bound}

Equation \eqref{eq:fogenerator} can be explicitly evaluated because we can directly compute the commutator using the commutation relation \eqref{eq:Majoranas} of Majorana fermions. From the generator, we can derive an error bound on the first-order Trotter error of the SYK model. We state our result below:

\begin{theo}
\label{thm:FOE}
    The first-order Trotter error of the SYK model is bounded by 
    \begin{align}
        \frac{\opnorm{e^{iHt} - S_1(t/r)^r}_{*,p}}{\opnorm{I}_{*,p}}
        &\leq \Delta_1(n,k,t,r,p),
    \end{align}
    where
    \begin{equation}
       \Delta_1(n,k,t,r,p):= 4\sqrt{2} p^2 \sigma^2 \sqrt{\binom{n}{k}Q(n,k)}\;t^2\left[\frac{1}{2r} 
    + \sigma \sqrt{Q(n,k)}\frac{t}{3r^2} \right].
    \end{equation}
    Here $\sigma^2$ is the variance of the coupling coefficients, see equation~\eqref{eq:mean and variance of J}, and $Q(n,k)$ is defined as
    \begin{equation}
        \label{eq:Q(n,k)}
        Q(n,k) := \begin{cases}
        \mathlarger{\sum}\limits_{s \text{ odd},\;s\geq\mu}^{k-1}\binom{n-k}{k-s}\binom{k}{s} & \text{if $k$ is even},\\[15pt]
        \mathlarger{\sum}\limits_{s \text{ even},\;s\geq\mu}^{k-1}\binom{n-k}{k-s}\binom{k}{s} & \text{if $k$ is odd},
        \end{cases}
    \end{equation}
    where $\mu := \max\{0, k - (n - k)\}$ is the minimal overlap of two hyperedges.
\end{theo}

Proving the above result revolves around evaluating the commutators in equation \eqref{eq:fogenerator}. To do that, we use the following anti-commutation lemma for local terms in the SYK model. The proof of this lemma is presented in Appendix~\ref{apx:Majoranas}. The same result is also recorded in \cite{Garc_a_Garc_a_2018}.

\begin{restatable}{lemma}{lemmaanticom}
\label{lemma: anti_com}
Let $H = \sum_{i=1}^{\Gamma} H_{\gamma(i)}$ be the $k$-local SYK model with $\gamma: [\Gamma] \to E$ an arbitrary ordering of hyperedges. Then for all $1\leq i<j\leq \Gamma$,
\begin{equation}
    H_{\gamma(i)}H_{\gamma(j)} = (-1)^{k+m}H_{\gamma(j)}H_{\gamma(i)},
\end{equation}
where $m = |\gamma(i)\cap \gamma(j)|$ is the overlap between hyperedges $\gamma(i)$ and $\gamma(j).$
\end{restatable}

The above lemma leads to the following commutation relations if $k$ is even:
\begin{align}
    [H_{\gamma(i)},H_{\gamma(j)}] &=
    \begin{cases}
        0 & \text{if }|\gamma(i)\cap\gamma(j)| \text{ is even},\\
        2H_{\gamma(i)}H_{\gamma(j)} & \text{else},
    \end{cases} \\
    [H_{\gamma(i)}, [H_{\gamma(i)},H_{\gamma(j)}]] &=
    \begin{cases}
        0 & \text{if }|\gamma(i)\cap\gamma(j)| \text{ is even},\\
        4H_{\gamma(i)}^2H_{\gamma(j)} & \text{else},
    \end{cases}
\end{align}
and if $k$ is odd:
\begin{align}
    [H_{\gamma(i)},H_{\gamma(j)}] &=
    \begin{cases}
        0 & \text{if }|\gamma(i)\cap\gamma(j)| \text{ is odd},\\
        2H_{\gamma(i)}H_{\gamma(j)} & \text{else},
    \end{cases} \\
    [H_{\gamma(i)}, [H_{\gamma(i)},H_{\gamma(j)}]] &=
    \begin{cases}
        0 & \text{if }|\gamma(i)\cap\gamma(j)| \text{ is odd},\\
        4H_{\gamma(i)}^2H_{\gamma(j)} & \text{else}.
    \end{cases}
\end{align}

\begin{proof}[Proof (of Theorem~\ref{thm:FOE})]
Note from the above expressions that if $2\leq i\leq \Gamma$ and $k$ is even then
\begin{align}
    \sum_{j=1}^{i-1}\opnorm{[H_{\gamma(j)}, H_{\gamma(i)}]}_{*,p}^2
     &\leq \sum_{\substack{j=1\\j\neq i}}^{\Gamma} \opnorm{[H_{\gamma(j)}, H_{\gamma(i)}]}_{*,p}^2\\
    &\leq 4M_1^2\left\{\sum_{\substack{j=1\\
    |\gamma(j)\cap \gamma(i)| = 1}}^{\Gamma} 1 +\sum_{\substack{j=1\\
    |\gamma(j)\cap \gamma(i)| = 3}}^{\Gamma} 1 +\cdots + \sum_{\substack{j=1\\
    |\gamma(j)\cap \gamma(i)| = k-1}}^{\Gamma}1\right\},
\end{align}
while if $k$ is odd then
\begin{align}
    \sum_{j=1}^{i-1}\opnorm{[H_{\gamma(j)}, H_{\gamma(i)}]}_{*,p}^2 &\leq 4M_1^2\left\{\sum_{\substack{j=1\\
    |\gamma(j)\cap \gamma(i)| = 0}}^{\Gamma} 1 +\sum_{\substack{j=1\\
    |\gamma(j)\cap \gamma(i)| = 2}}^{\Gamma} 1 +\cdots + \sum_{\substack{j=1\\
    |\gamma(j)\cap \gamma(i)| = k-1}}^{\Gamma}1\right\},
\end{align}
with 
\begin{equation}
    M_1 := \frac{1}{\opnorm{I}_{*,p}}\max\left\{\opnorm{H_{\gamma(k)}H_{\gamma(l)}}_{*,p}: k,l \in [\Gamma], k \neq l\right\}.
\end{equation}

Observe that for any $0\leq s\leq k-1$,
\begin{equation}
    \sum_{\substack{j=1\\
    |\gamma(j)\cap \gamma(i)| = s}}^{\Gamma}1 = \binom{n-k}{k-s} \binom{k}{s},
\end{equation}
since we can construct the hyperedge $\gamma(j)$ in two steps: first we choose $k-s$ elements out of $n-k$ that do not belong to $\gamma(i)$, and then we choose $s$ additional elements from $\gamma(i)$. Therefore,
\begin{equation}
    \sum_{j=1}^{i-1}\opnorm{[H_{\gamma(j)}, H_{\gamma(i)}]}_{*,p}^2\leq 4M_1^2Q(n,k)
\end{equation}
with $Q(n,k)$ defined in equation~\eqref{eq:Q(n,k)}.\footnote{The same $Q(n,k)$ turns out to be important also in the higher-order case since it captures the commutation relation of the Majoranas in the SYK model.}
Similarly, observe that for $2\leq i \leq \Gamma$,
\begin{equation}
    \left(\sum_{j=1}^{i-1}\frac{\tau^2}{2}\opnorm{[H_{\gamma(j)},[H_{\gamma(j)},H_{\gamma(i)}]]}_{*,p}\right)^2
    \leq (2\tau^2M_2Q(n,k))^2,
\end{equation}
where 
\begin{equation}
    M_2 := \frac{1}{\opnorm{I}_{*,p}}\max\left\{\opnorm{H_{\gamma(k)}^2H_{\gamma(l)}}_{*,p}: k,l \in [\Gamma], k \neq l\right\}.
\end{equation} 
Putting all of the above together,
\begin{equation}
    \label{eq:Etau}
    \opnorm{\mathcal{E}(\tau)}_{*,p}^2 \leq 2\opnorm{I}_{*,p}^2 C_p \sum_{i=2}^{\Gamma}\left[16C_p\tau^2 M_1^2Q(n,k)
    +4\tau^4 M_2^2Q(n,k)^2\right].
\end{equation}

It remains to upper bound the constants $M_1$ and $M_2$. When we map the SYK model to a qubit Hamiltonian through, for instance, the Jordan--Wigner transformation, each local term of the SYK model can be written as an i.i.d.\@ standard Gaussian coefficient multiplied by a deterministic Pauli string:
\begin{equation}
    H_{\gamma(i)} \mapsto i^{k(k-1)/2} J_{\gamma(i)}\sigma_{\gamma(i)}.
\end{equation}
For every Pauli string, its $*$-norm is exactly equal to $\opnorm{I}_{*,p}$ by unitarity. Therefore,
\begin{align}
    \label{eq: M1}
    M_1 &= \max\left\{\of[\Big]{\E\of[\big]{|J_{\gamma(k)}J_{\gamma(l)}|^p}}^{\frac{1}{p}}: k,l \in [\Gamma], k \neq l\right\}, \\
    \label{eq: M2}
    M_2 &= \max\left\{\of[\Big]{\E\of[\big]{|J_{\gamma(k)}^2J_{\gamma(l)}|^p}}^{\frac{1}{p}}: k,l \in [\Gamma], k \neq l\right\}.
\end{align}
Since all $J_{\gamma(i)}$ are i.i.d., let $\tilde{J}$ denote a Gaussian coefficient as in equation~\eqref{eq:mean and variance of J}. Then $M_1$ and $M_2$ can be bounded using central absolute moments of the Gaussian distribution $\tilde{J}$:
\begin{align}
    M_1 &= (\E\of{|\tilde{J}|^{p}})^{1/p} \cdot (\E\of{|\tilde{J}|^{p}})^{1/p} \leq \of[\big]{\sigma^p(p-1)!!}^{2/p} \leq \sigma^2 p, \\
    M_2 &= (\E\of{|\tilde{J}|^{2p}})^{1/p} \cdot (\E\of{|\tilde{J}|^{p}})^{1/p} \leq \of[\big]{\sigma^{2p}(2p - 1)!!}^{1/p} \cdot \of[\big]{\sigma^p(p-1)!!}^{1/p} \leq \sigma^2(2p) \cdot \sigma \sqrt{p},
\end{align}
where $\sigma = \mathcal{J}\sqrt{\frac{(k-1)!}{kn^{k-1}}}$ is the standard deviation of $\tilde{J}$. 

Inserting this back into equation~\eqref{eq:Etau},
\begin{align}
    \opnorm{\mathcal{E}(\tau)}_{*,p}^2 &\leq 2\opnorm{I}_{*,p}^2 C_p \sum_{i=2}^{\Gamma}\left[16C_p\tau^2 \sigma^4 p^2Q(n,k)
    +16\tau^4 \sigma^6 p^3 Q(n,k)^2\right]\\
    &\leq 32\opnorm{I}_{*,p}^2 p^4 \binom{n}{k}\left[\tau^2 \sigma^4 Q(n,k)
    +\tau^4 \sigma^6 Q(n,k)^2\right]
\end{align}
where in the second inequality we used $C_p < p$ to pull all the factors of $p$ to the front, as well as $\Gamma = \binom{n}{k}$ to remove the sum. 

Hence, by subadditivity of the square root,
\begin{equation}
    \opnorm{\mathcal{E}(\tau)}_{*,p} \leq 4\sqrt{2}\opnorm{I}_{*,p} p^2 \sigma^2 \sqrt{\binom{n}{k}Q(n,k)}\left[\tau 
    +\tau^2 \sigma \sqrt{Q(n,k)}\right].
\end{equation}

If we divide the total evolution time $t$ into $r$ rounds, where each round has duration $t/r$, the error per round can be bounded using the triangle inequality:
\begin{align}
    \opnorm{e^{iHt/r} - S_1(t/r)}_{*,p} &\leq \int_0^{t/r}\opnorm{\mathcal{E(\tau)}}_{*,p}d\tau\\
    &\leq 4\sqrt{2}\opnorm{I}_{*,p} p^2 \sigma^2 \sqrt{\binom{n}{k}Q(n,k)}\left[\frac{t^2}{2r^2} 
    + \sigma \sqrt{Q(n,k)}\frac{t^3}{3r^3} \right].
\end{align}
The total first-order Trotter error is then bounded by 
\begin{align}
    \opnorm{e^{iHt} - S_1(t/r)^r}_{*,p} &\leq r\cdot \opnorm{e^{iHt/r} - S_1(t/r)}_{*,p}\\
    &\leq 4\sqrt{2}\opnorm{I}_{*,p} p^2 \sigma^2 \sqrt{\binom{n}{k}Q(n,k)}\;t^2\left[\frac{1}{2r} 
    + \sigma \sqrt{Q(n,k)}\frac{t}{3r^2} \right]
\end{align}
which concludes the proof.
\end{proof}

\subsection{Gate complexity}\label{subsec:comp}
Now that we have a bound on the first-order Trotter error, we can bound the associated gate complexity. 

\begin{corr}
When simulating the $k$-local SYK model \eqref{eq:SYKdef} with $r$ rounds of the first-order product formula,
the error bound
\begin{equation}
    \label{eq:concineq again}
    \Prob\left(\norm{e^{iHt} - S_1(t/r)^r} \geq \epsilon \right) < \delta
\end{equation}
can be ensured with gate complexity
\begin{equation}
    G = \tilde{\mathcal{O}}\left(n^{g(k)}\left(n + \log(e^2/\delta)\right)^2 \frac{\mathcal{J}^2t^2}{\epsilon} + n^{g(k) - \frac{1}{4}}\left(n + \log(e^2/\delta)\right) \frac{\mathcal{J}^{\frac{3}{2}}t^{\frac{3}{2}}}{\sqrt{\epsilon}}\right),
\end{equation}
where $\mathcal{J}$ controls the variance of the coupling coefficients, see equation~\eqref{eq:mean and variance of J}, and
\begin{equation}
    \label{eq:g}
    g(k) := \begin{cases}
        k + \frac{1}{2} & \text{if $k$ is even},\\
        k + 1& \text{if $k$ is odd}.
    \end{cases}
\end{equation}
In terms of $n$, the gate complexity $G$ scales as 
\begin{equation}
    \label{eq:G}
   G \sim \begin{cases}
        n^{k + \frac{5}{2}}\mathcal{J}^2 t^2 & \text{if $k$ is even},\\
        n^{k + 3} \mathcal{J}^{2}t^2 & \text{if $k$ is odd}.
    \end{cases}
\end{equation}
\end{corr}

\begin{proof}
By Theorem~\ref{thm:FOE}, the first-order Trotter error is bounded by
\begin{equation}
    \opnorm{e^{iHt} - S_1(t/r)^r}_p\leq p\opnorm{I}_p \lambda(p,r)
\end{equation}
where
\begin{equation}
    \lambda(p,r) := 4\sqrt{2} p \sigma^2 \sqrt{\binom{n}{k}Q(n,k)}\;t^2\left[\frac{1}{2r} 
    + \sigma \sqrt{Q(n,k)}\frac{t}{3r^2} \right].
\end{equation}
This gives us an explicit expression for the concentration inequality~\eqref{eq: actual_concineq}.
By Lemma~\ref{lem:concineq}, finding a Trotter number $r$ that ensures the error bound \eqref{eq:concineq again} is equivalent to finding $r$ that solves the inequality
\begin{equation}
    e4\sqrt{2} p_D^2 \cdot \sigma^2 \sqrt{\binom{n}{k}Q(n,k)}\;\frac{t^2}{\epsilon}\left[\frac{1}{2r} 
    + \sigma \sqrt{Q(n,k)}\frac{t}{3r^2} \right] \leq 1
\end{equation}
where $p_D := \log(e^2D/\delta)$ with $D := \opnorm{I}_p^p$.

To ensure the error bound \eqref{eq:concineq again}, we need
\begin{equation}
   r \geq e4\sqrt{2}p_D^2 \cdot \sigma^2 \sqrt{\binom{n}{k}Q(n,k)}\;\frac{t^2}{2\epsilon} + \left(e4\sqrt{2}p_D^2 \cdot \sigma^3 \sqrt{\binom{n}{k}}Q(n,k) \frac{t^3}{3\epsilon}\right)^{\frac{1}{2}}.
\end{equation}
Recall from Section~\ref{sec:rep and norm} that for (even) $n$ Majoranas $\opnorm{I}_p$ denotes the Schatten $p$-norm on $\mathrm{GL}(2^{n/2}, \C)$, hence $D=\opnorm{I}_p^p = 2^{n/2}$ and
$p_D = \frac{n}{2} \log 2 + \log(e^2/\delta)$.

Recall from Section~\ref{sec:gate complexity} that the gate complexity scales as $G = \tilde{\mathcal{O}}(\Upsilon \Gamma r)$, where $\Upsilon = \mathcal{O}(1)$ for first-order formulas, $\Gamma = \binom{n}{k}$ is the number of terms, and $r$ is the number of rounds.
Hence,
\begin{align}
    G 
    &=\tilde{\mathcal{O}}\left(\binom{n}{k}\cdot \left[e4\sqrt{2}\left(\frac{\log2}{2}n + \log(e^2/\delta)\right)^2 \sigma^2 \sqrt{\binom{n}{k}Q(n,k)}\;\frac{t^2}{2\epsilon} \right.\right.\\
    &\quad\quad\quad \quad + \left.\left. \left(e4\sqrt{2}\left(\frac{\log2}{2}n + \log(e^2/\delta)\right)^2 \sigma^3 \sqrt{\binom{n}{k}}Q(n,k) \frac{t^3}{3\epsilon}\right)^{\frac{1}{2}}\right]\right).\nonumber
\end{align}
To analyze the asymptotic scaling of $G$ for fixed $k$, note that $\binom{n}{k} = \mathcal{O}(n^k)$. Moreover, note from equation~\eqref{eq:Q(n,k)} that $Q(n,k) = \mathcal{O}(n^{k-1})$ for even $k$ while $Q(n,k) = \mathcal{O}(n^k)$ for odd $k$.
Finally, note from equation~\eqref{eq:mean and variance of J} that $\sigma = \mathcal{O}(\mathcal{J} / \sqrt{n^{k-1}})$.
Putting this all together,
\begin{equation}
\label{eq: GFO}
    G = \tilde{\mathcal{O}}\left(n^{g(k)}\left(n + \log(e^2/\delta)\right)^2 \frac{\mathcal{J}^2t^2}{\epsilon} + n^{g(k) - \frac{1}{4}}\left(n + \log(e^2/\delta)\right) \frac{\mathcal{J}^{\frac{3}{2}}t^{\frac{3}{2}}}{\sqrt{\epsilon}}\right)
\end{equation}
where $g(k)$ is defined in equation~\eqref{eq:g}.
The asymptotic scaling of $G$ is as stated in equation~\eqref{eq:G} since the first term in \eqref{eq: GFO} is dominant.
\end{proof}

Alternatively, we can consider the error bound in equation \eqref{eq:concineqpsi} for a fixed input state $\ket{\psi}.$ In that case, we use Lemma \ref{lem:concineqpsi} instead of Lemma \ref{lem:concineq}, which leads us to the following result:\footnote{Using Lemma \ref{lem:concineqpsi} is ultimately the same as replacing the factors of $n + \log(e^2/\delta)$ with $\log(e^2/\delta)$.}

\begin{corr}
When simulating the $k$-local SYK model \eqref{eq:SYKdef} with $r$ rounds of the first-order product formula, the error bound
\begin{equation}
    \Prob\left(\norm{(e^{iHt} - S_1(t/r)^r)\ket{\psi}}_{l^2} \geq \epsilon \right) < \delta
\end{equation}
for an arbitrary input state $\ket{\psi}$ can be ensured with gate complexity
\begin{equation}
    G = \tilde{\mathcal{O}}\left(n^{g(k)}\left( \log(e^2/\delta)\right)^2 \frac{\mathcal{J}^2t^2}{\epsilon} + n^{g(k) - \frac{1}{4}}\left(\log(e^2/\delta)\right) \frac{\mathcal{J}^{\frac{3}{2}}t^{\frac{3}{2}}}{\sqrt{\epsilon}}\right),
\end{equation}
where $\mathcal{J}$ controls the variance of the coupling coefficients, see equation~\eqref{eq:mean and variance of J}, and $g(k)$ is defined in equation~\eqref{eq:g}.
In terms of $n$, the gate complexity $G$ scales as 
\begin{equation}
   G \sim \begin{cases}
        n^{k + \frac{1}{2}}\mathcal{J}^2 t^2 & \text{if $k$ is even},\\
        n^{k + 1} \mathcal{J}^{2}t^2 & \text{if $k$ is odd}.
    \end{cases}
\end{equation}
\end{corr}

\section{Higher-order Trotter error of the SYK model}
\label{sec:HOSYK}

In this section, we will extend our first-order Trotter error results from the Section~\ref{sec:FOE} to higher-order Trotter error. As our main tool we will use the uniform smoothness Theorem~\ref{thm:RPG} to evaluate the error generator as random matrix polynomials. Based on the error generator, we will then derive bounds for higher-order Trotter error, from which we will derive the gate complexity for simulating the SYK Hamiltonian using higher-order product formulas.

In equation~\eqref{eq:expT} in Section~\ref{subsec:Error Representation}, we introduced a general exponential representation of the Trotter error, where we wrote the product formula as
\begin{equation}
    \prod_{j=1}^J e^{ia_jH_{\gamma(b_j)}t} = \mathrm{exp}_{\mathcal{T}}\left(i\int(\mathcal{E}(t) + H)dt\right),
\end{equation}
where the error generator $\mathcal{E}(t)$ defined in equation~\eqref{eq:Et def} takes the form 
\begin{equation}
    \mathcal{E}(t) := \sum_{j=1}^J\left(\prod_{k=j+1}^J e^{a_k\Lag_{\gamma(b_k)}t}[a_jH_{\gamma(b_j)}] - a_jH_{\gamma(b_j)}\right).
\end{equation}
Here $\Lag_{\gamma(j)}[O] := i[H_{\gamma(j)}, O]$, see equation~\eqref{eq:Lj}, and $J=\Gamma \cdot \Upsilon$ with $\Gamma$ the number of terms in the $k$-local Hamiltonian and $\Upsilon$ the number of stages of the product formula.

Following the approach in \cite{Anthony}, we can do a Taylor expansion on the error generator $\mathcal{E}(t)$ and collect terms by the degree of product of local terms $H_{\gamma(1)},\dotsc, H_{\gamma(\Gamma)}$, which leads to
\begin{equation}
    \label{eq:E}
    \mathcal{E}(t) = \sum_{g=l+1}^{g'-1}\mathcal{E}_g(t) + \mathcal{E}_{\geq g'}(t),
\end{equation}
where $g'$ is some cutoff and the sum starts at $l+1$ since the $l$-th order product formula is exact up to order $l$. The $\mathcal{E}_g(t)$ term is the $g$-th order error generator given by 
\begin{equation}
    \mathcal{E}_g(t) := \sum_{j=1}^{J}\sum_{g_J + \dotsb + g_{j+1} = g-1}\Lag_{\gamma(J)}^{g_J} \cdots \Lag_{\gamma(j+1)}^{g_{j+1}}[H_{\gamma(j)}] \frac{t^{g-1}}{g_J!\dotsm g_{j+1}!},
\end{equation}
and $\mathcal{E}_{\geq g'}(t)$ is the residual term given by 
\begin{align}
    \label{eq:E rest}
    \mathcal{E}_{\geq g'}(t) &:= \sum_{j=1}^J\sum_{m=j+1}^Je^{\Lag_{\gamma(J)}t}\cdots e^{\Lag_{\gamma(m+1)}t}\\
    &\times \int_0^td\xi \sum_{g_m + \cdots +g_{j+1} = g'-1, g_m\geq 1}e^{\Lag_{\gamma(m)}\xi} \Lag_{\gamma(m)}^{g_m}\cdots \Lag_{\gamma(j+1)}^{g_{j+1}}[H_{\gamma(j)}]\frac{(t-\xi)^{g_m - 1} t^{g' - g_m -1}}{(g_m-1)! g_{m-1}! \cdots g_{j+1}!}.
    \nonumber
\end{align}

\subsection{Error bound}\label{sec:error bound}
We will proceed by separately evaluating the norms of $\mathcal{E}_g$ and $\mathcal{E}_{\geq g'}$ for the SYK model, and then collecting the results together. We will try to write everything down in terms of $Q(n,k)$ as defined in equation~\eqref{eq:Q(n,k)}. In the end, we get the following result:

\begin{theo}
\label{thm:SYKRPG}
     The (even) $l$-th order Trotter error of the SYK model for $l\geq 2$ is bounded by
     \begin{align}
         \frac{\opnorm{e^{iHt} - S_l(t/r)^r}_{*,p}}{\opnorm{I}_{*,p}} &\leq \Delta_{l}(n,k,t,r,p),
     \end{align}
     where 
     \begin{equation}
          \Delta_{l}:= \frac{\mathcal{C}(l)\sqrt{p} \sigma t}{\sqrt{Q(n,k)}} \left( \binom{n}{k} \left[\sqrt{p} \sigma \sqrt{Q(n,k)}\frac{t}{r}\right]^{l}  +  \binom{n}{k}^2\left[\sqrt{p} \sigma\sqrt{Q(n,k)}\frac{t}{r}\right]^{l+1} \right),     
     \end{equation}
     and $D(l)$ is a factor only depending on $l$.
\end{theo}

The rest of Section~\ref{sec:error bound} constitutes a proof of this theorem.

\begin{proof}
When proving this result, we assume that $k$ is even, as the odd $k$ case follows the same argument. We shall first compute the $g$-th order error generator $\mathcal{E}_g$ and the tail bound $\mathcal{E}_{\geq g'}$ separately, and then join them together to obtain the total error. Our steps follow largely the approach in \cite{Anthony}, but with our own computation of coefficients specific to the SYK model.

\subsubsection{\texorpdfstring{$g$}{g}-th order error generator}

In this subsection, we will upper bound
\begin{equation}
    \label{eq:Eg norm}
    \opnorm{\mathcal{E}_g}_{*,p} = \opnorm[\Bigg]{\sum_{j=1}^{J}\sum_{g_J + \dotsb + g_{j+1} = g-1}\Lag_{\gamma(J)}^{g_J} \cdots \Lag_{\gamma(j+1)}^{g_{j+1}}[H_{\gamma(j)}] \frac{t^{g-1}}{g_J! \dotsm g_{j+1}!}}_{*,p}
\end{equation}
by using uniform smoothing in the form of Theorem~\ref{thm:RPG}.

As the summation indices are difficult to work with here, we consider the following symmetrization:
\begin{equation}
    \opnorm{\mathcal{E}_g}_{*,p} = t^{g-1} \opnorm[\Bigg]{\sum_{j_{g-1}=1}^J \cdots \sum_{j_0=1}^J \tilde{\theta}_{\bm{j}}\Lag_{\gamma(j_{g-1})}\cdots \Lag_{\gamma(j_1)}[H_{\gamma(j_0)}]}_{*,p},
\end{equation}
where the sum is over $\bm{j} = (j_{g-1}, \dotsc, j_0) \in [J]^g$, and we have included coefficients $\tilde{\theta}_{\bm{j}}\in [0,1]$ that absorb the factorials $\frac{1}{g_J! \dotsm g_{j+1}!}$ and reduce the range of the sum to match with equation \eqref{eq:Eg norm}. We will not write an explicit expression for these coefficients since later we will ignore them for simplicity anyway.
All that matters is that $\tilde{\theta}_{\bm{j}} \leq 1$ and they impose an extra summation constraint. As $J = \Upsilon\cdot \Gamma$, with $\Gamma = \binom{n}{k}$ the total number of local terms in the Hamiltonian, we have the following expression to bound:
\begin{equation}
    \label{eq:Eg}
    \opnorm{\mathcal{E}_g}_{*,p} \leq \Upsilon^gt^{g-1} \opnorm[\Bigg]{\sum_{j_{g-1}=1}^{\Gamma} \cdots \sum_{j_0=1}^{\Gamma} \tilde{\theta}_{\bm{j}}\Lag_{\gamma(j_{g-1})}\cdots \Lag_{\gamma(j_1)}[H_{\gamma(j_0)}]}_{*,p},
\end{equation}
where the number of stages $\Upsilon$ accounts for repeated values of $\gamma(\bm{j})$ thanks to the periodicity of the ordering map $\gamma$, see equation~\eqref{eq:labeling}.

We will think of the expression we are dealing with as a non-commutative polynomial $F_g$ in the local terms $H_i$ of the Hamiltonian. Following our discussion in Section~\ref{subsec:randommatpol}, we write it as
\begin{align}
    \label{eq:F}
    F_g(H_{\gamma(\Gamma)},\dotsc,H_{\gamma(1)}) &= \sum_{j_{g-1}=1}^{\Gamma} \cdots \sum_{j_0=1}^{\Gamma} \tilde{\theta}_{\bm{j}}\Lag_{\gamma(j_{g-1})}\cdots \Lag_{\gamma(j_1)}[H_{\gamma(j_0)}]\\ 
    &=: \sum_{j_{g-1}=1}^{\Gamma} \cdots \sum_{j_0=1}^{\Gamma} T_{\bm{j}} H_{\gamma(j_{g-1})}\cdots H_{\gamma(j_0)}.
    \label{eq:Tj def}
\end{align}

While the coefficients $T_{\bm{j}}$ can be calculated explicitly, we will not need to do that since for applying Theorem \ref{thm:RPG} it is sufficient to only have an upper bound on their norm.

Note that the commutator chain $\Lag_{\gamma(j_{g-1})}\cdots \Lag_{\gamma(j_1)}[H_{\gamma(j_0)}]$ is a polynomial in the local terms $H_{\gamma(j_{g-1})},\dotsc,H_{\gamma(j_0)}$, with an overall pre-factor of $i^g$: 
\begin{equation}
    \label{eq:com chain}
    \Lag_{\gamma(j_{g-1})} \cdots \Lag_{\gamma(j_1)}[H_{\gamma(j_0)}] = i^g \sum_{\bm{k} \in \bm{j}^g} c_{\bm{j},\bm{k}} H_{\gamma(k_g)}\cdots H_{\gamma(k_1)},
\end{equation}
where $c_{\bm{j},\bm{k}}$ are some coefficients and the sum runs over vectors $\bm{k} = (k_g, \dotsc, k_1)$ in $\bm{j}^g := (j_{g-1},\dotsc,j_0)^g$. As each term $H_{\gamma(k_g)}\cdots H_{\gamma(k_1)}$ emerges from a commutator, its coefficient $c_{\bm{j},\bm{k}}$ is either zero or $\pm 1$ (it is zero when $\bm{j} \neq \bm{k}$ as sets). Since the local terms $H_{\gamma(j_{g-1})},\dotsc,H_{\gamma(j_0)}$ are proportional to Pauli strings, each pair of them either commutes or anti-commutes. Hence, every term $H_{\gamma(k_g)}\cdots H_{\gamma(k_1)}$ with a non-zero $c_{\bm{j},\bm{k}}$ can be reordered to the product $H_{\gamma(j_{g-1})}\cdots H_{\gamma(j_0)}$ that matches the ordering $\bm{j}$ of the commutator chain in equation~\eqref{eq:com chain}, with an extra pre-factor $d_{\bm{j},\bm{k}}$:
\begin{equation}
    H_{\gamma(k_g)}\cdots H_{\gamma(k_1)} = d_{\bm{j},\bm{k}} H_{\gamma(j_{g-1})}\cdots H_{\gamma(j_0)}.
\end{equation} 
The pre-factor $d_{\bm{j},\bm{k}}$ is either $+1$ or $-1$ depending on the (anti-)commutation between the local terms. Therefore, 
\begin{equation}
    \Lag_{\gamma(j_{g-1})}\cdots \Lag_{\gamma(j_1)}[H_{\gamma(j_0)}] = i^g \left(\sum_{\bm{k} \in \bm{j}^g} c_{\bm{j},\bm{k}} d_{\bm{j},\bm{k}} \right) H_{\gamma(j_{g-1})} \cdots H_{\gamma(j_0)}.
\end{equation}
Substituting this in equation~\eqref{eq:F}, we can write the $T_{\bm{j}}$ introduced in equation~\eqref{eq:Tj def} as
\begin{equation}
\label{eq:Tj}
    T_{\bm{j}} = i^g \tilde{\theta}_{\bm{j}} \sum_{\bm{k} \in \bm{j}^g} c_{\bm{j},\bm{k}} d_{\bm{j},\bm{k}}.
\end{equation}

To give an upper bound of $T_{\bm{j}}$, notice that the sum in equation~\eqref{eq:com chain} vanishes if the commutator chain evaluates to zero. We can also remove the coefficients $\tilde{\theta}_{\bm{j}}$ from equation~\eqref{eq:Tj} since they are, by construction, at most one. Hence, by triangle inequality,
\begin{equation}
\label{eq:|T|}
    |T_{\bm{j}}| \leq  \left(\sum_{\bm{k} \in \bm{j}^g}1\right)\cdot \ind(\bm{j}) = g^g \ind(\bm{j}),
\end{equation}
where the indicator function $\ind(\bm{j})$ defined for all $\bm{j}\in [J]^g$ enforces the commutation constraints:
\begin{equation}
    \label{eq:inda}
    \ind(\bm{j}):= \begin{cases}
        1&\quad \text{if }\Lag_{\gamma(j_{g-1})}\cdots \Lag_{\gamma(j_1)}[H_{\gamma(j_0)}]\neq 0,\\
        0&\quad \text{else.}
    \end{cases}
\end{equation}

Lastly, before we can apply Theorem~\ref{thm:RPG}, we rewrite each local term $H_{\gamma(i)}$ in the SYK Hamiltonian~\eqref{eq:HSYK} as
\begin{equation}
    H_{\gamma(i)} = J_{\gamma(i)} K_{\gamma(i)},
\end{equation}
where each $J_{\gamma(i)} \in \R$ is an i.i.d.\@ Gaussian coefficient and $K_{\gamma(i)}$ is a deterministic matrix bounded as $\norm{K_{\gamma(i)}} = 1.$\footnote{Recall from equation~\eqref{eq:H gamma i} that the matrices $K_{\gamma(i)}$ are products of Jordan--Wigner representations of Majorana fermions $\gamma(i)$ ordered lexicographically, and Jordan--Wigner representations of Majorana fermions are Pauli strings, which are unitaries.} However, note that Theorem \ref{thm:RPG} is only applicable when the mean of the coefficients $J_{\gamma(i)}$ is zero and the standard deviation is 1, but the standard deviation of each $J_{\gamma(i)}$ is $\sigma$, see equation~\eqref{eq:mean and variance of J}. Thus, we multiply each $K_{\gamma(i)}$ by $\sigma$ and divide $J_{\gamma(i)}$ by $\sigma$:
\begin{equation}
    H_{\gamma(i)} = \left(\frac{J_{\gamma(i)}}{\sigma}\right) \cdot \sigma K_{\gamma(i)} =: g_{\gamma(i)} \tilde{K}_{\gamma(i)}.
\end{equation}
This way, each $g_{\gamma(i)}$ is a standard Gaussian coefficient with 0 mean and standard deviation 1. Since each $K_{\gamma(i)}$ has norm 1, each $\tilde{K}_{\gamma(i)}$ is a deterministic matrix bounded by $\norm{\tilde{K}_{\gamma(i)}}\leq \sigma$.

We can now view the $F_g$ in equation~\eqref{eq:F} as a standard Gaussian random matrix polynomial. We can directly apply Theorem \ref{thm:RPG} with $p\geq 2$ and $m=\Gamma$ to obtain
\begin{equation}
    \label{eq:Fg bound}
    \opnorm{F_g}_{*,p}^2 \leq C_p^g g^{g}\sum_{\bm{w}: |\bm{w}|\leq g} \left(\sum_{\pi\in \mathcal{S}_g}\sum_{ \bm{v}: |\bm{w}|+ 2|\bm{v}|=g} \left|T_{\pi[\bm{\eta}(\bm{w}, 2\bm{v})]} \right|\prod_{i=1}^m \sigma^{w_i + 2v_i}\right)^2 \opnorm{I}_{*,p}^2,
\end{equation}
where $\bm{w},\bm{v}\in \{0,\dotsc,g\}^{\Gamma}$ such that $|\bm{w}|+2|\bm{v}|=g$. The coefficients $T_{\pi[\bm{\eta}(\bm{w}, 2\bm{v})]}$ are inherited from the corresponding $T_{\bm{j}}$ in equation~\eqref{eq:Tj} by replacing the vector $\bm{j}$ with $\pi[\bm{\eta}(\bm{w}, 2\bm{v})]$ where $\bm{\eta}$ is the unary encoding function defined in equation~\eqref{eq:eta}. By the bound in equation~\eqref{eq:|T|},
\begin{equation}
\label{eq:T}
    \left|T_{\pi[\bm{\eta}(\bm{w}+2\bm{v})]}\right| \leq g^g\ind(\pi[\bm{\eta}(\bm{w}, 2\bm{v})]),
\end{equation}
where the indicator $\ind(\bm{a})$ is defined in equation~\eqref{eq:inda} on all vectors $\bm{a} \in [\Gamma]^g$.
Since $|\bm{w}| + 2|\bm{v}| = g$ by construction, $\prod_{i=1}^m \sigma^{w_i + 2v_i} = \sigma^g$ and hence equation~\eqref{eq:Fg bound} simplifies to
\begin{equation}
\label{eq:F_g1}
    \opnorm{F_g}_{*,p}^2 \leq C_p^g g^{3g}\sigma^{2g}\sum_{\bm{w}: |\bm{w}|\leq g} \left(\sum_{\pi\in \mathcal{S}_g}\sum_{ \bm{v}: |\bm{w}|+2|\bm{v}|=g}\ind(\pi[\bm{\eta}(\bm{w}, 2\bm{v})]) \right)^2 \opnorm{I}_{*,p}^2.
\end{equation}

Let us split the sum over $\bm{w}$ as
\begin{equation}
    \label{eq:sum split}
    \sum_{\bm{w}: |\bm{w}|\leq g} = \sum_{w=0}^g \sum_{\bm{w}:|\bm{w}| = w}
\end{equation}
and rewrite equation~\eqref{eq:F_g1} as
\begin{equation}
    \label{eq:F_g simpler}
    \opnorm{F_g}_{*,p}^2 \leq C_p^g g^{3g}\sigma^{2g}
    \sum_{w=0}^g G_w(H) \opnorm{I}_{*,p}^2
\end{equation}
where
\begin{equation}
    \label{eq:Gw(H)}
    G_w(H) := \sum_{\bm{w}: |\bm{w}| = w}\left(\sum_{\pi\in \mathcal{S}_g}\sum_{ \bm{v}: |\bm{w}|+2|\bm{v}|=g}\ind(\pi[\bm{\eta}(\bm{w}, 2\bm{v})])\right)^2.
\end{equation}
We can then fix a $w \in \set{0,\dotsc,g}$ and focus on upper bounding $G_w(H)$. We can do this using the following lemma:

\begin{restatable}{lemma}{lemmaGwH}
\label{lem:Gw}
    Let $H = \sum_{i=1}^m  H_{\gamma(i)}$ be a Hamiltonian of $m$ terms with arbitrary ordering $\gamma$, such that for all $1\leq i<j\leq m$, $H_{\gamma(i)}$ and $H_{\gamma(j)}$ either commute or anti-commute. Fix $g \geq 1$ and $w \in \set{0,\dotsc,g}$, and consider the $H$-dependent sum
    \begin{equation}
        G_w(H) := \sum_{\bm{w}: |\bm{w}| = w}\left(\sum_{\pi\in \mathcal{S}_g}\sum_{ \bm{v}: |\bm{w}|+2|\bm{v}|=g}\ind(\pi[\bm{\eta}(\bm{w}, 2\bm{v})])\right)^2,
    \end{equation}
    where $\bm{w},\bm{v} \in \{0,\dotsc,g\}^m$ such that $|\bm{w}|=w$ and $|\bm{w}| + 2|\bm{v}|=g$. The permutation $\pi\in \mathcal{S}_g$ reorders the entries of vectors in $[m]^g$. The unary encoding $\bm{\eta} $ and the indicator $\ind$ are defined on vectors $\bm{a},\bm{b}\in \{0,\dotsc,g\}^m$ and $\bm{j}\in [m]^g$ as follows:
    \begin{equation}
        \bm{\eta}(\bm{a}, \bm{b}) := (\underbrace{1,\dotsc,1}_{a_1\text{ times}}, \underbrace{2,\dotsc,2}_{a_2\text{ times}}, \dotsc, \underbrace{m,\dotsc,m}_{a_m\text{ times}},\underbrace{1,\dotsc,1}_{b_1\text{ times}}, \underbrace{2,\dotsc,2}_{b_2\text{ times}}, \dotsc, \underbrace{m,\dotsc,m}_{b_m\text{ times}}),
    \end{equation}
    and
    \begin{equation}
        \ind(\bm{j}) := \begin{cases}
        1 & \text{if }\Lag_{\gamma(j_{g})}\cdots \Lag_{\gamma(j_2)}[H_{\gamma(j_1)}]\neq 0,\\
        0 & \text{else}.
    \end{cases}
    \end{equation}
    Then
    \begin{equation}
        G_w(H) \leq g^{3g-2} m^2 Q_{\max}(H)^{g-2}
    \end{equation}
    where
    \begin{equation}
    \label{eq:Qmax}
        Q_{\max}(H) := \max_{1\leq i\leq m} \left|\left\{j \in [m]: [H_{\gamma(j)},H_{\gamma(i)}]\neq 0\right\}\right|.
    \end{equation}
\end{restatable}

The proof of this lemma can be found in Appendix~\ref{apx:Gw}. For the SYK model, we have $m = \binom{n}{k}$ and
$Q_{\max}(H) = Q(n,k)$.
Hence, by Lemma~\ref{lem:Gw}
\begin{equation}
    G_w(H) \leq g^{3g - 2} \binom{n}{k}^2 Q(n,k)^{g-2}.
\end{equation}
Plugging this back into equation~\eqref{eq:F_g simpler} gives us
\begin{equation}
    \opnorm{F_g}_{*,p}^2 \leq C_p^g g^{3g}\sigma^{2g}
    \cdot (g+1) \cdot
    g^{3g - 2} \binom{n}{k}^2 Q(n,k)^{g-2}
    \opnorm{I}_{*,p}^2
\end{equation}
where the factor $g+1$ comes from the sum over $w$.
Taking the square root of both sides we get
\begin{equation}
\label{eq:F_g2}
    \opnorm{F_g}_{*,p} \leq C_p^{\frac{g}{2}} (g+1)^{\frac{1}{2}}g^{3g-1}\sigma^{g} \binom{n}{k} Q(n,k)^{\frac{g}{2}-1} \opnorm{I}_{*,p}.
\end{equation}
Plugging this back into equation~\eqref{eq:Eg} gives us
\begin{equation}
\label{eq:E_g}
    \opnorm{\mathcal{E}_g}_{*,p} \leq \Upsilon^gt^{g-1}C_p^{\frac{g}{2}} (g+1)^{\frac{1}{2}}g^{3g-1}\sigma^{g} \binom{n}{k} Q(n,k)^{\frac{g}{2}-1} \opnorm{I}_{*,p}.
\end{equation}

\subsubsection{Beyond \texorpdfstring{$g$}{g}-th order}

Next, we want to upper bound the residual error term defined in equation~\eqref{eq:E rest}:
\begin{align}
   \opnorm{ \mathcal{E}_{\geq g'}}_{*,p} = \opnorm[\Bigg]{&\sum_{j=1}^J\sum_{m=j+1}^J e^{\Lag_{\gamma(J)}t}\cdots e^{\Lag_{\gamma(m+1)}t}\\
    &\times \int_0^td\xi \sum_{g_m + \cdots +g_{j+1} = g'-1, g_m\geq 1}e^{\Lag_{\gamma(m)}\xi} \Lag_{\gamma(m)}^{g_m}\cdots \Lag_{\gamma(j+1)}^{g_{j+1}}[H_{\gamma(j)}]\frac{(t-\xi)^{g_m - 1} t^{g' - g_m -1}}{(g_m-1)! g_{m-1}! \cdots g_{j+1}!}}_{*,p}.
    \nonumber
\end{align}

We would first like to remove the exponentials inside the expected Schatten norm. This is hindered by the fact that they depend on the summation index $m$. The workaround is to replace the triangular double summation by
\begin{equation}
    \sum_{j=1}^J \sum_{m=j+1}^{J} \mapsto \sum_{m=2}^{J} \sum_{j=1}^{m-1}.
\end{equation}
Hence,
\begin{align}
   \opnorm{ \mathcal{E}_{\geq g'}}_{*,p} = \opnorm[\Bigg]{&\sum_{m=2}^{J}\sum_{j=1}^{m-1} e^{\Lag_{\gamma(J)}t}\cdots e^{\Lag_{\gamma(m+1)}t}\\
    &\times \int_0^td\xi \sum_{g_m + \cdots +g_{j+1} = g'-1, g_m\geq 1}e^{\Lag_{\gamma(m)}\xi} \Lag_{\gamma(m)}^{g_m}\cdots \Lag_{\gamma(j+1)}^{g_{j+1}}[H_{\gamma(j)}]\frac{(t-\xi)^{g_m - 1} t^{g' - g_m -1}}{(g_m-1)! g_{m-1}! \cdots g_{j+1}!}}_{*,p}.
    \nonumber
\end{align}
By triangle inequality (and taking the integral outside the sum),
\begin{align}
   \opnorm{ \mathcal{E}_{\geq g'}}_{*,p} \leq \sum_{m=2}^{J}\int_0^td\xi\opnorm[\Bigg]{&\sum_{j=1}^{m-1} e^{\Lag_{\gamma(J)}t}\cdots e^{\Lag_{\gamma(m+1)}t}\\
    &\times  \sum_{g_m + \cdots +g_{j+1} = g'-1, g_m\geq 1}e^{\Lag_{\gamma(m)}\xi} \Lag_{\gamma(m)}^{g_m}\cdots \Lag_{\gamma(j+1)}^{g_{j+1}}[H_{\gamma(j)}]\frac{(t-\xi)^{g_m - 1} t^{g' - g_m -1}}{(g_m-1)! g_{m-1}! \cdots g_{j+1}!}}_{*,p}.
    \nonumber
\end{align}
Since the exponentials inside the expected Schatten norm do not dependent on the summations over $j$ and $g_i$'s, we can pull them in front of both sums and then by unitary invariance of the norm remove altogether:
\begin{align}
\label{eq:Egeqg}
    \opnorm{\mathcal{E}_{\geq g'}}_{*,p} &\leq  \sum_{m=2}^J\int_0^t d\xi \opnorm[\Bigg]{\sum_{j=1}^{m-1} \sum_{g_m + \cdots +g_{j+1} = g'-1, g_m\geq 1} \Lag_{\gamma(m)}^{g_m}\cdots \Lag_{\gamma(j+1)}^{g_{j+1}}[H_{\gamma(j)}]\frac{(t-\xi)^{g_m - 1} t^{g' - g_m -1}}{(g_m-1)! g_{m-1}! \cdots g_{j+1}!}}_{*,p}.
\end{align}

Consider the same type of symmetrization as in the previous section:
\begin{align}
    \label{eq:geq g terms}
    \opnorm{\mathcal{E}_{\geq g'}}_{*,p} &\leq  \sum_{m=2}^J\int_0^t d\xi \opnorm[\Bigg]{\sum_{j_{g'-1}=1}^J \cdots \sum_{j_0=1}^J \tilde{\theta}_{\bm{j}}\Lag_{\gamma(j_{g'-1})}\cdots \Lag_{\gamma(j_1)}[H_{\gamma(j_0)}]}_{*,p},
\end{align}
where we have absorbed $\frac{(t-\xi)^{g_m - 1} t^{g' - g_m -1}}{(g_m-1)! g_{m-1}! \cdots g_{j+1}!}$ into the coefficients $\tilde{\theta}_{\bm{j}}$, which are set to zero as needed to reduce this sum to the previous one. Compared to the previous section where $\tilde{\theta}_{\bm{j}}\in [0,1]$, the coefficients $\tilde{\theta}_{\bm{j}}$ are now bounded by
\begin{equation}
    |\tilde{\theta}_{\bm{j}}| \leq t^{g'-2}
\end{equation}
due to the extra factor of $(t-\xi)^{g_m - 1}t^{g' - g_m - 1}$ in their definition.

Similar to our strategy for the $g$-th order terms in equation~\eqref{eq:Tj def}, we can write the sums inside the norm in equation~\eqref{eq:geq g terms} as a random matrix polynomial with coefficients $T_{\bm{j}}$:
\begin{equation}
    F_{g'}(H_{\gamma(\Gamma)},\dotsc,H_{\gamma(1)}) := \sum_{j_{g'-1}=1}^{\Gamma} \cdots \sum_{j_0=1}^{\Gamma} T_{\bm{j}} H_{\gamma(j_{g'-1})}\cdots H_{\gamma(j_0)}
\end{equation}
where the coefficients $T_{\bm{j}}$ are now bounded by 
\begin{equation}
\label{eq:|T|new}
    |T_{\bm{j}}| \leq (g')^{g'} t^{g'-1} \ind(\bm{j}),
\end{equation}
where the indicator function $\ind(\bm{j})$ is defined for all vectors $\bm{j}\in [J]^{g'}$ in the same way as in equation~\eqref{eq:inda}.
By the same argument as in the previous section, we get the following bound that follows directly from equation~\eqref{eq:F_g2} with the updated upper bound of the coefficients $T_{\bm{j}}$ in equation \eqref{eq:|T|new}:
\begin{equation}
    \opnorm{F_{g'}}_{*,p} \leq C_p^{\frac{g'}{2}} (g'+1)^{\frac{1}{2}}(g')^{3g'-1}\sigma^{g'} \binom{n}{k} Q(n,k)^{\frac{g'}{2}-1} \opnorm{I}_{*,p}.
\end{equation}
Plug this back into equation \eqref{eq:Egeqg} (the integral over $\xi$ simply becomes a factor $t$), we get 
\begin{align}
    \label{eq:residual term bound}
    \opnorm{\mathcal{E}_{\geq g'}}_{*,p}
    &\leq \Upsilon^{g'+1} t^{g'-1} C_p^{\frac{g'}{2}} (g'+1)^{\frac{1}{2}}(g')^{3g'-1}\sigma^{g'} \binom{n}{k}^2 Q(n,k)^{\frac{g'}{2}-1} \opnorm{I}_{*,p}.
\end{align}

\subsubsection{Total error}

Now, we can finally put everything together. For an $l$-th order $\Upsilon$-stage product formula, the norm of the error generator $\mathcal{E}$ of the SYK model can be bounded by applying the triangle inequality in equation~\eqref{eq:E} and then using the bounds from equations~\eqref{eq:E_g} and~\eqref{eq:residual term bound}:
\begin{align}
    \opnorm{\mathcal{E}}_{*,p} &\leq \sum_{g=l+1}^{g'-1} \opnorm{\mathcal{E}_g}_{*,p} + \opnorm{\mathcal{E}_{\geq g'}}_{*,p}\\
    &\leq \opnorm{I}_{*,p}\sum_{g=l+1}^{g'-1} \Upsilon^gt^{g-1}C_p^{\frac{g}{2}} (g+1)^{\frac{1}{2}}g^{3g-1}\sigma^{g} \binom{n}{k} Q(n,k)^{\frac{g}{2}-1}\\
    & + \opnorm{I}_{*,p}\Upsilon^{g'+1} t^{g'-1} C_p^{\frac{g'}{2}} (g'+1)^{\frac{1}{2}}(g')^{3g'-1}\sigma^{g'} \binom{n}{k}^2 Q(n,k)^{\frac{g'}{2}-1}.\nonumber
\end{align}
For $g'=l+2$, we have
\begin{align}
    \opnorm{\mathcal{E}}_{*,p} 
    &\leq \opnorm{I}_{*,p} \mathcal{A}(l) \left(t^{l}\left[\sqrt{C_p} \sigma \sqrt{Q(n,k)}\right]^{l+1} \binom{n}{k} Q(n,k)^{-1}  +  t^{l+1} \left[\sqrt{C_p} \sigma\sqrt{Q(n,k)}\right]^{l+2} \binom{n}{k}^2 Q(n,k)^{-1}\right),
\end{align}
where $\mathcal{A}(l)=\Upsilon^{l+3}(l+3)^{\frac{1}{2}}(l+2)^{3(l+2)-1}.$
As we divide the total evolution time $t$ into $r$ rounds, with each round having duration $\frac{t}{r}$, the error per round can be bounded as
\begin{align}
    \opnorm{e^{iHt/r} - S_l(t/r)}_{*,p} &\leq \int_0^{\frac{t}{r}}\opnorm{\mathcal{E}(\tau)}d\tau \\
    &\leq \opnorm{I}_{*,p} \frac{\mathcal{A}(l)}{l+1} \left(\frac{t^{l+1}}{r^{l+1}}\left[\sqrt{C_p} \sigma \sqrt{Q(n,k)}\right]^{l+1} \binom{n}{k} Q(n,k)^{-1}\right.\\
    &\left. \qquad \qquad \quad \; {}+ \frac{t^{l+2}}{r^{l+2}} \left[\sqrt{C_p} \sigma\sqrt{Q(n,k)}\right]^{l+2} \binom{n}{k}^2 Q(n,k)^{-1}\right). \nonumber
\end{align}
Using $C_p < p$, the total error is then bounded as
\begin{align}
    \opnorm{e^{iHt} - S_l(t/r)}_{*,p} &\leq r\cdot \opnorm{e^{iHt/r} - S_l(t/r)}_{*,p}\\
    &\leq  \opnorm{I}_{*,p} \mathcal{C}(l)\sqrt{p} \sigma Q(n,k)^{-\frac{1}{2}} t\\
    & \quad\quad\quad \quad \cdot \left(\left[\sqrt{p} \sigma \sqrt{Q(n,k)}\frac{t}{r}\right]^{l} \binom{n}{k}   +  \left[\sqrt{p} \sigma\sqrt{Q(n,k)}\frac{t}{r}\right]^{l+1} \binom{n}{k}^2 \right) \nonumber
\end{align}
with $\mathcal{C}(l) = \frac{\mathcal{A}(l)}{l+1}$. This concludes the proof of Theorem~\ref{thm:SYKRPG}.
\end{proof}

\subsection{Gate complexity}

\begin{corr}
    Simulating a $k$-local SYK model with (even) $l$-th order product formula, the gate complexity
    \begin{equation}
     G = 
        \tilde{\mathcal{O}}\left[ n^{g_1(k)}\sqrt{n + \log(e^2/\delta)} \mathcal{J}t\cdot \left(\left[n^{g_2(k)}\sqrt{n + \log(e^2/\delta)}\frac{\mathcal{J}t}{\epsilon}\right]^{\frac{1}{l}} + \left[n^{g_3(k)}\sqrt{n+\log(e^2/\delta)} \frac{\mathcal{J}t}{\epsilon}\right]^{\frac{1}{l+1}}\right) \right],
\end{equation}
where
\begin{equation}
    \label{eq:3gs}
    \left(g_1(k),\; g_2(k),\; g_3(k)\right) :=
    \begin{cases}
        \left(k,\; 1,\; k+1\right) & \text{if $k$ is even},\\
        \left(k + \frac{1}{2},\; \frac{1}{2},\; k + \frac{1}{2}\right) & \text{if $k$ is odd},
    \end{cases}
\end{equation}
ensures the probability statement
\begin{equation}
    \Prob\left(\norm{e^{iHt} - S_l(t/r)^r} \geq \epsilon \right) < \delta.
\end{equation}
When $l$ is sufficiently large, this gate complexity scales in $n$ as
\begin{equation}
    G_{l\gg2} \sim
    \begin{cases}
        n^{k+\frac{1}{2}}\mathcal{J}t & \text{if $k$ is even},\\
        n^{k+1}\mathcal{J}t & \text{if $k$ is odd}.
    \end{cases}
\end{equation}
\end{corr}

\begin{proof}
To derive the $l$-th order complexity considering the probability statement in \eqref{eq:concineq}, we first bound the $l$-th Trotter error using Theorem \ref{thm:SYKRPG}:
\begin{align}
    \opnorm{e^{iHt} - S_l(t/r)^r}_p &\leq \opnorm{I}_p \cdot p \cdot \lambda(p,r),
\end{align}
where 
\begin{equation}
    \lambda(p,r) = \frac{1}{\sqrt{p}}\mathcal{C}(l)\sigma Q(n,k)^{-\frac{1}{2}} t\left( \binom{n}{k} \left[\sqrt{p} \sigma \sqrt{Q(n,k)}\frac{t}{r}\right]^{l}  +  \binom{n}{k}^2\left[\sqrt{p} \sigma\sqrt{Q(n,k)}\frac{t}{r}\right]^{l+1} \right).
\end{equation}
This gives us an explicit expression for the concentration inequality~\eqref{eq: actual_concineq}. By Lemma~\ref{lem:concineq}, finding the Trotter number $r$ that ensures the probability statement~\eqref{eq:concineq} is equivalent to finding an $r$ that solves the inequality
\begin{equation}
   e\sqrt{p_D}\mathcal{C}(l)\sigma Q(n,k)^{-\frac{1}{2}} \frac{t}{\epsilon}\left( \binom{n}{k} \left[\sqrt{p_D} \sigma \sqrt{Q(n,k)}\frac{t}{r}\right]^{l}  +  \binom{n}{k}^2\left[\sqrt{p_D} \sigma\sqrt{Q(n,k)}\frac{t}{r}\right]^{l+1} \right)\leq 1,
\end{equation}
where $p_D := \log(e^2D/\delta)$ with $D := \opnorm{I}_p^p.$

Hence,
\begin{align}
    r \geq \sqrt{\log(e^2D/\delta)} \sigma \sqrt{Q(n,k)}t&\cdot \left(\left[2e\sqrt{\log(e^2D/\delta)}\mathcal{C}(l)\binom{n}{k}\sigma Q(n,k)^{-\frac{1}{2}} \frac{t}{\epsilon}\right]^{\frac{1}{l}} \right. \\
    &\;+ \left.\left[2e\sqrt{\log(e^2D/\delta)}\mathcal{C}(l)\binom{n}{k}^2\sigma Q(n,k)^{-\frac{1}{2}} \frac{t}{\epsilon}\right]^{\frac{1}{l+1}}\right). \nonumber
\end{align}
Recall that the norm $\opnorm{I}_p$ should be interpreted as the Schatten $p$-norm on $\mathrm{GL}(n/2, \C)$ because we have (even) $n$ Majoranas. Hence, $D=\opnorm{I}_p^p = 2^{\frac{n}{2}}.$ Substituting $D = 2^{n/2}$ and multiply $r$ with the number of local terms $\Gamma$ and the number of stages $\Upsilon = 2\cdot 5^{\frac{l}{2}- 1}$, the $l$-th order gate complexity scales with
\begin{align}
    G & = \tilde{\mathcal{O}}\left(\Upsilon \sqrt{\frac{\log(2)}{2}n + \log(e^2/\delta)} \sigma \binom{n}{k}\sqrt{Q(n,k)}t\cdot \left(\left[2e\sqrt{\frac{\log(2)}{2}n + \log(e^2/\delta)}\mathcal{C}(l)\binom{n}{k}\sigma Q(n,k)^{-\frac{1}{2}} \frac{t}{\epsilon}\right]^{\frac{1}{l}} \right. \right.\\
    &\qquad \qquad \qquad \qquad \qquad \qquad \qquad \quad\;+ \left.\left.\left[2e\sqrt{\frac{\log(2)}{2}n + \log(e^2/\delta)}\mathcal{C}(l)\binom{n}{k}^2\sigma Q(n,k)^{-\frac{1}{2}} \frac{t}{\epsilon}\right]^{\frac{1}{l+1}}\right)\right). \nonumber
\end{align}
For fixed $k$ and $l$, the $k$ and $l$ dependent factors can be removed from the big-$\tilde{\mathcal{O}}$ notation. Asymptotically, we have $\binom{n}{k} = \mathcal{O}(n^k)$, and $Q(n,k) = \mathcal{O}(n^{k-1})$ for even $k$ while $Q(n,k) = \mathcal{O}(n^k)$ for odd $k$. By writing out $\sigma = \mathcal{J}\sqrt{\frac{(k-1)!}{kn^{k-1}}}$, we can write $G$ asymptotically as
\begin{equation}
     G = 
        \tilde{\mathcal{O}}\left[ n^{g_1(k)}\sqrt{n + \log(e^2/\delta)} \mathcal{J}t\cdot \left(\left[n^{g_2(k)}\sqrt{n + \log(e^2/\delta)}\frac{\mathcal{J}t}{\epsilon}\right]^{\frac{1}{l}} + \left[n^{g_3(k)}\sqrt{n+\log(e^2/\delta)} \frac{\mathcal{J}t}{\epsilon}\right]^{\frac{1}{l+1}}\right) \right],
\end{equation}
where $g_1(k)$, $g_2(k)$, $g_3(k)$ are defined in equation~\eqref{eq:3gs}.
This concludes the proof.
\end{proof}

Similar to Section~\ref{sec:FOE}, we can consider the probability statement in equation~\eqref{eq:concineqpsi}, and apply Lemma~\ref{lem:concineqpsi} instead of Lemma~\ref{lem:concineq}, which leads us to the following result:

\begin{corr}
    For an arbitrary fixed input state $\ket{\psi}$, simulating an even $k$-local SYK model with the $l$-th order product formula (for even $l\geq 2$), the gate complexity
    \begin{equation}
     G = 
        \tilde{\mathcal{O}}\left[ n^{g_1(k)}\sqrt{\log(e^2/\delta)} \mathcal{J}t\cdot \left(\left[n^{g_2(k)}\sqrt{\log(e^2/\delta)}\frac{\mathcal{J}t}{\epsilon}\right]^{\frac{1}{l}} + \left[n^{g_3(k)}\sqrt{\log(e^2/\delta)} \frac{\mathcal{J}t}{\epsilon}\right]^{\frac{1}{l+1}}\right) \right]
\end{equation}
where
\begin{equation}
    \left(g_1(k),\; g_2(k),\; g_3(k)\right) = \begin{cases}
        \left(k,\; 1,\; k+1\right) & \text{if $k$ is even},\\
        \left(k + \frac{1}{2},\; \frac{1}{2},\; k + \frac{1}{2}\right) & \text{if $k$ is odd}
    \end{cases}
\end{equation}
ensures the probability statement
\begin{equation}
    \Prob\left(\norm{(e^{iHt} - S_l(t/r)^r)\ket{\psi}}_{l^2} \geq \epsilon \right) < \delta.
\end{equation}
When $l$ is sufficiently large, this gate complexity scales in $n$ as
\begin{equation}
   G \sim \begin{cases}
        n^{k}\mathcal{J} t & \text{if $k$ is even},\\
        n^{k + \frac{1}{2}} \mathcal{J}t & \text{if $k$ is odd}.
    \end{cases}
\end{equation}
\end{corr}


\section{Higher-order Trotter error of the sparsified SYK model}
\label{sec:SSYK}
Our methods for evaluating higher-order Trotter errors also work in the case of the \emph{sparse SYK model} which is defined as the regular SYK model~\eqref{eq:SYKdef} but with an additional Bernoulli variable $B_{i_1,\dotsc, i_k} \in \set{0,1}$ attached to each term:
\begin{equation}
    H_{\mathrm{SSYK}} = \sum_{1\leq i_1<\cdots<i_k\leq n}B_{i_1,\dotsc, i_k}J_{i_1,\dotsc,i_k}\chi_{i_1}\cdots\chi_{i_k}.
\end{equation}
Similar to our treatment for the regular SYK model in equation~\eqref{eq:HSYK}, we write the sparse SYK model in the local Hamiltonian form with an arbitrary ordering $\gamma$
\begin{equation}
    H_{\mathrm{SSYK}} = \sum_{i = 1}^{\Gamma} B_{i} H_{\gamma(i)}.
\end{equation}

The sparse SYK model has two sources of randomness, one from the Gaussian variables $J_{\gamma(i)}$ within $H_{\gamma(i)}$ as in the regular SYK model, the other from the Bernoulli variables $B_{i}$. We pose that the Bernoulli variables are independent of each other and also independent of the Gaussian variables $J_{\gamma(i)}$. For simplicity, each Bernoulli variable $B_i$ is equal to $1$ with probability $p_B$ and equal to $0$ with probability $1-p_B$ for some $p_B \in [0,1]$. Usually, we set
\begin{equation}
    p_B = \frac{\kappa n}{\Gamma}
\end{equation}
where $\kappa \in \R_{\geq 0}$ is the degree of sparsity, and relative to equation~\eqref{eq:mean and variance of J} we renormalize the variance of the Gaussian variables $J_{\gamma(i)}$ to 
\begin{equation}
    \label{eq:new sigma}
    \sigma^2 = \frac{1}{p_B}\frac{(k-1)! \mathcal{J}^2}{kn^{k-1}}.
\end{equation}

When evaluating the expected Schatten norm of the Trotter error of the sparse SYK model, we will deal with these two types of randomness separately. We will write $\bm{B}$ to denote the Bernoulli variables $B_1,\dotsc, B_{\Gamma}$. When we want to specify that this series of random variables takes a certain configuration $\bm{b}$, we will write $\bm{B} = \bm{b}$ where $\bm{b}$ is a fixed vector in $\{0,1\}^{\Gamma}$. This is equivalent to $B_1 = b_1,$ $B_2 = b_2$, $\dotsc,$ $B_{\Gamma} = b_{\Gamma}.$
Using this notation, we can write the Trotter error given a specific sample $\bm{b}\in \{0,1\}^\Gamma$ as 
\begin{equation}\label{eq:D(B)}
    \mathcal{D}(\bm{B}=\bm{b}) := e^{iH_{\mathrm{SSYK}}(\bm{B}=\bm{b})t} - S_l(t/r)^r,
\end{equation}
indicating that the Trotter error depends on the specific sample $\bm{b}$ of the Bernoulli variables.

\subsection{Average Trotter error}

Let $b := |\bm{b}| = \sum_{i=1}^{\Gamma} b_i$ denote the number of ones in the sample $\bm{b}$ of the Bernoulli variables $\bm{B}$. When the Bernoulli variables take the value of this sample, we get a configuration of the sparse SYK model where we treat the sample values as coefficients:
\begin{equation}
    H_{\SSYK}(\bm{B} = \bm{b}) = \sum_{i=1}^{\Gamma} b_iH_{\gamma(i)},
\end{equation}
which contains exactly $b$ terms. Of course, this configuration is still a $k$-local Hamiltonian, although it becomes more difficult to count the anti-commuting terms, since the all-to-all connected symmetry in the full SYK model is broken. Therefore, we have to resort to average statements about both the Trotter error and gate complexity of these sparse models. The following theorem bounds the Trotter error averaged over the Bernoulli variables, defined as\footnote{For the rest of this work, we will use $\langle\cdot\rangle$ to denote averaging over the Bernoulli variables.}
\begin{equation}
\label{eq:avD(B)}
    \Langle\opnorm{\mathcal{D}(\bm{B})}_{*,p}\Rangle := \Langle\opnorm{e^{iH_{\mathrm{SSYK}}(\bm{B})t} - S_l(t/r)^r}_{*,p}\Rangle,
\end{equation}
where we use $\langle \cdot \rangle$ to denote averaging over the Bernoulli variables, while the norm $\opnorm{\cdot}_{*,p}$ only concerns the Gaussian variables and treats the Bernoulli variables as coefficients.
\begin{theo}
\label{thm:avgSSYKerror}
    Consider the sparse SYK model $H_{\mathrm{SSYK}}(\bm{B})$ with Bernoulli variables $\bm{B}=(B_1,\dotsc,B_{\Gamma}$ such that $\Prob(B_i = 1)=p_B$ for all $1\leq i\leq \Gamma$. Given $l\geq 2$ even, if $p_BQ(n,k)\geq 1$, the $l$-th order average Trotter error~\eqref{eq:avD(B)} is bounded by
    \begin{equation}
\label{eq:avTrotter1}
    \begin{aligned}
        \frac{\Langle\opnorm{\mathcal{D}(\bm{B})}_{*,p}\Rangle }{\opnorm{I}_{*,p}} \leq \beta(l)\frac{\Gamma \sqrt{p}\sigma\sqrt{p_B} t}{\sqrt{Q(n,k)}}\cdot\Bigg( \left[\sqrt{p} \sigma\sqrt{p_BQ(n,k)}\frac{t}{r}\right]^{l} 
        + \Gamma \cdot \left[\sqrt{p}\sigma\sqrt{p_BQ(n,k)}\frac{t}{r}\right]^{l+1}\Bigg),
    \end{aligned}
\end{equation}
    where $p \geq 2$ describes the norm, $\sigma$ is the variance in~\eqref{eq:new sigma}, $\beta(l)$ is a purely $l$-dependent factor, and $Q(n,k)$ is given in equation \eqref{eq:Q(n,k)}. If $p_BQ(n,k)\leq 1$ instead, the $l$-th order average Trotter error~\eqref{eq:avD(B)} is bounded by
    \begin{equation}
\label{eq:avTrotter2}
    \begin{aligned}
         \frac{\Langle\opnorm{\mathcal{D}(\bm{B})}_{*,p}\Rangle }{\opnorm{I}_{*,p}} \leq\beta(l)\frac{\Gamma \sqrt{p}\sigma t}{Q(n,k)}\cdot\Bigg( \left[\sqrt{p} \sigma\frac{t}{r}\right]^{l} 
        + \Gamma \cdot \left[\sqrt{p}\sigma\frac{t}{r}\right]^{l+1}\Bigg).
    \end{aligned}
\end{equation}
\end{theo}

\begin{proof}
Let us first compute the error generator $\opnorm{\mathcal{E}_g(\bm{B}=\bm{b})}_{*,p}$ given a specific configuration $\bm{b}\in \{0,1\}^{\Gamma}:$
\begin{equation}
    \opnorm{\mathcal{E}_g(\bm{B}=\bm{b})}_{*,p} = \opnorm[\Bigg]{\sum_{j = 1}^J\sum_{g_J+\cdots+g_{j+1} = g-1}\prod_{i=j}^{J}b_{i}^{g_{i}} \Lag^{g_J}_{\gamma(J)}\cdots \Lag^{g_{j+1}}_{\gamma(j+1)}[H_{\gamma(j)}]\frac{t^{g-1}}{g_J!\cdots g_{j+1}!}}_{*,p}
\end{equation}
where we have included the product of coefficients $b_i$ from each local term. However, since each $b_i\in \{0,1\}$, it suffices to consider each of them to be linear and we can remove their exponent:
\begin{equation}
    \opnorm{\mathcal{E}_g(\bm{B}=\bm{b})}_{*,p} = \opnorm[\Bigg]{\sum_{j = 1}^J\sum_{g_J+\cdots+g_{j+1} = g-1} \prod_{i=j}^{J}b_{i} \Lag^{g_J}_{\gamma(J)}\cdots \Lag^{g_{j+1}}_{\gamma(j+1)}[H_{\gamma(j)}]\frac{t^{g-1}}{g_J!\cdots g_{j+1}!}}_{*,p}.
\end{equation}
Similar to Section \ref{sec:HOSYK}, we consider the symmetrization 
\begin{equation}
    \opnorm{\mathcal{E}_g(\bm{B}=\bm{b})}_{*,p} = t^{g-1}\opnorm[\Bigg]{\sum_{j_{g-1}=1}^J \cdots \sum_{j_0=1}^J \tilde{\theta}_{\bm{j}} \prod_{i=0}^{g-1} b_{j_i} \Lag_{\gamma(j_{g-1})}\cdots\Lag_{\gamma(j_1)}[H_{\gamma(j_0)}]}_{*,p},
\end{equation}
where we have included extra coefficients $\tilde{\theta}_{\bm{j}}\in [0,1]$ that reduces this expression exactly to the previous one. As $J=\Upsilon \cdot \Gamma$, we get the following upper bound 
\begin{equation}
    \opnorm{\mathcal{E}_g(\bm{B}=\bm{b})}_{*,p} =\Upsilon ^gt^{g-1}\opnorm[\Bigg]{\sum_{j_{g-1}=1}^\Gamma \cdots \sum_{j_0=1}^\Gamma \tilde{\theta}_{\bm{j}} \prod_{i=0}^{g-1} b_{j_i} \Lag_{\gamma(j_{g-1})}\cdots\Lag_{\gamma(j_1)}[H_{\gamma(j_0)}]}_{*,p}
\end{equation}
thanks to the periodicity of the ordering map $\gamma$ in equation \eqref{eq:labeling}. 

Keeping the coefficients $b_i$ and following our discussion in Section \ref{subsec:randommatpol}, we can write the sum of commutator chains as the following random matrix polynomial:
\begin{align}
    F^{(\bm{b})}_g(H_{\gamma(\Gamma)},\dotsc,H_{\gamma(1)}) &:= \sum_{j_{g-1}=1}^\Gamma \cdots \sum_{j_0=1}^\Gamma \tilde{\theta}_{\bm{j}} \prod_{i=0}^{g-1} b_{j_i} \Lag_{\gamma(j_{g-1})}\cdots\Lag_{\gamma(j_1)}[H_{\gamma(j_0)}]\\
    &=: \sum_{j_{g-1}=1}^\Gamma \cdots \sum_{j_0=1}^\Gamma T^{(\bm{b})}_{\bm{j}}H_{\gamma(j_{g-1})} \cdots H_{\gamma(j_0)},
\end{align}
where the coefficients $T^{(\bm{b})}_{\bm{j}}$ are defined as 
\begin{equation}
    T^{(\bm{b})}_{\bm{j}} := T_{\bm{j}} \cdot \prod_{i=0}^{g-1} b_{j_i} 
\end{equation}
with $T_{\bm{j}}$ defined in equation \eqref{eq:Tj def}. Hence, these coefficients are bounded by 
\begin{equation}
\label{eq: Tb bound}
    \left|T^{(\bm{b})}_{\bm{j}}\right| \leq g^g \ind(\bm{j}) \cdot \prod_{i=0}^{g-1} b_{j_i},
\end{equation}
where the indicator is defined exactly the same as before for vectors $\bm{j}\in \{1,\dotsc,\Gamma\}^{g}$
\begin{equation}
    \ind(\bm{j}) =  \begin{cases}
        1&\quad \text{if }\Lag_{\gamma(j_{g-1})}\cdots \Lag_{\gamma(j_1)}[H_{\gamma(j_0)}]\neq 0,\\
        0&\quad \text{else.}
    \end{cases}
\end{equation}

By standardizing the local terms $H_{\gamma(i)}$ as in Section \ref{sec:HOSYK} and apply Theorem \ref{thm:RPG} with $p\geq 2$ and $m=\Gamma$, we get
\begin{equation}
    \opnorm{F_g^{(\bm{b})}}_{*,p}^2 \leq C^g_p g^{g} \sum_{\bm{w}:|\bm{w}|\leq g}\left(\sum_{\pi \in \mathcal{S}_g} \sum_{\bm{v}:|\bm{w}|+2|\bm{v}|=g}\left|T^{(\bm{b})}_{\pi[\bm{\eta}(\bm{w}, 2\bm{v})]}\right| \prod_{i=1}^{\Gamma}\sigma^{w_i + 2v_i}\right)^2\opnorm{I}_{*,p}^2,
\end{equation}
where $\bm{w},\bm{v}\in \{0,\dotsc,g\}^{\Gamma}$ such that $|\bm{w}|+2|\bm{v}|=g$. The coefficients $T^{(\bm{b})}_{\pi[\bm{\eta}(\bm{w}, 2\bm{v})]}$ are inherited from the corresponding $T^{(\bm{b})}_{\bm{j}}$ via the unary encoding function $\bm{\eta}$ defined in equation \eqref{eq:eta} by replacing $\bm{j}$ with $\pi[\bm{\eta}(\bm{w}, 2\bm{v})]$. By equation \eqref{eq: Tb bound}, 
\begin{equation}
    \left|T^{(\bm{b})}_{\pi[\bm{\eta}(\bm{w}+2\bm{v})]}\right| \leq g^g \ind(\pi[\bm{\eta}(\bm{w}, 2\bm{v})]) \cdot \prod_{i\in \Supp(\bm{w}+ 2\bm{v})} b_{i}.
\end{equation}
Since $|\bm{w}| +2|\bm{v}| = g$ by construction, $\prod_{i=1}^{\Gamma} \sigma^{w_i + 2v_i} = \sigma^g$ and hence
\begin{equation}
    \opnorm{F_g^{(\bm{b})}}_{*,p}^2 \leq C^g_p g^{3g}\sigma^{2g} \sum_{\bm{w}:|\bm{w}|\leq g}\left(\sum_{\pi \in \mathcal{S}_g} \sum_{\bm{v}:|\bm{w}|+2|\bm{v}|=g} \prod_{i\in \Supp(\bm{w}+2\bm{v})} b_{i} \ind(\pi[\bm{\eta}(\bm{w}, 2\bm{v})])\right)^2\opnorm{I}_{*,p}^2.
\end{equation}
Because $\bm{w}$ and $\bm{v}$ do not necessarily have complementary support, we consider the following splitting by moving all coefficients $b_i$ for $i\in \Supp(\bm{w})$ outside of the squared parenthesis. The remaing coefficients $b_j$ are now taken with $j\in \Supp(\bm{v})\backslash \Supp(\bm{w}).$ This gives us 
\begin{equation}
    \opnorm{F_g^{(\bm{b})}}_{*,p}^2 \leq C^g_p g^{3g}\sigma^{2g} \sum_{\bm{w}:|\bm{w}|\leq g}\prod_{i\in \Supp(\bm{w})}b_i\left(\sum_{\pi \in \mathcal{S}_g} \sum_{\bm{v}:|\bm{w}|+2|\bm{v}|=g} \prod_{j\in \Supp(\bm{v})\backslash \Supp{\bm{w}}} b_{i} \ind(\pi[\bm{\eta}(\bm{w}, 2\bm{v})])\right)^2\opnorm{I}_{*,p}^2.
\end{equation}
By splitting $\sum_{\bm{w}:|\bm{w}|\leq g}$ into $\sum_{w=0}^g\sum_{\bm{w}:|\bm{w}|=w}$ and taking square roots on both sides, we arrive at 
\begin{equation}
    \opnorm{F_g^{(\bm{b})}}_{*,p}\leq C^{\frac{g}{2}}_p g^{3g/2}\sigma^{g}\sqrt{\sum_{w=0}^g G_w(H_{\SSYK}(\bm{B}=\bm{b}))}\opnorm{I}_{*,p}.
\end{equation}
where 
\begin{equation}
    G_w(H_{\SSYK}(\bm{B}=\bm{b})) := \sum_{\bm{w}:|\bm{w}|=w}\prod_{i\in \Supp(\bm{w})}b_i\left(\sum_{\pi \in \mathcal{S}_g} \sum_{\bm{v}:|\bm{w}|+2|\bm{v}|=g} \prod_{j\in \Supp(\bm{v})\backslash \Supp(\bm{w})} b_{j} \ind(\pi[\bm{\eta}(\bm{w}, 2\bm{v})])\right)^2.
\end{equation}
We will attempt to evaluate $G_w$ explicitly later. For now, let us keep $G_w$ and write the $g$-th order generator as 
\begin{equation}
    \opnorm{\mathcal{E}_g(\bm{B}=\bm{b})}_{*,p} \leq \Upsilon^g t^{g-1} C^{\frac{g}{2}}_p g^{3g/2}\sigma^{g}\sqrt{\sum_{w=0}^g G_w(H_{\SSYK}(\bm{B}=\bm{b}))}\opnorm{I}_{*,p}.
\end{equation}
Similarly, the residual error is bounded by 
\begin{equation}
    \opnorm{\mathcal{E}_{\geq g'}(\bm{B}=\bm{b})}_{*,p} \leq \Gamma \Upsilon^{g'+1}t^{g'-1}C^{\frac{g'}{2}}_p (g')^{3g'/2}\sigma^{g'}\sqrt{\sum_{w=0}^{g'} G_w(H_{\SSYK}(\bm{B}=\bm{b}))}\opnorm{I}_{*,p}.
\end{equation}
Hence, for $g'=l+2$ and use $C_p<p$, the $l$-th order error generator is bounded by 
\begin{equation}
    \begin{aligned}
        \opnorm{\mathcal{E}(\bm{B}=\bm{b})}_{*,p} &\leq \sum_{g=l+1}^{g'-1} \opnorm{\mathcal{E}_g(\bm{B}=\bm{b})}_{*,p} + \opnorm{\mathcal{E}_{\geq g'}(\bm{B}=\bm{b})}_{*,p}\\
        &\leq \opnorm{I}_{*,p} \Bigg( \Upsilon^{l+1} t^{l} p^{\frac{l+1}{2}} (l+1)^{3(l+1)/2}\sigma^{l+1}\sqrt{\sum_{w=0}^{l+1} G_w(H_{\SSYK}(\bm{B}=\bm{b}))}\\
        &\quad+ \Gamma \Upsilon^{l+3}t^{l+1}p^{\frac{l+2}{2}}(l+2)^{3(l+2)/2}\sigma^{l+2}\sqrt{\sum_{w=0}^{l+2} G_w(H_{\SSYK}(\bm{B}=\bm{b}))}\Bigg).
    \end{aligned}
\end{equation}
Then, as we did in Section \ref{sec:HOSYK}, upper bounding by collecting common prefactors, integrating over time from $0$ to $t/r$, and multiplying by $r$ gives us the Trotter error:
\begin{equation}
\label{eq: sparse Trotter b}
    \begin{aligned}
        \opnorm{\mathcal{D}(\bm{B}=\bm{b})}_{*,p} \leq\opnorm{I}_{*,p} \alpha(l)\sqrt{p}\sigma t\cdot\Bigg( &\left[\sqrt{p} \sigma\frac{t}{r}\right]^{l} \sqrt{\sum_{w=0}^{l+1} G_w(H_{\SSYK}(\bm{B}=\bm{b}))}\\
        + \Gamma \cdot &\left[\sqrt{p}\sigma\frac{t}{r}\right]^{l+1}\sqrt{\sum_{w=0}^{l+2} G_w(H_{\SSYK}(\bm{B}=\bm{b}))}\Bigg)
    \end{aligned}
\end{equation}
where $\alpha(l)$ is a purely $l$-dependent factor given by 
\begin{equation}
    \alpha(l) := \frac{\Upsilon^{l+3}(l+2)^{3(l+2)/2}}{l+1}.
\end{equation}
Equation \eqref{eq: sparse Trotter b} gives the $l$-th order Trotter error of a given configuration of the sparse SYK model. However, it is difficult to parse this expression as evaluating $G_w(H_{\SSYK}(\bm{B} = \bm{b}))$ for specific sample $\bm{b}$ is quite complicated, but giving an average statement is tractable thanks to the following lemma:
\begin{restatable}{lemma}{lemmaavGwH}
\label{lem:avGw}
    Let $H = \sum_{i=1}^m  B_iH_{\gamma(i)}$ be a sparse Hamiltonian of $m$ terms with arbitrary ordering $\gamma$, such that for all $1\leq i<j\leq m$, $H_{\gamma(i)}$ and $H_{\gamma(j)}$ either commute or anti-commute. The $B_i$'s are independent Bernoulli variables with $\Prob(B_i = 1) = p_B$ for all $i$. Fix $g \geq 1$ and $w \in \set{0,\dotsc,g}$, and consider the average $H$-dependent sum over all samples of the Bernoulli variables
    \begin{equation}
        \Langle G_w(H(\bm{B}))\Rangle = \left\langle \sum_{\bm{w}:|\bm{w}|=w}\prod_{i\in\Supp(\bm{w})}B_i\left(\sum_{\pi\in \mathcal{S}_g}\sum_{\bm{v}:|\bm{w}| + 2|\bm{v}| = g}\prod_{j\in\Supp(\bm{v})}\!\!\!\!\!B_j\ind(\pi[\bm{\eta}(\bm{w} , 2\bm{v})])\right)^2\right\rangle,
    \end{equation}
    where $\bm{w},\bm{v} \in \{0,\dotsc,g\}^m$ such that $|\bm{w}|=w$ and $|\bm{w}| + 2|\bm{v}|=g$. The permutation $\pi\in \mathcal{S}_g$ reorders the entries of vectors in $[m]^g$. The unary encoding $\bm{\eta} $ and the indicator $\ind$ are defined on vectors $\bm{a},\bm{b}\in \{0,\dotsc,g\}^m$ and $\bm{j}\in [m]^g$ as follows:
  \begin{equation}
        \bm{\eta}(\bm{a}, \bm{b}) := (\underbrace{1,\dotsc,1}_{a_1\text{ times}}, \underbrace{2,\dotsc,2}_{a_2\text{ times}}, \dotsc, \underbrace{m,\dotsc,m}_{a_m\text{ times}},\underbrace{1,\dotsc,1}_{b_1\text{ times}}, \underbrace{2,\dotsc,2}_{b_2\text{ times}}, \dotsc, \underbrace{m,\dotsc,m}_{b_m\text{ times}}),
    \end{equation}
    and
    \begin{equation}
        \ind(\bm{j}) := \begin{cases}
        1 & \text{if }\Lag_{\gamma(j_{g})}\cdots \Lag_{\gamma(j_2)}[H_{\gamma(j_1)}]\neq 0,\\
        0 & \text{else}.
    \end{cases}
    \end{equation}
    Then
    \begin{equation}
        \Langle G_w(H(\bm{B}))\Rangle \leq g^{4g} m^2 Q_{\max}(H)^{-2}\sum_{s=0}^w \sum_{q,q'=0}^{(g-w)/2}\sum_{c=0}^{\min(q,q')} (p_B Q_{\max}(H))^{s+q+q'-c}.
    \end{equation}
    where
    \begin{equation}
    \label{eq:Qmax 2}
        Q_{\max}(H) := \max_{1\leq i\leq m} \left|\left\{j \in [m]: [H_{\gamma(j)},H_{\gamma(i)}]\neq 0\right\}\right|.
    \end{equation}
\end{restatable}
To get some intuition for this lemma, recall that for the full Hamiltonian, we could upper-bound $G_w(H)$ by 
\begin{equation}
    G_w(H) \leq g^{3g-2}m^2 Q_{\max}(H)^{g-2}
\end{equation}
in Lemma \ref{lem:Gw}. Since we perform the sparsification by attaching i$.$i$.$d$.$ Bernoulli variables to the summands of the Hamiltonian, and the average of each Bernoulli variable is $p_B$, we could expect the average $\Langle G_w(H(\bm{B}))\Rangle$ to scale with some polynomial (of utmost degree $g$)
\begin{equation}
    \Langle G_w(H(\bm{B}))\Rangle \sim \sum_sa_s(p_BQ_{\max}(H))^{s}. 
\end{equation}
To illustrate this, let us naively forget the sum of squares in $G_w$ and pretend it scales with the following sum of commutator chains (each commutator chain has weight 1):
\begin{equation}
    G_w(H(\bm{B})) \sim \sum_\pi\sum_{i_1,\dotsc,i_g}\left(\prod_{j=1}^g B_{i_j}\right)\ind(\pi(i_1,\dotsc,i_g)).
\end{equation}
If we take the expectation now, we get 
\begin{equation}
    \Langle G_w(H(\bm{B}))\Rangle \sim \sum_\pi\sum_{i_1,\dotsc,i_g}\Langle\prod_{j=1}^g B_{i_j}\Rangle\ind(\pi(i_1,\dotsc,i_g)).
\end{equation}
Note that the average of the product $\prod_{j=1}^g B_i$ is not $p_B^g$, because we do not know how many distinct elements there are in the product. Therefore, we need to write our summation more carefully to control the distinct elements: 
\begin{equation}
    G_w(H(\bm{B})) \sim \sum_\pi\sum_{s=1}^{g}\sum_{\phi: \text{ surj}}\sum_{\text{distinct }i_1,\dotsc, i_s} \left(\prod_{j=1}^s B_{i_j}\right)\ind(\pi(i_{\phi(1)},\dotsc,i_{\phi(g)})),
\end{equation}
where we now sum over $s$ distinct elements $i_1,\dotsc, i_s$. We use surjective functions $\phi: \{1,\dotsc, g\} \to \{1,\dotsc, s\}$ to distribute $i_1,\dotsc, i_s$ over the indicator. In this way, the average becomes:
\begin{equation}
    \Langle G_w(H(\bm{B})) \Rangle \sim \sum_\pi\sum_{s=1}^{g}\sum_{\phi: \text{ surj}}\sum_{\text{distinct }i_1,\dotsc, i_s} p_B^s\ind(\pi(i_{\phi(1)},\dotsc,i_{\phi(g)})),
\end{equation}
The right-hand side is upper-bounded by 
\begin{equation}
\label{eq: naivebound}
     \sum_\pi\sum_{s=1}^{g}\sum_{\phi: \text{ surj}}\sum_{\text{distinct }i_1,\dotsc, i_s} p_B^s\ind(\pi(i_{\phi(1)},\dotsc,i_{\phi(g)})) \leq g!\cdot g^g mQ_{\max}(H)^{-1}\sum_{s=1}^g (p_B Q_{\max}(H))^{s}
\end{equation}
where $m$ denotes the total number of terms in the Hamiltonian. We use $g!$ to bound the number of permutations $\pi \in \mathcal{S}_g$ and $g^g$ to bound the number of surjective functions from $\{1,\dotsc, g\} \to \{1,\dotsc, s\}$. The above bound in equation \eqref{eq: naivebound} is indeed a polynomial as expected, and resembles the bound in Lemma \ref{lem:avGw} apart from additional sums there. In reality, we cannot simply ignore the sum of squares in $G_w$ as it plays a crucial role in capturing the scaling behavior of the Trotter error. To get the correct bound, one needs to expand the sum of squares and treat the overlapping summations carefully. (This is taken care of by the additional summations over $x$, $q$ and $q'$ in Lemma \ref{lem:avGw}.) For details of this treatment, we refer the reader to the proof in Appendix \ref{apx:avGw}. 

To apply Lemma \ref{lem:avGw}, let us first put expectations on both side of equation \eqref{eq: sparse Trotter b} and apply Jensen's inequality:
\begin{equation}
    \begin{aligned}
        \Langle\opnorm{\mathcal{D}(\bm{B})}_{*,p}\Rangle \leq\opnorm{I}_{*,p} \alpha(l)\sqrt{p}\sigma t\cdot\Bigg( &\left[\sqrt{p} \sigma\frac{t}{r}\right]^{l} \sqrt{\sum_{w=0}^{l+1} \Langle G_w(H_{\SSYK}(\bm{B}))\Rangle}\\
        + \Gamma \cdot &\left[\sqrt{p}\sigma\frac{t}{r}\right]^{l+1}\sqrt{\sum_{w=0}^{l+2} \Langle G_w(H_{\SSYK}(\bm{B}))\Rangle}\Bigg).
    \end{aligned}
\end{equation}
For the SYK model, we have $Q_{\max} = Q(n,k)$, and therefore, by Lemma \ref{lem:avGw}, we get 
\begin{equation}
    \begin{aligned}
        \Langle\opnorm{\mathcal{D}(\bm{B})}_{*,p}\Rangle \leq\opnorm{I}_{*,p} \alpha'(l)\frac{\Gamma \sqrt{p}\sigma t}{Q(n,k)}\cdot\Bigg( &\left[\sqrt{p} \sigma\frac{t}{r}\right]^{l} \sqrt{\sum_{w=0}^{l+1} \sum_{s=0}^w \sum_{q,q'=0}^{(l+1-w)/2} \sum_{c=0}^{\min(q,q')}\left(p_BQ(n,k)\right)^{s+q+q'-c}}\\
        + \Gamma \cdot &\left[\sqrt{p}\sigma\frac{t}{r}\right]^{l+1}\sqrt{\sum_{w=0}^{l+2}\sum_{s=0}^w \sum_{q,q'=0}^{(l+2-w)/2} \sum_{c=0}^{\min(q,q')}\left(p_BQ(n,k)\right)^{s+q+q'-c}}\Bigg),
    \end{aligned}
\end{equation}
where
\begin{equation}
    \alpha'(l) := (l+2)^{2(l+2)}\cdot \frac{\Upsilon^{l+3}(l+2)^{3(l+2)/2}}{l+1}.
\end{equation}

We can simplify the expression above depending on how we tune the probability $p_B$. If $p_B Q(n,k) \geq 1$, then 
\begin{equation}
\begin{aligned}
    \sum_{s=0}^w \sum_{q,q'=0}^{(g-w)/2} \sum_{c=0}^{\min(q,q')}\left(p_BQ(n,k)\right)^{s+q+q'-c} &\leq (p_BQ(n,k))^{g}\sum_{s=0}^w \sum_{q,q'=0}^{(g-w)/2}\sum_{c=0}^{\min(q,q')}1 \\
    &\leq (p_BQ(n,k))^{g}(w+1) \left(\frac{g-w}{2} + 1\right)^3\\
    &\leq (p_BQ(n,k))^{g}\cdot (g+1)^4.
\end{aligned}
\end{equation}
Plugging these into the expression above, we get 
\begin{equation}
    \begin{aligned}
        \frac{\Langle\opnorm{\mathcal{D}(\bm{B})}_{*,p}\Rangle }{\opnorm{I}_{*,p}} \leq \beta(l)\frac{\Gamma \sqrt{p}\sigma\sqrt{p_B} t}{\sqrt{Q(n,k)}}\cdot\Bigg( \left[\sqrt{p} \sigma\sqrt{p_BQ(n,k)}\frac{t}{r}\right]^{l} 
        + \Gamma \cdot \left[\sqrt{p}\sigma\sqrt{p_BQ(n,k)}\frac{t}{r}\right]^{l+1}\Bigg),
    \end{aligned}
\end{equation}
where
\begin{equation}
    \beta(l) = (l+3)^{\frac{5}{2}}(l+2)^{2(l+2)}\cdot \frac{\Upsilon^{l+3}(l+2)^{3(l+2)/2}}{l+1}.
\end{equation}

If $p_B$ is tuned such that $p_BQ(n,k)\leq 1$, then 
\begin{equation}
\begin{aligned}
    \sum_{s=0}^w\sum_{q,q'=0}^{(g-w)/2} \sum_{c=0}^{\min(q,q')}\left(p_BQ(n,k)\right)^{s+q+q'-c} &\leq \sum_{s=0}^w \sum_{x=0}^{(g-w)/2}\sum_{q,q'=0}^{(g-w)/2}1 \\
    &\leq (g+1)^4.
\end{aligned}
\end{equation}
This gives us 
\begin{equation}
    \begin{aligned}
         \frac{\Langle\opnorm{\mathcal{D}(\bm{B})}_{*,p}\Rangle }{\opnorm{I}_{*,p}} \leq\beta(l)\frac{\Gamma \sqrt{p}\sigma t}{Q(n,k)}\cdot\Bigg( \left[\sqrt{p} \sigma\frac{t}{r}\right]^{l} 
        + \Gamma \cdot \left[\sqrt{p}\sigma\frac{t}{r}\right]^{l+1}\Bigg).
    \end{aligned}
\end{equation}
This proves the statement. 
\end{proof}

\subsection{Average gate complexity}
As we have seen in equation \eqref{eq: sparse Trotter b}, the Trotter error depends on the specific sample $\bm{b}$ of the Bernoulli variables $\bm{B}$, so will the gate complexity. If the support of $\bm{b}$ happens to have size $\Gamma$, we recover the full SYK Hamiltonian. If the support of $\bm{b}$ is much smaller than $\Gamma$, we get a Hamiltonian with much fewer terms that requires much fewer gates to simulate. It is therefore tricky to derive one gate complexity that covers all the cases. Moreover, due to the difficulty of evaluating $G_w(H_{\SSYK}(\bm{B}=\bm{b}))$ for specific sample $\bm{b}$ in equation \eqref{eq: sparse Trotter b}, it is also difficult to formulate a gate count specific to each sample. Hence, we decided on the following approach, leading us to a notion of average gate complexity. Let 
\begin{equation}
    \mathcal{G} := \{G(\bm{B}=\bm{b}): \bm{b}\in \{0,1\}^{\Gamma}\}
\end{equation}
be the set of gate complexities for all configurations $\bm{b}$, where each $G(\bm{B}=\bm{b})$ is the gate complexity of the $l$-th order product formula for approximating the configuration $H_{\SSYK}(\bm{B}=\bm{b})$ and ensuring the probability statement in equation \eqref{eq:concineq} or \eqref{eq:concineqpsi}. Each gate complexity $G(\bm{B}=\bm{b})$ is the minimal gate complexity derived from the Trotter error of the corresponding sample $\opnorm{\mathcal{D}(\bm{B}=\bm{b})}_{*,p}$. If we take the expectation of $\mathcal{G}$ over all samples Bernoulli variables, the quantity 
\begin{equation}
    \overline{G}:= \Langle G(\bm{B})\Rangle
\end{equation}
yields a measure for he average gate complexity when simulating the sparse SYK model. We shall formulate our result in the following corollary:
\begin{corr}\label{cor:511 new}
Simulating the $k$-local $\kappa$-sparse SYK model using (even) $l$-th order product formula has average gate complexity 
\begin{equation}
     \overline{G} = \tilde{\mathcal{O}}\left(n\left(\sqrt{p_BQ(n,k)}\right)^{1+\frac{2}{l}}\sqrt{n+\log(e^2/\delta)} \mathcal{J}t\cdot \Bigg(\left[\sqrt{n+\log(e^2/\delta)}\frac{\mathcal{J}t}{\epsilon}\right]^{\frac{1}{l}} 
    + \left[n^{k}\sqrt{n+\log(e^2/\delta)}\frac{\mathcal{J}t}{\epsilon}\right]^{\frac{1}{l+1}}\Bigg)\right),
\end{equation}
in the sense that $\overline{G}$ upper bounds the average of the set
\begin{equation}
    \mathcal{G} = \{G(\bm{B}=\bm{b}): \;\bm{b}\in [0,1]^{\Gamma}\},
\end{equation}
where $G(\bm{B}=\bm{b})$ is the minimal gate complexity of $S_l(t/r)^r$ approximating the configuration $H_{\mathrm{SSYK}}(\bm{B}=\bm{b})$ that ensures the probability statement
\begin{equation}
     \Prob\left(\norm{e^{iH_{\mathrm{SSYK}}(\bm{B}=\bm{b})t} - S_l(t/r)^r} \geq \epsilon \right) < \delta.
\end{equation}
The term $Q(n,k)$ is given by equation \eqref{eq:Q(n,k)}. When $l$ is sufficiently large, the average gate complexity scales in $n$ as
\begin{equation}
    \overline{G} \sim \begin{cases}
        n^{1+\frac{1}{2}} \mathcal{J}t & \text{if $k$ is even},\\
        n^2\mathcal{J}t & \text{if $k$ is odd}.
    \end{cases}
\end{equation}
This result is irrespective of the size of locality $k$, but rather its parity.
\end{corr}

\begin{proof}
Given a sample $\bm{b}\in \{0,1\}^{\Gamma}$, let us start from equation \eqref{eq: sparse Trotter b}
\begin{equation}
    \begin{aligned}
        \opnorm{\mathcal{D}(\bm{B}=\bm{b})}_{*,p} \leq\opnorm{I}_{*,p} \cdot p \cdot \lambda(p,r)
    \end{aligned}
\end{equation}
where 
\begin{equation}
    \begin{aligned}
        \lambda(p,r) := \frac{1}{\sqrt{p}}\alpha(l)\sigma t\cdot\Bigg( &\left[\sqrt{p} \sigma\frac{t}{r}\right]^{l} \sqrt{\sum_{w=0}^{l+1} G_w(H_{\SSYK}(\bm{B}=\bm{b}))}\\
        + \Gamma \cdot &\left[\sqrt{p}\sigma\frac{t}{r}\right]^{l+1}\sqrt{\sum_{w=0}^{l+2} G_w(H_{\SSYK}(\bm{B}=\bm{b}))}\Bigg).
    \end{aligned}
\end{equation}
By Lemma \eqref{lem:concineq}, finding the Trotter number $r$ to ensure the probability statement 
\begin{equation}
    \Prob(\norm{\mathcal{D}(\bm{B}=\bm{b})} \geq \epsilon) \leq \delta, 
\end{equation}
is equivalent to finding $r$ such that 
\begin{equation}
    \begin{aligned}
        1\geq  e\sqrt{p_D}\alpha(l)\sigma \frac{t}{\epsilon}\cdot\Bigg( &\left[\sqrt{p_D} \sigma\frac{t}{r}\right]^{l} \sqrt{\sum_{w=0}^{l+1} G_w(H_{\SSYK}(\bm{B}=\bm{b}))}\\
        + \Gamma \cdot &\left[\sqrt{p_D}\sigma\frac{t}{r}\right]^{l+1}\sqrt{\sum_{w=0}^{l+2} G_w(H_{\SSYK}(\bm{B}=\bm{b}))}\Bigg).
    \end{aligned}
\end{equation}
where $p_D = \log(e^2D/\delta)$ with $D = \opnorm{I}^p_p$. Hence, we find that 
\begin{equation}
\begin{aligned}
    r \geq \sqrt{p_D}\sigma t \cdot \Bigg(&\left[2e\sqrt{p_D}\alpha(l)\sigma\frac{t}{\epsilon}\sqrt{\sum_{w=0}^{l+1} G_w(H_{\SSYK}(\bm{B}=\bm{b}))}\right]^{\frac{1}{l}} \\
    + &\left[2e\sqrt{p_D}\alpha(l)\sigma\frac{\Gamma t}{\epsilon}\sqrt{\sum_{w=0}^{l+2} G_w(H_{\SSYK}(\bm{B}=\bm{b}))}\right]^{\frac{1}{l+1}}\Bigg).
\end{aligned}
\end{equation}
For simplicity, we can move the sum over $G_w$ out of the square brackets by considering a higher $r$:
\begin{equation}
\begin{aligned}
    r \geq   \left(\sqrt{\sum_{w=0}^{l+2} G_w(H_{\SSYK}(\bm{B}=\bm{b}))}\right)^{\frac{1}{l}} \cdot R(n)
\end{aligned}
\end{equation}
with $R(n)$ defined as 
\begin{equation}
    R(n) := \sqrt{p_D}\sigma t\cdot \Bigg(\left[2e\sqrt{p_D}\alpha(l)\sigma\frac{t}{\epsilon}\right]^{\frac{1}{l}} 
    + \left[2e\sqrt{p_D}\alpha(l)\sigma\frac{\Gamma t}{\epsilon}\right]^{\frac{1}{l+1}}\Bigg)
\end{equation}
which is a coefficient that appears in the Trotter number of all configurations. 

The gate complexity specific to this configuration is then upper bounded by 
\begin{equation}
    G(\bm{B} = \bm{b}) \leq |\bm{b}| \cdot r = R(n)\cdot |\bm{b}|  \left(\sqrt{\sum_{w=0}^{l+2} G_w(H_{\SSYK}(\bm{B}=\bm{b}))}\right)^{\frac{1}{l}}.
\end{equation}
It is an upper bound because we did not consider the minimal $r$. 

Averaging over all Bernoulli variables gives us
\begin{equation}
    \Langle G(\bm{B})\Rangle \leq R(n) \left\langle |\bm{B}|\left(\sqrt{\sum_{w=0}^{l+2} G_w(H_{\SSYK}(\bm{B}))}\right)^{\frac{1}{l}}\right\rangle.
\end{equation}
By the Cauchy--Schwarz inequality,  and applying Jensen's inequality for the square root term, we get
\begin{equation}
    \Langle G(\bm{B})\Rangle \leq R(n) \sqrt{\Langle |\bm{B}|^2\Rangle }\sqrt{\left(\sum_{w=0}^{l+2} \Langle G_w(H_{\SSYK}(\bm{B}))\Rangle\right)^{\frac{1}{l}}}.
\end{equation}
The first expectation can be computed from the second moment of the binomial distribution:
\begin{equation}
\begin{aligned}
    \Langle |\bm{B}|^2\Rangle &= \sum_{b=0}^{\Gamma} b^2\binom{\Gamma}{b} p_B^b (1-p_B)^{\Gamma-b}\\
    &= \Gamma(\Gamma-1)p_B^2 + \Gamma p_B.
\end{aligned}
\end{equation}
For the second expectation, we get 
\begin{equation}
    \Langle G_w(H_{\SSYK}(\bm{B}))\Rangle \leq (p_BQ(n,k))^{g}\cdot (g+1)^4
\end{equation}
if $p_BQ(n,k) \geq 1$, and
\begin{equation}
    \Langle G_w(H_{\SSYK}(\bm{B}))\Rangle \leq (g+1)^4
\end{equation}
if $p_BQ(n,k) \leq 1$. 

We can finally put everything together. If $p_B Q(n,k)\geq 1$, we get
\begin{equation}
    \Langle G(\bm{B})\Rangle \leq A(n) \sqrt{\Gamma(\Gamma-1)p_B^2 + \Gamma p_B}\left(\sqrt{p_BQ(n,k)}\right)^{1+\frac{2}{l}}
\end{equation}
where 
\begin{equation}
    A(n) := (l+3)^{\frac{5}{2l}} R(n) = (l+3)^{\frac{5}{2l}}\sqrt{p_D}\sigma t\cdot \Bigg(\left[2e\sqrt{p_D}\alpha(l)\sigma\frac{t}{\epsilon}\right]^{\frac{1}{l}} 
    + \left[2e\sqrt{p_D}\alpha(l)\sigma\frac{\Gamma t}{\epsilon}\right]^{\frac{1}{l+1}}\Bigg).
\end{equation}
If $p_BQ(n,k)\leq 1$, we get 
\begin{equation}
    \Langle G(\bm{B})\Rangle \leq A(n) \sqrt{\Gamma(\Gamma-1)p_B^2 + \Gamma p_B}.
\end{equation}
Since we want to study the behavior of the gate complexity asymptotic in $n$, while fixing other parameters such as the sparsity $\kappa$, locality $k$, and evolution time $t$, we can safely assume that 
\begin{equation}
    p_B Q(n,k) \geq 1 \quad\quad\text{and}\quad\quad(p_B \Gamma)^m \geq p_B\Gamma
\end{equation}
for all $m\geq 1$. Since $D = \opnorm{I}_p^p = 2^{n/2}$, and $\sigma=\mathcal{O}(\mathcal{J})$, suppressing all $l$- and $k$-dependent factors gives us
\begin{equation}
    A(n) = \tilde{\mathcal{O}}\left(\sqrt{n+\log(e^2/\delta)} \mathcal{J}t\cdot \Bigg(\left[\sqrt{n+\log(e^2/\delta)}\frac{\mathcal{J}t}{\epsilon}\right]^{\frac{1}{l}} 
    + \left[n^{k}\sqrt{n+\log(e^2/\delta)}\frac{\mathcal{J}t}{\epsilon}\right]^{\frac{1}{l+1}}\Bigg)\right).
\end{equation}
Using this, we get the following asymptotic average gate complexities:
\begin{equation}
    \overline{G} = \tilde{\mathcal{O}}\left(n\left(\sqrt{p_BQ(n,k)}\right)^{1+\frac{2}{l}}\sqrt{n+\log(e^2/\delta)} \mathcal{J}t\cdot \Bigg(\left[\sqrt{n+\log(e^2/\delta)}\frac{\mathcal{J}t}{\epsilon}\right]^{\frac{1}{l}} 
    + \left[n^{k}\sqrt{n+\log(e^2/\delta)}\frac{\mathcal{J}t}{\epsilon}\right]^{\frac{1}{l+1}}\Bigg)\right),
\end{equation}
As for $Q(n,k)$, it's behavior depends on the parity of locality $k$. If $k$ is even, $Q(n,k) = \mathcal{O}(n^{k-1})$, and if $k$ is odd, $Q(n,k) = \mathcal{O}(n^k).$ Suppressing all constants other than $n$, $t$ and $\mathcal{J}$, we get the following scaling for the asymptotic average gate complexities when $l$ is sufficiently large:
\begin{equation}
    \overline{G} \sim \begin{cases}
        n^{1+\frac{1}{2}} \mathcal{J}t & \text{if $k$ is even},\\
        n^2\mathcal{J}t & \text{if $k$ is odd}.
    \end{cases}
\end{equation}
This proves the statement. 

\end{proof}

Similar to Section \ref{sec:FOE}, by consider the probability statement in equation \eqref{eq:concineqpsi} and applying Lemma \ref{lem:concineqpsi} instead of Lemma \ref{lem:concineq}, we get the following result for simulating an arbitrary fixed input state:

\begin{corr}\label{cor:512 new}
For an arbitrary fixed input state $\ket{\psi}$, simulating the $k$-local $\kappa$-sparse SYK model using (even) $l$-th order product formula has average gate complexity 
\begin{equation}
     \overline{G} = \tilde{\mathcal{O}}\left(n\left(\sqrt{p_BQ(n,k)}\right)^{1+\frac{2}{l}}\sqrt{\log(e^2/\delta)} \mathcal{J}t\cdot \Bigg(\left[\sqrt{\log(e^2/\delta)}\frac{\mathcal{J}t}{\epsilon}\right]^{\frac{1}{l}} 
    + \left[n^{k}\sqrt{\log(e^2/\delta)}\frac{\mathcal{J}t}{\epsilon}\right]^{\frac{1}{l+1}}\Bigg)\right),
\end{equation}
in the sense that $\overline{G}$ upper bounds the average of the set
\begin{equation}
    \mathcal{G} = \{G(\bm{B}=\bm{b}): \;\bm{b}\in [0,1]^{\Gamma}\},
\end{equation}
where $G(\bm{B}=\bm{b})$ is the minimal gate complexity of $S_l(t/r)^r$ approximating the configuration $H_{\mathrm{SSYK}}(\bm{B}=\bm{b})$ that ensures the probability statement
\begin{equation}
     \Prob\left(\norm{\left(e^{iH_{\mathrm{SSYK}}(\bm{B}=\bm{b})t} - S_l(t/r)^r\right)\ket{\psi}}_{l^2} \geq \epsilon \right) < \delta.
\end{equation}
The term $Q(n,k)$ is given by equation \eqref{eq:Q(n,k)}. When $l$ is sufficiently large, the average gate complexity scales in $n$ as
\begin{equation}
    \overline{G} \sim \begin{cases}
        n \mathcal{J}t & \text{if $k$ is even},\\
        n^{1+\frac{1}{2}}\mathcal{J}t & \text{if $k$ is odd}.
    \end{cases}
\end{equation}
This result is irrespective of the size of locality $k$, but rather its parity.
\end{corr}

\section{Discussion}
\label{sec:Discussion}

\subsection{Gate complexities for simulating the SYK model}

We obtained bounds for various orders $l$ of the Trotter error, and estimated the gate complexity of simulating the SYK model with accuracy described by the probability statements in equations~\eqref{eq:concineq} and~\eqref{eq:concineqpsi}. 

The probability statement in equation~\eqref{eq:concineq} characterizes small Trotter error by the difference in spectral norm between the time-evolution operator and its product formula approximation. The gate complexities, displayed in Table~\ref{tab:results}, for simulating the SYK model ensure the probability statement in equation~\eqref{eq:concineq}. In Table~\ref{tab:results}, we only focus on the scaling in $n$ and $t$ while keeping all other parameters fixed. As mentioned earlier, the $\log n$ overhead from the fermion-to-qubit mapping is omitted. The results in the large $l$ limit are obtained by taking $l$ sufficiently large, such that the terms with $\frac{1}{l}$ and $\frac{1}{l+1}$ in their exponent become irrelevant to the scaling behavior. 

\begin{table}[!hbt]
\begin{adjustbox}{minipage=\paperwidth,margin=5pt,center}
\centering
\setlength\tabcolsep{9pt}
\begin{tabular}{|c|c|c| }
 \hline
 \multirow{2}*{Orders of product formulas } & \multicolumn{2}{c|}{Gate complexity $G$} \\
 \cline{2-3}
     & Even $k$ & Odd $k$ \\
 \hline \rule{0pt}{1.5\normalbaselineskip}
First-order & $n^{k + \frac{5}{2}} t^2$ & $n^{k + 3}t^2$ \\[2ex]
 \hline \rule{0pt}{1.5\normalbaselineskip}
 $l$-th order & $n^{k + \frac{1}{2}} t\cdot \left(\left[n^{\frac{3}{2}}t\right]^{\frac{1}{l}} + \left[n^{k+ \frac{3}{2}} t\right]^{\frac{1}{l+1}}\right)$ & $n^{k + 1} t\cdot \left(\left[n t\right]^{\frac{1}{l}} + \left[n^{k + 1} t\right]^{\frac{1}{l+1}}\right)$\\ [2ex]
 \hline \rule{0pt}{1.5\normalbaselineskip}
 Large $l$ limit & $n^{k + \frac{1}{2}} t$ & $n^{k + 1} t$\\[2ex]
 \hline 
\end{tabular}
\end{adjustbox}
\caption{The gate complexity of using different orders of product formulas to simulate the SYK model of locality~$k$, ensuring the probability statement in equation~\eqref{eq:concineq}, for fixed $\epsilon>0$, $\delta\in (0,1)$ and $k\leq 1$, omitting the $\log n$ overhead produced by the fermion-to-qubit mapping. We omitted the big-$\tilde{\mathcal{O}}$ notation in the cells. The $\epsilon$-scaling of these bounds go with $\sim \frac{1}{\epsilon}$ for first-order formulas and $\sim \left(\frac{1}{\epsilon}\right)^{\frac{1}{l}}$ for $l$-th order formulas.}
\label{tab:results}
\end{table}

We observe that the scaling behavior of the gate complexities is different for even $k$ compared to odd $k$. This is a direct consequence of $Q(n,k)$, see equation~\eqref{eq:Q(n,k)}, which behaves differently for even and odd $k$. As a result, the Trotter error, and subsequently the corresponding gate complexities, will have different scaling behaviors. We expect this phenomenon, since the SYK model behaves differently depending on the parity of its locality $k$. It is also reasonable that the gate complexity for simulating odd-$k$ SYK models is higher than even-$k$ SYK models. As $Q(n,k) = \mathcal{O}(n^k)$ for odd $k$, most of the terms will anti-commute in the odd-$k$ SYK model, whereas even-$k$ SYK models have $Q(n,k) = \mathcal{O}(n^{k-1})$, meaning that most of the terms commute as $n$ grows. Since the Trotter error depends on the number of anti-commuting pairs of terms, simulating the odd-$k$ SYK models will naturally require a higher gate complexity than the even-$k$ SYK models.

By ignoring the $\log n$ overhead produced by the fermion-to-qubit mapping, we found that for the SYK model with $k=4$, the gate complexity of its digital quantum simulation using first-order product formula scales with $\tilde{\mathcal{O}}(n^{6.5}t^2)$, while using higher-order product formulas (given the order $l$ is sufficiently large) it scales with $\tilde{\mathcal{O}}(n^{4.5} t)$. This offers an improvement of $\mathcal{O}(n^{3.5})$ and $\mathcal{O}(n^{5.5})$ in $n$ (up to $\log$-factors) respectively compared to the previous result $\mathcal{O}(n^{10}t^2)$ obtained by \cite{Garc_a_lvarez_2017}. Compared to the result $\mathcal{O}(n^5t^2)$ \cite{PhysRevD.109.105002}, our higher-order bound offers a slight improvement of $\mathcal{O}(\sqrt{n})$ in $n$. 

Additionally, any order of product formula simulating the SYK model requires at least $\Omega(n^k)$ gates, because there are $\mathcal{O}(n^k)$ local term exponentials per round in the product formula due to the fact that the SYK model has $\Gamma = \binom{n}{k}= \mathcal{O}(n^k)$ local terms. In this sense, one could interpret our results $\tilde{\mathcal{O}}(n^{k+1/2}t)$ and $\tilde{\mathcal{O}}(n^{k+1}t)$ in the large $l$ limit as close to optimal. 

These gate complexities can be further improved upon, if we consider the probability statement in equation \eqref{eq:concineqpsi}, which characterizes small Trotter error by the difference in the $l^2$-norm between the time evolution of an arbitrary fixed input state $\ket{\psi}$ and the state obtained by applying the product formula approximation on $\ket{\psi}$. Overall, considering this probability statement allows us to reduce the previous gate complexities by $\mathcal{O}(n^2)$ for first-order product formulas and by $\mathcal{O}(\sqrt{n})$ for higher-order formulas. These results are presented in Table~\ref{tab:results_fixed}, where we only focus on the scaling in $n$ and $t$ while keeping all other parameters fixed. The results in the large $l$ limit are obtained in the same way by taking $l$ sufficiently large, such that the terms with $\frac{1}{l}$ and $\frac{1}{l+1}$ in their exponent vanish. As we can see in Table~\ref{tab:results_fixed}, simulating the SYK model with even $k$ using higher-order formulas requires only $\tilde{\mathcal{O}}(n^kt)$ gates in this case. This means that for any arbitrary fixed input state $\ket{\psi}$, higher-order product formulas are actually optimal for simulating the time-evolution of $\ket{\psi}$ under the SYK Hamiltonian with even $k$, apart from the inevitable $\log n$ overhead from the fermion-to-qubit mapping.

\begin{table}[!hbt]
\begin{adjustbox}{minipage=\paperwidth,margin=5pt,center}
\centering
\setlength\tabcolsep{9pt}
\begin{tabular}{|c|c|c| }
 \hline
 \multirow{2}*{Orders of product formulas} & \multicolumn{2}{c|}{Gate complexity $G$} \\
 \cline{2-3}
     & Even $k$ & Odd $k$ \\
 \hline \rule{0pt}{1.5\normalbaselineskip}
First-order & $n^{k + \frac{1}{2}} t^2$ & $n^{k + 1}t^2$ \\[2ex]
 \hline \rule{0pt}{1.5\normalbaselineskip}
 $l$-th order & $n^{k} t\cdot \left(\left[nt\right]^{\frac{1}{l}} + \left[n^{k+ 1} t\right]^{\frac{1}{l+1}}\right)$ & $n^{k + \frac{1}{2}} t\cdot \left(\left[n^{\frac{1}{2}} t\right]^{\frac{1}{l}} + \left[n^{k + \frac{1}{2}} t\right]^{\frac{1}{l+1}}\right)$\\ [2ex]
 \hline \rule{0pt}{1.5\normalbaselineskip}
 Large $l$ limit & $n^{k } t$ & $n^{k + \frac{1}{2}} t$\\[2ex]
 \hline 
\end{tabular}
\end{adjustbox}
\caption{The gate complexities of using different orders of product formulas to simulate the SYK model of locality~$k$ on an arbitrary fixed input state $\ket{\psi}$, ensuring the probability statement in equation \eqref{eq:concineqpsi}, for fixed $\epsilon>0$, $\delta\in (0,1)$ and $k\leq 1$, omitting the $\log n$ overhead produced by the fermion-to-qubit mapping. We omitted the big-$\tilde{\mathcal{O}}$ notation in the cells. The $\epsilon$-scaling of these bounds go with $\sim \frac{1}{\epsilon}$ for first-order formulas and $\sim \left(\frac{1}{\epsilon}\right)^{\frac{1}{l}}$ for $l$-th order formulas.}
\label{tab:results_fixed}
\end{table}

\subsection{Gate complexities for simulating the sparse SYK model}

Additionally, using similar methods, we derived average higher-order Trotter error bounds for simulating the sparse SYK model. They lead to the following average gate complexity in the large $l$-limit\footnote{The complexities present here have an $\epsilon$-scaling that goes with $\sim \left(\frac{1}{\epsilon}\right)^{\frac{1}{l}}$ for $l$-th order product formulas.}:
\begin{equation}
    \overline{G} \sim \begin{cases}
        n^{1+\frac{1}{2}} \mathcal{J}t & \text{if $k$ is even},\\
        n^2\mathcal{J}t & \text{if $k$ is odd}.
    \end{cases}
\end{equation}
Compared to the full SYK model, this result does not depend on the locality $k$, but rather its parity. Since the sparse SYK model on average contains $\kappa n$ terms in the Hamiltonian, the gate complexity of product formulas 
 scales at least with $\kappa n$. Hence, similar to the case for the dense SYK model, one could interpret our result in the large $l$ limit as close to optimal.

Similarly, if we consider the probability statement for an arbitrary fixed input state $\ket{\psi}$ in equation \eqref{eq:concineqpsi} instead of the statement in equation \eqref{eq:concineq}, we get 
\begin{equation}
    \overline{G} \sim \begin{cases}
        n \mathcal{J}t & \text{if $k$ is even},\\
        n^{1+\frac{1}{2}}\mathcal{J}t & \text{if $k$ is odd}.
    \end{cases}
\end{equation}
This shows that for any fixed input state $\ket{\psi}$, higher-order product formulas are near-optimal on average for simulating the time evolution of $\ket{\psi}$ under the sparse SYK Hamiltonian with even $k$, apart from the inevitable $\log n$ overhead cost from the fermion-to-qubit mapping. Together with the results in the previous section, we demonstrated the potential of product formulas to efficiently simulate the SYK and sparse SYK models.

\subsection{Error bounds for general Gaussian models}
Due to the generality of Lemma~\ref{lem:Gw}, our Trotter error bounds for the SYK Hamiltonians in Theorems~\ref{thm:FOE} and~\ref{thm:SYKRPG} can be easily generalized to other Gaussian models by simply replacing $Q(n,k)$ with $Q_{\mathrm{max}}$ (equation~\eqref{eq:Qmax}) and $\binom{n}{k}$ with the total number of terms in the specific Hamiltonian. In other words, consider a Hamiltonian of the form 
\begin{equation}
\label{eq:generalGaussian}
    H = \sum_{i =1}^{\Gamma} J_i H_i,
\end{equation}
where the $J_i$'s are i.i.d.\@ Gaussian variables with zero mean and standard deviation $\sigma$, and the $H_i$'s are linearly independent Pauli strings or Pauli string representations of fermionic operators of arbitrary locality. For $p\geq 2$, the first-order normalized Trotter error is bounded by
\begin{equation}
\label{eq:generalDelta_1}
       \Delta_1(H):= 4\sqrt{2} p^2 \sigma^2 \sqrt{\Gamma Q_{\mathrm{max}}(H)}\;t^2\left[\frac{1}{2r} 
    + \sigma \sqrt{Q_{\mathrm{max}}(H)}\frac{t}{3r^2} \right],
    \end{equation}
and the $l$-th order normalized Trotter error is bounded by 
\begin{equation}
\label{eq:generalDelta_l}
          \Delta_{l}(H):= \frac{\mathcal{D}(l)\sqrt{p}  \sigma t}{\sqrt{Q_{\mathrm{max}}(H)}}\left( \Gamma \left[\sqrt{p} \sigma \sqrt{Q_{\mathrm{max}}(H)}\frac{t}{r}\right]^{l}  + \Gamma^2\left[\sqrt{p} \sigma\sqrt{Q_{\mathrm{max}}(H)}\frac{t}{r}\right]^{l+1} \right),     
     \end{equation}
where $\mathcal{D}(l)$ is a factor only depending on $l$. Hence, bounding the Trotter error of a Hamiltonian of the form~\eqref{eq:generalGaussian} reduces to the problem of counting the maximal number of anti-commuting terms in the Hamiltonian. Subsequently, one can use equations~\eqref{eq:generalDelta_1} and~\eqref{eq:generalDelta_l} to estimate the corresponding gate complexity as we did for the SYK models.

Similarly, we can generalize our results for the sparse SYK model to a broader class of sparse Hamiltonians where the sparsification is governed by Bernoulli processes similar to the sparse SYK model.  To be specific, given a full Hamiltonian 
\begin{equation}
    H_{\mathrm{full}} := \sum_{i=1}^m J_iH_i,
\end{equation}
we consider its sparsification of the form 
\begin{equation}
\label{eq:generalsparseH}
    H_{\mathrm{sparse}} = \sum_{i=1}^m B_i J_i H_i
\end{equation}
where $B_i$'s are i.i.d.\@ Bernoulli variables and $J_i$'s are i.i.d.\@ Gaussian variables, such that given constants $\kappa, \mathcal{J}\geq 0$,  
\begin{equation}
   \Prob(B_i = 1) = p_B, \quad \mathbb{E}(J_i) = 0,\quad\text{and}\quad \mathbb{E}(J_i^2) = \mathcal{O}(\mathcal{J}^2)
\end{equation}
for all $1\leq i\leq m$. The $H_i$'s are linearly independent Pauli strings (or can be represented by Pauli strings, such as the Majoranas) of arbitrary locality. Then, by simply replacing $Q(n,k)$ with $Q_{\max}(H_{\mathrm{full}})$ in equations \eqref{eq:avTrotter1} and \eqref{eq:avTrotter2}, the average Trotter error is bounded by 
\begin{equation}
    \begin{aligned}
        \frac{\Langle\opnorm{\mathcal{D}(\bm{B})}_{*,p}\Rangle }{\opnorm{I}_{*,p}} \leq \beta(l)\frac{\Gamma \sqrt{p}\sigma\sqrt{p_B} t}{\sqrt{Q_{\max}(H_{\mathrm{full}})}}\cdot\Bigg( \left[\sqrt{p} \sigma\sqrt{p_BQ_{\max}(H_{\mathrm{full}})}\frac{t}{r}\right]^{l} 
        + \Gamma \cdot \left[\sqrt{p}\sigma\sqrt{p_BQ_{\max}(H_{\mathrm{full}})}\frac{t}{r}\right]^{l+1}\Bigg),
    \end{aligned}
\end{equation}
    where $\beta(l)$ is a purely $l$-dependent factor. If $p_BQ_{\max}(H_{\mathrm{full}})\leq 1$ instead, the $l$-th order average Trotter error is bounded by
    \begin{equation}
    \begin{aligned}
         \frac{\Langle\opnorm{\mathcal{D}(\bm{B})}_{*,p}\Rangle }{\opnorm{I}_{*,p}} \leq\beta(l)\frac{\Gamma \sqrt{p}\sigma t}{Q_{\max}(H_{\mathrm{full}})}\cdot\Bigg( \left[\sqrt{p} \sigma\frac{t}{r}\right]^{l} 
        + \Gamma \cdot \left[\sqrt{p}\sigma\frac{t}{r}\right]^{l+1}\Bigg).
    \end{aligned}
\end{equation}
Subsequently, gate complexities can be derived in a similar fashion as Corollaries \ref{cor:511 new} and \ref{cor:512 new}.

\subsection{Outlook and open questions}

Although we have derived bounds for several orders of Trotter error and found the near-optimal complexity for the product formula simulation of both the SYK and sparse SYK models, there are still some open questions left. 

The first one concerns the error bound $\Delta_l.$ Recall that by Theorem \ref{thm:SYKRPG},
\begin{align}
    \frac{\opnorm{e^{iHt} - S_l(t/r)^r}_p}{\opnorm{I}_p} &\leq \Delta_l(n,k,t,r,p),
\end{align}
where 
\begin{equation}
     \Delta_l:= \frac{\mathcal{D}(l)\sqrt{p} \sigma t}{\sqrt{Q(n,k)}}\left( \binom{n}{k} \left[\sqrt{p} \sigma \sqrt{Q(n,k)}\frac{t}{r}\right]^{l}  +  \binom{n}{k}^2\left[\sqrt{p} \sigma\sqrt{Q(n,k)}\frac{t}{r}\right]^{l+1} \right). 
\end{equation}
Asymptotically for even $k$, note that $\sigma = \mathcal{O}(n^{-\frac{k-1}{2}})$ and $Q(n,k) = \mathcal{O}(n^{k-1})$. Hence, $\sigma\sqrt{Q(n,k)} = \mathcal{O}(1).$ Therefore, for any even $l>1$, the scaling of $\Delta_l$ in $n$ will asymptotically not depend on the terms in the square brackets. In other words, fix $k,t,r,$ and $p$, then for any even $l>1$, $\Delta_l$ will have the following asymptotic scaling in $n$:
\begin{equation}
\label{eq:DUSO}
    \Delta_l(n,k,t,r,p)
    = \mathcal{O}\left(\frac{\sigma}{\sqrt{Q(n,k)}}\binom{n}{k}^2\right)
    = \mathcal{O}\left(n^{k+1}\right),
\end{equation}
which is independent of $l$. This behavior is illustrated in Figure~\ref{fig:loglogDeltaUS}, where the $\Delta_l$ for different $l$ start with different curves, but after some $n_0$ (in Figure~\ref{fig:loglogDeltaUS} we could take $n_0\in [e^3, e^4]$), they all become parallel straight lines for $n \geq n_0$, indicating the same scaling behavior in $n$.
\begin{figure}[!hbt]
    \centering
    \includegraphics[width=\linewidth]{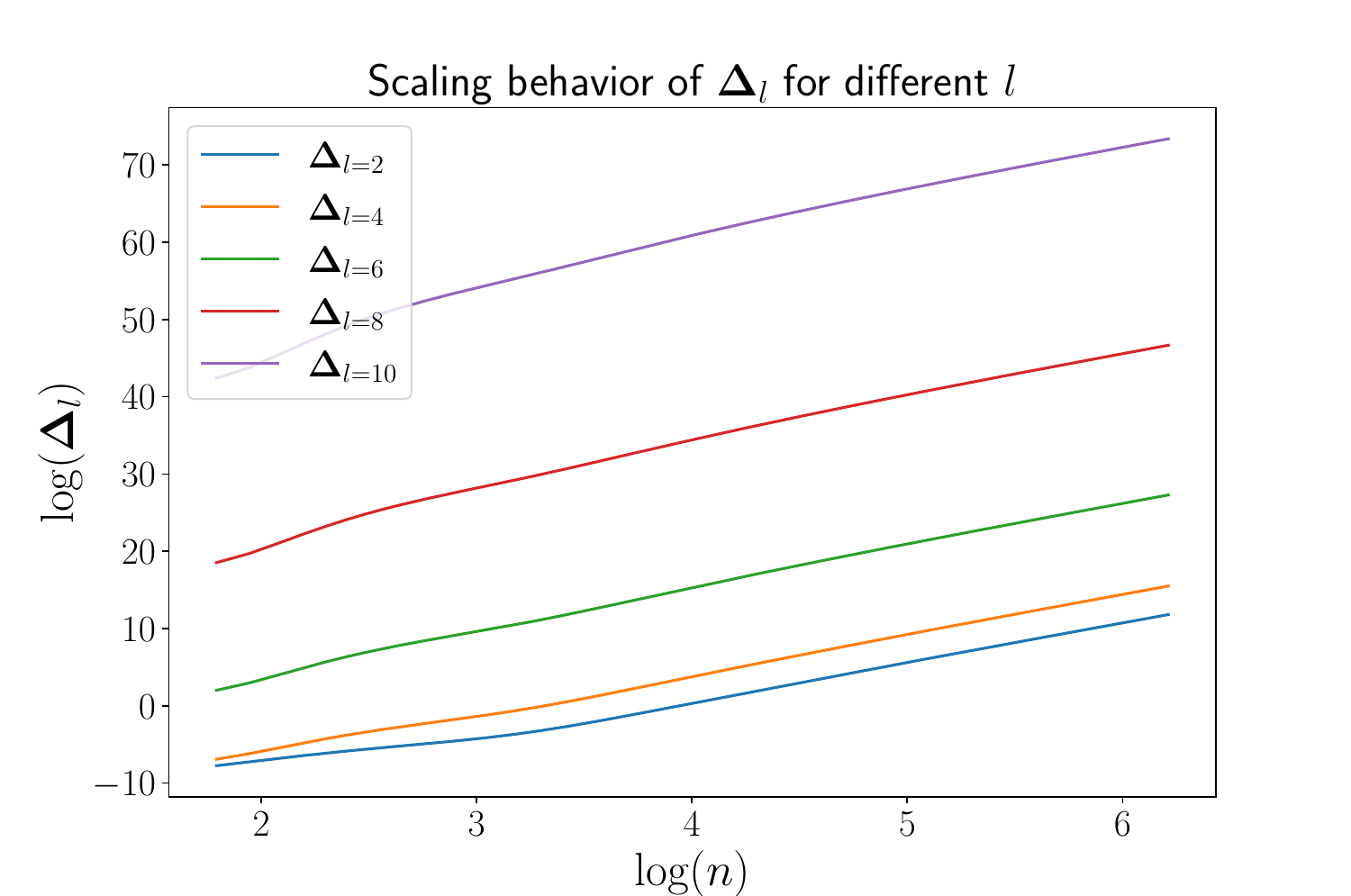}
    \caption{Scaling behavior of $\Delta_l$ in $n\in [6,500]$ for $l=2,4,6,8,10$ at fixed $k=4$, $t=10$, $r=100000$, $p=2$. When plotting these, we omitted the calculation of the prefactor $\mathcal{D}(l)$ because it gets too large for us to compute. Therefore, this plot does not give the correct prefactors but still captures the correct scaling behavior in $n$.}
    \label{fig:loglogDeltaUS}
\end{figure}

Since $\Delta_l$ is an upper bound for the $l$-th Trotter error, this indicates that the $l$-th order Trotter error of the SYK model with even $k$ for any even $l>1$ could also have this uniform asymptotic scaling independent of $l$. Numerically confirming this interesting observation by directly computing the Trotter error is beyond the reach of our computational capability. We leave the verification of this result to further research. 

The second question addresses the gate complexity of the sparse SYK model. The quantity we provide does not yield the full picture, as it only gives a measure for the gate complexity averaged over all samples of the Bernoulli variables. To obtain a fuller picture, one could consider computing higher-order moments of the gate complexity, such as the variance, and give a probabilistic description via Chebyshev's inequality in the form 
\begin{equation}
    \Prob\left(|G(\bm{B}) - \overline{G}|\geq c\sqrt{\mathrm{Var}(G(\bm{B}))} \right)\leq \frac{1}{c^2},
\end{equation}
where $G(\bm{B})$ is the gate complexity dependent on the Bernoulli variables, $\overline{G}$ is the average gate complexity we computed, and $c$ is a constant strictly larger than 1. Alternatively, one could define a typical set for the samples of the Bernoulli variables, and provide a statement about gate complexities for typical samples. We leave both approaches for future research. 

Lastly, inspired by \cite{PhysRevD.109.105002}, we could consider a graph-coloring optimization to reduce the number of gates per round in a product formula. This idea is based on the fact that we can decompose the SYK Hamiltonian into clusters of mutually commuting terms. If one considers the anti-commutation graph $A$ of the SYK Hamiltonian, i.e. the graph obtained by placing each local term as a vertex, and an edge between two vertices if the corresponding local terms anti-commute, then the chromatic number $\chi_A$ is the least number of commuting clusters we can decompose the SYK Hamiltonian into. The commuting clusters have the advantage that they can be simultaneously diagonalized. Implementing each commuting cluster as Pauli strings requires $\mathcal{O}(n^2)$ gates, coming from the simultaneous diagonalization of the cluster \cite{diagPauli1,diagPauli2,diagPauli3}. This leads to a complexity of $\mathcal{O}(n^2\chi(A))$ gates per round. Because the anti-commutation graph $A$ is a $Q(n,k)$-regular graph, its chromatic number $\chi_A$ is upper bounded by $Q(n,k)+1.$ Hence, using this approach, the gate complexity per round would scale at most with $\mathcal{O}(n^{k+1})$, which is the same as the gate complexity where we use the Jordan--Wigner transformation to map each Majorana fermion to a Pauli string instead of the ternary tree mapping. However, the chromatic number $\chi_A$ can potentially exhibit a better scaling than $Q(n,k)$, as we can see in Figure~\ref{fig:chia}. Hence, although decomposing the SYK Hamiltonian into commuting clusters does not impact the Trotter number as it is not ordering dependent, this approach could potentially lead to a more efficient implementation of single rounds of the product formulas. 

\begin{figure}[!hbt]
    \centering
    \includegraphics[width=0.8\linewidth]{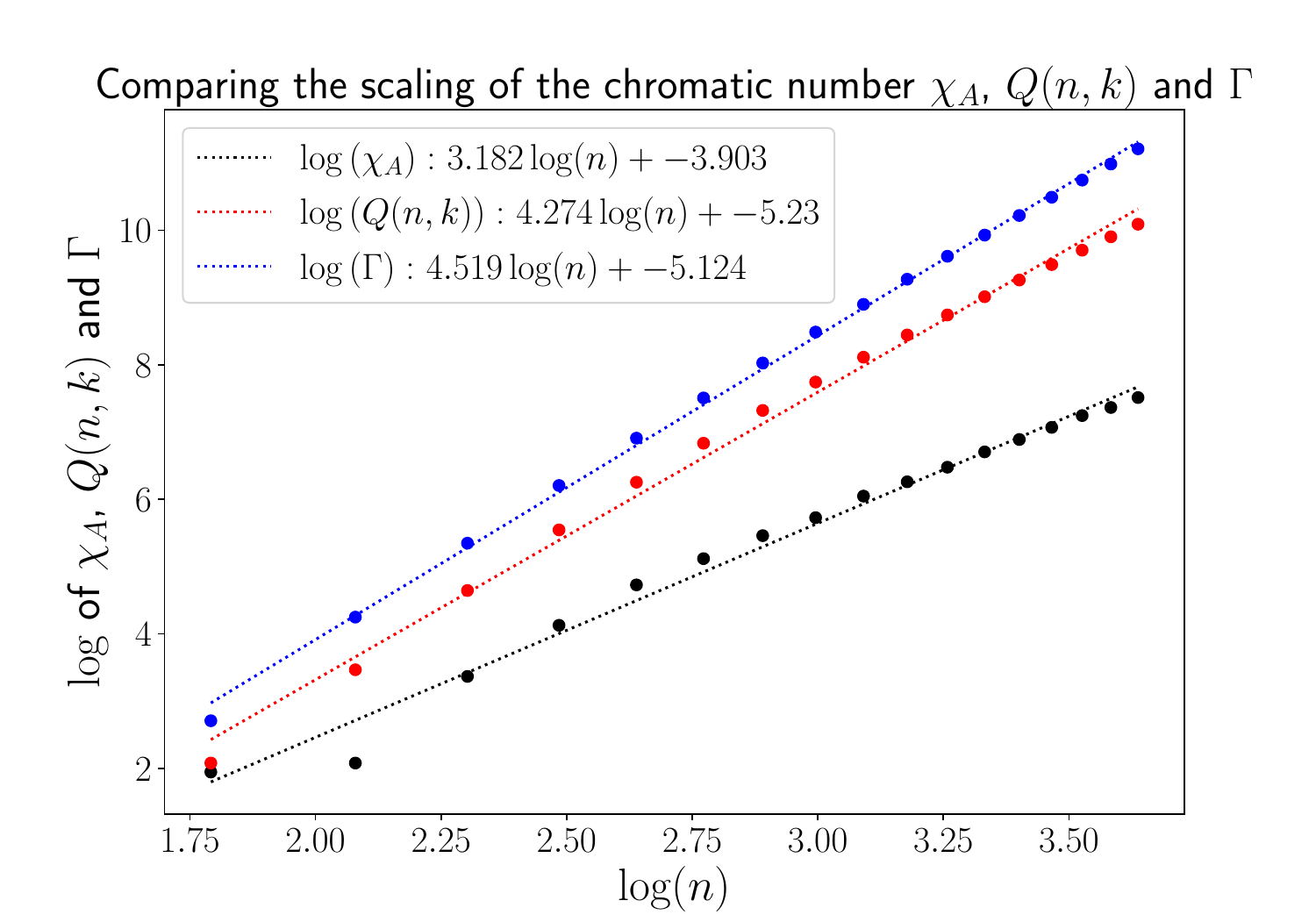}
    \caption{Scaling comparison between the chromatic number $\chi_A$ of the anti-commutation graph of the SYK model $A$ and $Q(n,k)$ for even $n\in [6,38]$ and $k=4$. The chromatic number is estimated using a greedy coloring algorithm.}
    \label{fig:chia}
\end{figure}

\pagebreak

\bibliographystyle{alphaurl}
\bibliography{reference}

\appendix
\section{Lemmas for evaluating the probability statements}
\label{apx:concineq}
\lemmaconcineq*

\begin{proof}
Let $H$ be a Hamiltonian on a Hilbert space of dimension $D$. Suppose its $l$-th order Trotter error is bounded by 
\begin{equation}
        \opnorm{e^{iHt} - S_l(t/r)^r}_p \leq p\lambda(p, r) \opnorm{I}_p
    \end{equation}
for all $p\geq 2$.
By applying Markov's inequality, we obtain
\begin{equation}
\label{eq:markov}
    \Prob\left(\norm{e^{iHt} - S_l(t/r)^r}\geq \epsilon\right) \leq \frac{\opnorm{e^{iHt} - S_l(t/r)^r}_p^p}{\epsilon^p}\leq D\left(p\frac{\lambda(p,r)}{\epsilon}\right)^p,
\end{equation}
where $D=\opnorm{I}_p^p$. We want to achieve $D\left(p\frac{\lambda(p,r)}{\epsilon}\right)^p<\delta$. Let 
    $p_D = \log(e^2D/\delta)$.
Since equation \eqref{eq:markov} holds for all $p\geq 2$ and $p_D\geq 2$,
it must also hold for $p_D$.\footnote{We have $p_D\geq 2$ because $D> \delta,$ since $D>1$ and $\delta\in (0,1)$ for it bounds a probability.}

Observe the following equivalence:
\begin{align}
    \left(p_D\frac{\lambda(p_D,r)}{\epsilon}\right)^{p_D} \leq \frac{\delta}{D} &\iff \log\left(\left(p_D\frac{\lambda(p_D,r)}{\epsilon}\right)^{p_D}\right)\leq \log\left(\frac{\delta}{D}\right)\\
    &\iff \log\left(p_D\frac{\lambda(p_D,r)}{\epsilon}\right)\cdot p_D \leq \log\left(\frac{\delta}{D}\right)\\
    &\iff \log\left(\frac{\lambda(p_D,r)}{\epsilon}\right) \leq \frac{-\log(D/\delta)}{p_D} - \log\left(p_D\right)\\
    &\iff \frac{\lambda(p_D,r)}{\epsilon} \leq \exp\left(\frac{-\log(D/\delta)}{\log(D/\delta) + 2}\right) \cdot \frac{1}{p_D}.
\end{align}

Now, the minimal value of $e^{\frac{-x}{x+2}}$ for $x \geq 0$ is at $x\to \infty$ with $e^{-1}\leq e^{\frac{-x}{x+2}}$ for all $x\geq 0.$ Thus, 
\begin{equation}
    \frac{\lambda(p_D,r)}{\epsilon} \leq \frac{1}{e}\cdot \frac{1}{p_D} \implies \frac{\lambda(p_D,r)}{\epsilon} \leq \exp\left(\frac{-\log{D/\delta}}{p_D}\right) \cdot \frac{1}{p_D} \iff\left(p_D\frac{\lambda(p_D,r)}{\epsilon}\right)^{p_D} \leq \frac{\delta}{D} 
\end{equation}
which means
\begin{equation}
    \frac{\lambda(p_D,r)}{\epsilon} \leq \frac{1}{e}\cdot \frac{1}{p_D} \implies \Prob\left(\norm{e^{iHt} - S_1(t/r)^r}> \epsilon\right) \leq \delta. 
\end{equation}
In other words, solving 
\begin{equation}
    \frac{\lambda(p_D,r)}{\epsilon} \leq \frac{1}{e}\cdot \frac{1}{p_D}
\end{equation}
for $r$ ensures the concentration inequality. This concludes the lemma.
\end{proof}

\lemmaconcineqpsi*

\begin{proof}
Let $H$ be a Hamiltonian on a Hilbert space of dimension $D$ and $\ket{\psi}$ an arbitrary fixed input state in the Hilbert space. Suppose the $l$-th order Trotter error is bounded by 
\begin{equation}
        \opnorm{e^{iHt} - S_l(t/r)^r}_{fix,p} \leq p\lambda(p, r)
    \end{equation}
for all $p\geq 2$.
Observe that by Markov's inequality
\begin{align}
    \Prob\left(\norm{(e^{iHt} - S_l(t/r)^r)\ket{\psi}}_{l^2}\geq \epsilon\right) &\leq \frac{\mathbb{E}(\norm{(e^{iHt} - S_l(t/r)^r)\ket{\psi}}_{l^2}^p)}{\epsilon^p} \\
    &\leq \sup_{\norm{\ket{\phi}}_{l^2} = 1}\left\{\frac{\mathbb{E}(\norm{(e^{iHt} - S_l(t/r)^r)\ket{\phi}}_{l^2}^p)}{\epsilon^p} \right\}.
\end{align}
Hence,
\begin{equation}
\label{eq:markovpsi}
    \Prob\left(\norm{(e^{iHt} - S_l(t/r)^r)\ket{\psi}}_{l^2}\geq \epsilon\right) \leq\frac{\opnorm{e^{iHt} - S_l(t/r)^r}_{fix,p}^p}{\epsilon^p} \leq \left(p\frac{\lambda(p,r)}{\epsilon}\right)^p.
\end{equation}
We want to achieve $\left(p\frac{\lambda(p,r)}{\epsilon}\right)^p<\delta$. Let 
    $p_1 = \log(e^2/\delta)$.
Since equation \eqref{eq:markovpsi} holds for all $p\geq 2$ and $p_1\geq 2$,
it must also hold for $p_1$. The rest of the proof is analogous to the previous one by replacing $p_D$ with $p_1$. This concludes the lemma.
\end{proof}

\section{Anti-commutation lemma for Majorana fermions}\label{apx:Majoranas}

\lemmaanticom*

\begin{proof}
Let $\alpha = \gamma(i)$ and $\beta=\gamma(j)$ two different hyperedges. Suppose $H_{\alpha} = J_{\alpha}\chi_{\alpha_1}\cdots \chi_{\alpha_k}$ and $H_{\beta} = J_{\beta}\chi_{\beta_1}\cdots \chi_{\beta_k}$, where $\alpha_i\neq \alpha_j$ and $\beta_i\neq\beta_j$ for $i\neq j.$ 
Given the anti-commutation relation
\begin{equation}
    \{\chi_i,\chi_j\} = 2\delta_{ij}
\end{equation}
we are going to show
\begin{equation}
    H_{\alpha}H_{\beta} = (-1)^{k^2-m^2}H_{\beta}H_{\alpha}
\end{equation}
where $m = |\alpha\cap \beta|$.

We write the product
\begin{equation}
    H_{\alpha}H_{\beta} = J_{\alpha}J_{\beta}\chi_{\alpha_1}\cdots \chi_{\alpha_k}\chi_{\beta_1}\cdots \chi_{\beta_k} =: J_{\alpha}J_{\beta}(\alpha_1,\dotsc,\alpha_k,\beta_1,\dotsc,\beta_k)
\end{equation}
as an ordered set. Now, suppose $|\alpha\cap \beta| = 0.$ This means that $H_{\alpha}$ and $H_{\beta}$ contain completely different fermions. Therefore, by the anti-commutation relation:
\begin{align}
    H_{\alpha}H_{\beta} &= J_{\alpha}J_{\beta}(\alpha_1,\dotsc,\alpha_k,\beta_1,\dotsc,\beta_k)\\
    &= (-1)^k \cdot J_{\alpha}J_{\beta}(\beta_1,\alpha_1,\dotsc,\alpha_k,\beta_2,\dotsc,\beta_k)\\
    &= (-1)^{2k} \cdot J_{\alpha}J_{\beta}(\beta_1,\beta_2, \alpha_1,\dotsc,\alpha_k,\beta_3,\dotsc,\beta_k)\\
    &= (-1)^{k^2} \cdot J_{\alpha}J_{\beta}(\beta_1,\dotsc,\beta_k,\alpha_1,\dotsc,\alpha_k)\\
    &= (-1)^{k^2} H_{\beta}H_{\alpha}.
\end{align}

Now, suppose $|\alpha\cap \beta| = m >0.$ This means that $H_{\alpha}$ and $H_{\beta}$ share common fermions $\{\gamma_1,\dotsc,\gamma_m\}$. Define
\begin{align}
    \tilde{H}_{\alpha} &= J_{\alpha}(\alpha_1,\dotsc, \alpha_{k-m}, \gamma_1,\dotsc,\gamma_{m}), \\
    \tilde{H}_{\beta} &= J_{\beta}(\beta_1,\dotsc, \beta_{k-m}, \gamma_1,\dotsc,\gamma_{m}).
\end{align}
Note that all permutations $\pi$ of $(\alpha_1,\dotsc, \alpha_{k-m}, \gamma_1,\dotsc,\gamma_{m})$ only yield a parity factor $\sigma(\pi)$ at the front. In other words,
\begin{equation}
    \pi (\alpha_1,\dotsc, \alpha_{k-m}, \gamma_1,\dotsc,\gamma_{m}) = \sigma(\pi)\cdot (\alpha_1,\dotsc, \alpha_{k-m}, \gamma_1,\dotsc,\gamma_{m}).
\end{equation}
Since $H_{\alpha}$ and $H_{\beta}$ can be written as permutations of $\tilde{H}_{\alpha}$ and $\tilde{H}_{\beta}$ respectively, we have 
\begin{equation}
    H_{\alpha} = \sigma(\pi_1)\tilde{H}_{\alpha},\qquad\text{and}\qquad H_{\beta} = \sigma(\pi_2)\tilde{H}_{\beta}
\end{equation}
where $\pi_1$ and $\pi_2$ are the corresponding permutations that send $H_{\alpha}$ and $H_{\beta}$ to $\tilde{H}_{\alpha}$ and $\tilde{H}_{\beta}$, respectively. Hence,
\begin{align*}
    H_{\alpha}H_{\beta} &= \sigma(\pi_1)\tilde{H}_{\alpha} \cdot \sigma(\pi_2)\tilde{H}_{\beta}\\
    &= \sigma(\pi_1)\sigma(\pi_2)\cdot J_{\alpha}J_{\beta}(\alpha_1,\dotsc,\alpha_{k-m},\gamma_1,\dotsc,\gamma_m,\beta_1,\dotsc,\beta_{k-m}, \gamma_1,\dotsc,\gamma_m)\\
    &=\sigma(\pi_1)\sigma(\pi_2)\cdot (-1)^{(k-m)k}\cdot J_{\alpha}J_{\beta}(\beta_1,\dotsc,\beta_{k-m}, \alpha_1,\dotsc,\alpha_{k-m},\gamma_1,\dotsc,\gamma_m, \gamma_1,\dotsc,\gamma_m)\\
    &= \sigma(\pi_1)\sigma(\pi_2)\cdot (-1)^{(k-m)k + m(k-m)}\cdot J_{\alpha}J_{\beta}(\beta_1,\dotsc,\beta_{k-m}, \gamma_1,\dotsc,\gamma_m,\alpha_1,\dotsc,\alpha_{k-m},\gamma_1,\dotsc,\gamma_m)\\
    &= \sigma(\pi_1)\sigma(\pi_2)\cdot (-1)^{k^2-m^2}\cdot \tilde{H}_{\beta}\tilde{H}_{\alpha}\\
    &= (-1)^{k^2-m^2}H_{\beta}H_{\alpha}\\
    &=(-1)^{k+m}H_{\beta}H_{\alpha}.
    \qedhere
\end{align*}
\end{proof}

\section{Random matrix polynomials with Gaussian coefficients}\label{sec:RMPGC}

In the main text we used \ref{thm:matpol}, which comes from Chen \& Brand\~{a}o \cite[Theorem VII.3]{Anthony} (with a slight variation in the prefactor). For completeness, we present a proof of this result in this appendix (which takes a slightly different track than \cite{Anthony}).

For context, consider a random matrix polynomial of degree $g$
    \begin{equation}
        F(X_m, \dotsc, X_1) = \sum_{\bm{i}} T_{\bm{i}} X_{i_g} \cdots X_{i_1}
    \end{equation}
    where each term $X_i$ can be written as
    \begin{equation}
        X_i = g_i K_i
    \end{equation}
with $g_i$ i.i.d.\@ standard Gaussian variables and $K_i$ a deterministic matrix bounded by $\norm{K_i}\leq \sigma_i$. 
By Rademacher expansion and resummation by a subset $S$ of the index set $\mathcal{I}:= [m]\times [N]$, we can write $F$ as 
\begin{equation}
\begin{aligned}
    F =\sum_{x + y = g} \sum_{s_1,\dotsc, s_x \in S}\sum_{s'_1,\dotsc,s'_y \in S^c} \sum_{\pi \in \mathcal{S}_g\big/\mathcal{S}_x\times \mathcal{S}_y}T_{\pi\left[\bm{\pr}_{1}  (s_1,\dotsc, s_x, s'_1,\dotsc, s'_y)\right]}\pi\left[Y_{s_1}\cdots Y_{s_x} Y_{s'_1}\dotsc Y_{s'_y}\right].
\end{aligned}
\end{equation}
which prepares us to act the smoothing operator on it. For the notation here, we refer the reader to Section \ref{subsec:randommatpol}. The rest of this Appendix constitutes the proof of the following theorem.
\thmRPG* 

\begin{proof}

Applying $\mathbb{E}_{S^c}$ to $F$ gives
\begin{equation}
\label{eq: EScF1}
    \begin{aligned}
        \mathbb{E}_{S^c}[F] = \sum_{x + y = g} \sum_{s_1,\dotsc, s_x \in S}\sum_{s'_1,\dotsc,s'_y \in S^c }\sum_{\pi \in \mathcal{S}_g\big/\mathcal{S}_x\times \mathcal{S}_{y}}T_{\pi\left[\bm{\pr}_{1}  (s_1,\dotsc, s_x, s'_1,\dotsc, s'_{y})\right]} \pi\left[Y_{s_1}\cdots Y_{s_x} \mathbb{E}(Y_{s'_1}\cdots Y_{s'_{y}})\right].
    \end{aligned}
\end{equation}
We now analyze the sum over $s'_1,\dotsc,s'_y$ more carefully. Observe that
\begin{equation}
    \mathbb{E}(Y_{s'_1}\cdots Y_{s'_y}) \neq 0
\end{equation}
if and only if $y$ is even, since the product is nonzero in expectation only when it contains no odd moments of $Y_{s'_i}$. Thus, we scale $y$ to $2y$ and update the summation constraints accordingly:
\begin{equation}
    \begin{aligned}
        \mathbb{E}_{S^c}[F] = \sum_{x + 2y = g} \sum_{s_1,\dotsc, s_x \in S}\sum_{s'_1,\dotsc,s'_y \in S^c }\sum_{\pi \in \mathcal{S}_g\big/\mathcal{S}_x\times \mathcal{S}_{2y}}T_{\pi\left[\bm{\pr}_{1}  (s_1,\dotsc, s_x, s'_1,\dotsc, s'_{2y})\right]} \pi\left[Y_{s_1}\cdots Y_{s_x} \mathbb{E}(Y_{s'_1}\cdots Y_{s'_{2y}})\right].
    \end{aligned}
\end{equation}
Since each Rademacher variable contributes only through exponents $0$, $1$, or $2$, any non-vanishing $\mathbb{E}(Y_{s'_1}\cdots Y_{s'_{2y}})$ must consist of a product of second moments of Rademachers. Therefore, the subscripts $\{s'_1,\dotsc, s'_{2y}\}$ contain exactly $y$ distinct elements $\{s''_1,\dotsc, s''_y\}$, and the expectation takes the form
\begin{equation}
     \mathbb{E}(Y_{s'_1}\cdots Y_{s'_{2y}}) = \pi'\left[\mathbb{E}(Y_{s''_1}^2)\cdots \mathbb{E}(Y_{s''_{y}}^2)\right],
\end{equation}
where $\pi'\in \mathcal{S}_{2y}$ reorders the product. This reordering is necessary due to the non-commutativity of the Rademacher variables. Hence, the corresponding sum over expectations in equation \eqref{eq: EScF1} can be expressed as
\begin{equation}
    \sum_{s'_1,\dotsc, s'_{2y}\in S^c} \mathbb{E}(Y_{s'_1}\cdots Y_{s'_y})  =\sum_{\substack{s''_1,\dotsc, s''_{y}\in S^c \\ \mathrm{distinct}}}\sum_{\pi' \in \mathcal{S}_{2y} \big/ \mathcal{S}_2^{y}\rtimes \mathcal{S}_y} \pi'\left[\mathbb{E}(Y_{s''_1}^2)\cdots \mathbb{E}(Y_{s''_{y}}^2)\right].
\end{equation}
We do not sum over the entire group $\mathcal{S}_{2y}$ for $\pi'$ to avoid overcounting. Specifically, two classes of permutations must be excluded: 
(i) internal flips that leave the argument invariant, of the form $(1\, 2)^{\zeta_1}(3\, 4)^{\zeta_2}\cdots$ with $\zeta_i \in\{0,1\}$, corresponding to the group
\begin{equation}
    \mathcal{S}^y_2:=  \underbrace{\mathcal{S}_2\times\mathcal{S}_2\times \cdots \times \mathcal{S}_2}_{\text{$y$ times}},
\end{equation}
and 
(ii) pair permutations, i.e. permutations that exchange entire pairs, corresponding to $\mathcal{S}_y$. 
The composition of these two symmetry types (first the internal flips, then the pair permutations) forms the group
\begin{equation}
    \mathcal{S}^y_2 \rtimes \mathcal{S}_y.
\end{equation}
The semi-direct product ensures the correct group structure.\footnote{Normally, such semi-direct product does not give rise to a group structure, but in this case the symmetries of $\mathcal{S}^y_2$ and $\mathcal{S}_y$ have trivial intersection. Therefore, the semi-direct product does yield a group structure. To visualize this, consider a vector of $2y$ entries divided into $y$ columns where each column contains a pair: 
\begin{equation}
    (s''_1, s''_2 \vert s''_3, s''_4 \vert \cdots \vert s''_{2y-1}, s''_{2y}).
\end{equation}
The symmetries of $\mathcal{S}^y_2$ act on this vector by flipping the pairs inside each column. For instance, flipping the pair in the second column gives 
\begin{equation}
    (s''_1, s''_2 \vert s''_4, s''_3 \vert \cdots \vert s''_{2y-1}, s''_{2y}).
\end{equation}
The symmetries of $\mathcal{S}_y$ permute the columns among each other while leaving the pairs inside the columns invariant. For instance, permuting the first column and the last column gives us
\begin{equation}
    (s''_{2y-1}, s''_{2y}  \vert s''_3, s''_4 \vert \cdots \vert s''_1, s''_2).
\end{equation}
Then it is clear that the symmetries of $\mathcal{S}^y_2$ and $\mathcal{S}_y$ have trivial intersection.
} This explains the notation $\pi' \in \mathcal{S}_{2y} \big/ \mathcal{S}_2^{y}\rtimes \mathcal{S}_y$ in the summation. Plugging this back into the sum in equation \eqref{eq: EScF1} gives
\begin{equation}
    \begin{aligned}
        \mathbb{E}_{S^c}[F] = \sum_{x + 2y = g} \sum_{s_1,\dotsc, s_x \in S}&\sum_{\pi \in \mathcal{S}_g\big/\mathcal{S}_x\times \mathcal{S}_{2y}}\sum_{\substack{s''_1,\dotsc, s''_{y}\in S^c \\ \mathrm{distinct}}}\sum_{\pi' \in \mathcal{S}_{2y} \big/ \mathcal{S}_2^{y}\rtimes \mathcal{S}_y} \\
        &T_{\pi\left[\bm{\pr}_{1}  (s_1,\dotsc, s_x, \pi'[ s'_1,s'_1,\dotsc,s'_y, s'_y])\right]} \pi\left[Y_{s_1}\cdots Y_{s_x}\pi'\left[\mathbb{E}(Y_{s''_1}^2)\cdots \mathbb{E}(Y_{s''_{y}}^2)\right]\right].
    \end{aligned}
\end{equation}
\\

Applying the central part $(1-\mathbb{E})_S$ to $\mathbb{E}_{S^c}[F]$ and using the triangle inequality yields
\begin{equation}
\label{eq: FSp1}
    \begin{aligned}
        \opnorm{[F]_S}_{*,p} \leq &\sum_{\pi \in \mathcal{S}_g\big/\mathcal{S}_{2y}}\sum_{\substack{s''_1,\dotsc, s''_{y}\in S^c \\ \mathrm{distinct}}}\sum_{\pi' \in \mathcal{S}_{2y} \big/ \mathcal{S}_2^{y}\rtimes \mathcal{S}_y} \\ &\left|T_{\pi\left[\bm{\pr}_{1}  (s_1,\dotsc, s_{|S|},\pi'[ s''_1,s''_1,\dotsc,s''_y, s''_y])\right]}\right|\opnorm[\Bigg]{\pi\left[Y_{s_1}\cdots Y_{s_{|S|}} \pi'\left[\mathbb{E}(Y_{s''_1}^2)\cdots \mathbb{E}(Y_{s''_{y}}^2)\right]\right]}_{*,p}.
    \end{aligned}
\end{equation}
Here $y$ is fixed by $y = (g - |S|) / 2$. This follows from the fact that 
\begin{equation}
    (1-\mathbb{E})_S\left[Y_{s_1}\cdots Y_{s_{x}}\right] \neq 0
\end{equation}
if and only if all subscripts in $S$ appear in the product $Y_{s_1}\cdots Y_{s_{x}}$ and each term has an odd exponent. Since each contributing Rademacher variable can only have exponent $0$, $1$, or $2$, all must have exponent $1$, fixing $x=|S|$ and thus $y=(g-|S|)/2$. We also extend the summation over $\pi$ to the entire group $\mathcal{S}_g$\footnote{We are allowed to do this here, because all the summands after triangle inequality in equation \eqref{eq: FSp1} are positive.}, quotienting out permutations that reorder $Y_{s''_1},\dotsc,Y_{s''_y}$ among themselves. This extension means that the part of $\mathcal{S}_g$ we sum over now accounts for permutations that reorder $Y_{s_1},\dotsc, Y_{s_{|S|}}$ among themselves, allowing us to drop the summation over $s_1,\dotsc, s_{|S|}.$

By construction, each Rademacher variable can be written as
\begin{equation}
    Y_{s} = \frac{\epsilon_s}{\sqrt{N}} K_{\pr_1s}
\end{equation}
and recall that each $K_{\pr_1s}$ is bounded by $\opnorm{K_{\pr[m]s}}_{*,p} \leq \sigma_{\pr[m]s} \cdot \opnorm{I}_{*,p}$. 
Plugging this back in equation \eqref{eq: FSp1} gives us
\begin{equation}
    \begin{aligned}
        \opnorm{[F]_S}_{*,p} \leq \sum_{\pi \in \mathcal{S}_g\big/\mathcal{S}_{2y}}\sum_{\substack{s''_1,\dotsc, s''_{y}\in S^c \\ \mathrm{distinct}}}&\sum_{\pi' \in \mathcal{S}_{2y} \big/ \mathcal{S}_2^{y}\rtimes \mathcal{S}_y} \left|T_{\pi\left[\bm{\pr}_{1}  (s_1,\dotsc, s_{|S|},\pi'[ s''_1,s''_1,\dotsc,s''_y, s''_y])\right]}\right|\\
        &\frac{1}{\sqrt{N}^{|S|}}\frac{1}{N^{\frac{g-|S|}{2}}}
        \left[\sigma_{\pr_1s_1}\cdots \sigma_{\pr_1s_{|S|}} \sigma_{\pr_1s''_1}^2\cdots \sigma_{\pr_1s''_y}^2\right]\opnorm{I}_{*,p}
    \end{aligned}
\end{equation}
We can now simplify the right-hand side by making the summand independent of the summation over $\pi'$. To do so, we first extend the sum over $\pi$ to the full group $\mathcal{S}_g$. Since we are summing positive terms, the upper bound remains valid. Notice that the sum over the entire $\mathcal{S}_g$ commutes with the other two summations. Hence, we can write equation \eqref{eq: FSp1} as
\begin{equation}
\label{eq: FSp2}
    \begin{aligned}
        \opnorm{[F]_S}_{*,p} \leq \sum_{\substack{s''_1,\dotsc, s''_{y}\in S^c \\ \mathrm{distinct}}}&\sum_{\pi' \in \mathcal{S}_{2y} \big/ \mathcal{S}_2^{y}\rtimes \mathcal{S}_y}\sum_{\pi \in \mathcal{S}_g}\left|T_{\pi\left[\bm{\pr}_{1}  (s_1,\dotsc, s_{|S|},\pi'[ s''_1,s''_1,\dotsc,s''_y, s''_y])\right]}\right|\\
        &\frac{1}{\sqrt{N}^{|S|}}\frac{1}{N^{\frac{g-|S|}{2}}}
        \left[\sigma_{\pr_1s_1}\cdots \sigma_{\pr_1s_{|S|}} \sigma_{\pr_1s''_1}^2\cdots \sigma_{\pr_1s''_y}^2\right]\opnorm{I}_{*,p}.
    \end{aligned}
\end{equation}
Observe that the sum over $\pi\in \mathcal{S}_g$ has the following property: for any permutation $\tau \in \mathcal{S}_g$, 
\begin{equation}
    \sum_{\pi \in \mathcal{S}_g}\left|T_{\pi\left[\bm{\pr}_{1}  (\tau[s_1,\dotsc, s_{|S|},s''_1,s''_1,\dotsc,s''_y, s''_y])\right]}\right| = \sum_{\pi \in \mathcal{S}_g}\left|T_{\pi\left[\bm{\pr}_{1}  (s_1,\dotsc, s_{|S|}, s''_1,s''_1,\dotsc,s''_y, s''_y)\right]}\right|,
\end{equation}
because the composition of $\tau$ and $\pi$ yields another permutation in $\mathcal{S}_g$. We can view each $\pi' \in \mathcal{S}_{2y} \big/ \mathcal{S}_2^{y}\rtimes \mathcal{S}_y$ as a permutation in $\mathcal{S}_g$ that leaves the first $|S|$ elements invariant and acts only on the last $2y$ elements. Hence, for every $\pi'\in \mathcal{S}_{2y} \big/ \mathcal{S}_2^{y}\rtimes \mathcal{S}_y,$ we have
\begin{equation}
     \sum_{\pi \in \mathcal{S}_g}\left|T_{\pi\left[\bm{\pr}_{1}  (s_1,\dotsc, s_{|S|},\pi'[ s''_1,s''_1,\dotsc,s''_y, s''_y])\right]}\right| =  \sum_{\pi \in \mathcal{S}_g}\left|T_{\pi\left[\bm{\pr}_{1}  (s_1,\dotsc, s_{|S|}, s''_1,s''_1,\dotsc,s''_y, s''_y)\right]}\right|.
\end{equation}
Plugging this back into the sum in equation \eqref{eq: FSp2}, we get 
\begin{equation}
\label{eq: FSp3}
    \begin{aligned}
        \opnorm{[F]_S}_{*,p} \leq \sum_{\substack{s''_1,\dotsc, s''_{y}\in S^c \\ \mathrm{distinct}}}&\sum_{\pi' \in \mathcal{S}_{2y} \big/ \mathcal{S}_2^{y}\rtimes \mathcal{S}_y}\sum_{\pi \in \mathcal{S}_g}\left|T_{\pi\left[\bm{\pr}_{1}  (s_1,\dotsc, s_{|S|}, s''_1,s''_1,\dotsc,s''_y, s''_y)\right]}\right|\\
        &\frac{1}{\sqrt{N}^{|S|}}\frac{1}{N^{\frac{g-|S|}{2}}}
        \left[\sigma_{\pr_1s_1}\cdots \sigma_{\pr_1s_{|S|}} \sigma_{\pr_1s''_1}^2\cdots \sigma_{\pr_1s''_y}^2\right]\opnorm{I}_{*,p}
    \end{aligned}
\end{equation}
and find that the summand is completely independent of the summation over $\pi'$. Noting that 
\begin{equation}
    \sum_{\pi' \in \mathcal{S}_{2y} \big/ \mathcal{S}_2^{y}\rtimes \mathcal{S}_y}1= \frac{(2y)!}{2^y\cdot y!} = (2y - 1)!! \leq g^{g/2}.
\end{equation}
Plugging this back into equation \eqref{eq: FSp3} gives us
\begin{equation}
    \begin{aligned}
        \opnorm{[F]_S}_{*,p} \leq g^{g/2}\sum_{\substack{s''_1,\dotsc, s''_{y}\in S^c \\ \mathrm{distinct}}}&\sum_{\pi \in \mathcal{S}_g}\left|T_{\pi\left[\bm{\pr}_{1}  (s_1,\dotsc, s_{|S|}, s''_1,s''_1,\dotsc,s''_y, s''_y)\right]}\right|\\
        &\frac{1}{\sqrt{N}^{|S|}}\frac{1}{N^{\frac{g-|S|}{2}}}
        \left[\sigma_{\pr_1s_1}\cdots \sigma_{\pr_1s_{|S|}} \sigma_{\pr_1s''_1}^2\cdots \sigma_{\pr_1s''_y}^2\right]\opnorm{I}_{*,p}.
    \end{aligned}
\end{equation}
\\

To sum over all $\opnorm{[F]_S}_{*,p}$ in Theorem \ref{thm:matpol}, we parametrize every $S$ by 
\begin{equation}
    S = \bigcup_{i=1}^m \{(i, j^{(i)}_1),\dotsc, (i, j^{(i)}_{w_i})\}
\end{equation}
with a vector $\bm{w}\in \{0,\dotsc,g\}^m$ with $|\bm{w}| = |S|$, and integers $j^{(i)}_l \in [N]$. When $w_i = 0$, the corresponding set $\{(i, j^{(i)}_1),\dotsc, (i, j^{(i)}_{w_i})\}$ is empty. One can view $S$ as a union of vertical slices in the two-dimensional grid $[m]\times [N]$:
\begin{equation}
    S = \bigcup_{i=1}^m S_i,
\end{equation}
where each slice is determined by the entries of the vector $\bm{j}^{(i)} := (j^{(i)}_1,\dotsc, j^{(i)}_{w_i})$ via 
\begin{equation}
    S_i := \{(i, j^{(i)}_1),\dotsc, (i, j^{(i)}_{w_i})\}.
\end{equation}
The size of each slice $S_i$ is given by $w_i$, and $S_i$ is empty when $w_i=0$. We also define the slice complement
\begin{equation}
    S^c_i := [N]\backslash S_i,
\end{equation}
which should not be confused with $S^c$, the complement of $S$. These constructions are illustrated in Figure~\ref{fig:S_concept}.

\begin{figure}[!t]
\centering
\begin{tikzpicture}[scale=0.6, >=stealth]

\def\m{9}
\def\N{10}
\def\i{5} 

\draw[->] (0,0) -- (\m+0.8,0) node[below right] {$[m]$};
\draw[->] (0,0) -- (0,\N+0.8) node[above left] {$[N]$};

\foreach \x in {1,...,\m} \draw (\x,0) -- ++(0,-0.12);
\foreach \y in {1,...,\N} \draw (0,\y) -- ++(-0.12,0);

\draw[thick] (\i,0) -- (\i,-0.25);
\node[below] at (\i,-0.3) {$i$};

\path
  (1.2,2.0)
  .. controls (2.5,4.8) and (3.5,6.0) .. (4.2,7.5)
  .. controls (5.2,9.5) and (7.5,9.0) .. (8.0,7.2)
  .. controls (8.6,5.2) and (7.5,4.3) .. (6.7,3.8)
  .. controls (5.9,3.3) and (5.4,2.4) .. (4.8,2.1)
  .. controls (3.7,1.5) and (2.2,1.4) .. (1.2,2.0)
  -- cycle;

\fill[gray!15] (1.2,2.0)
  .. controls (2.5,4.8) and (3.5,6.0) .. (4.2,7.5)
  .. controls (5.2,9.5) and (7.5,9.0) .. (8.0,7.2)
  .. controls (8.6,5.2) and (7.5,4.3) .. (6.7,3.8)
  .. controls (5.9,3.3) and (5.4,2.4) .. (4.8,2.1)
  .. controls (3.7,1.5) and (2.2,1.4) .. (1.2,2.0)
  -- cycle;

\draw[thick] (1.2,2.0)
  .. controls (2.5,4.8) and (3.5,6.0) .. (4.2,7.5)
  .. controls (5.2,9.5) and (7.5,9.0) .. (8.0,7.2)
  .. controls (8.6,5.2) and (7.5,4.3) .. (6.7,3.8)
  .. controls (5.9,3.3) and (5.4,2.4) .. (4.8,2.1)
  .. controls (3.7,1.5) and (2.2,1.4) .. (1.2,2.0)
  -- cycle;

\node[font=\large] at (4.5,5.5) {$S$};

\draw[dashed, line width=0.8pt] (\i,0) -- (\i,\N); 

\begin{scope}
  \clip (1.2,2.0)
    .. controls (2.5,4.8) and (3.5,6.0) .. (4.2,7.5)
    .. controls (5.2,9.5) and (7.5,9.0) .. (8.0,7.2)
    .. controls (8.6,5.2) and (7.5,4.3) .. (6.7,3.8)
    .. controls (5.9,3.3) and (5.4,2.4) .. (4.8,2.1)
    .. controls (3.7,1.5) and (2.2,1.4) .. (1.2,2.0)
    -- cycle;
  \draw[line width=2.2pt] (\i,0) -- (\i,\N);
\end{scope}

\begin{scope}[shift={(6.2,-1.5)}]
  \draw[line width=2.2pt] (0,0) -- (0.8,0);
  \node[right] at (1.1,0) {$S_i$};

  \draw[dashed, line width=0.8pt] (0,-0.8) -- (0.8,-0.8);
  \node[right] at (1.1,-0.8) {$S_i^c = [N]\setminus S_i$};
\end{scope}

\end{tikzpicture}
\caption{Depiction of the set $S$ on the two-dimensional plane $[m]\times [N]$, the slice $S_i$, and slice complement $S^c_i = [N]\backslash S_i$. The same concepts also work for disconnected $S$.}
\label{fig:S_concept}
\end{figure}

These concepts allow us to parametrize the summation over $S$ by
\begin{equation}
    \begin{aligned}
        \sum_{S} \mapsto \sum_{\bm{w}}\sum_{\bm{j}^{(1)},\dotsc, \bm{j}^{(m)}}:= \sum_{\substack{j^{(1)}_1,\dotsc, j^{(1)}_{w_1} = 1 \\ \mathrm{distinct}}}^N \cdots \sum_{\substack{j^{(m)}_1,\dotsc, j^{(m)}_{w_m}=1\\ \mathrm{distinct}}}^N.
    \end{aligned}
\end{equation}
When $w_i= 0$ for a certain $1\leq i\leq m$, the sum $\sum_{\bm{j}^{(i)}} = \sum_{j^{(i)}_1,\dotsc, j^{(i)}_{w_i}=1}^N$ does not appear. Similarly, we can rewrite the summation $\sum_{\substack{s'_1,\dotsc,s'_y \in S^c \\ \mathrm{distinct}}}$ as
\begin{equation}
    \begin{aligned}
        \sum_{\substack{s'_1,\dotsc,s'_y \in S^c \\ \mathrm{distinct}}}\mapsto \sum_{\bm{v}}\sum_{\bm{k}^{(1)},\dotsc, \bm{k}^{(m)}} := \sum_{\substack{k^{(1)}_1,\dotsc, k^{(1)}_{v_1} \in S^c_1 \\ \mathrm{distinct}}} \cdots \sum_{\substack{k^{(m)}_1,\dotsc, k^{(m)}_{v_m}\in S^c_m \\ \mathrm{distinct}}}
    \end{aligned}
\end{equation}
such that $\bm{v}$ is a vector in $\{0,\dotsc, g\}^m$ such that $|\bm{w}|+2|\bm{v}| = g$. The slice complements $S^c_i$ here are understood as $[N]\big\backslash \{(i,j^{(i)}_1),\dotsc, (i, j^{(i)}_{w_i})\}$, which is fixed by the summation over $\bm{w}$ and $\bm{j}^{(i)}$. When $v_i=0$, the sum $\sum_{\bm{k}^{(i)}} = \sum_{k^{(i)}_1,\dotsc, k^{(i)}_{v_i}\in S^c_i}$ does not appear. In this way, we can still pick $y$ distinct elements from $S^c$ as
\begin{equation}
    \bigcup_{i=1}^m \{(i, k^{(i)}_1),\dotsc, (i, k^{(i)}_{v_i})\}.
\end{equation}
Using these reparametrizations, we have
\begin{equation}
\label{eq: sumFSp1}
    \begin{aligned}
        \sum_{S}C_p^{|S|}\opnorm{[F]_S}_{*,p}^2 \leq  g^{g}C_p^g\sum_{\bm{w}}&\sum_{\bm{j}^{(1)},\dotsc, \bm{j}^{(m)}}\frac{1}{N^{|S|}}\frac{1}{N^{g-|S|}} \opnorm{I}_{*,p}^2 \Bigg( \sum_{\bm{v}}\sum_{\bm{k}^{(1)},\dotsc, \bm{k}^{(m)}}\sum_{\pi \in \mathcal{S}_g}\\
        &\left|T_{\pi\left[\pr_1(1,j^{(1)}_1),\dotsc, \pr_1(m,j^{(m)}_{w_m}),\pr_1(1,k^{(1)}_1),\pr_1(1,k^{(1)}_1),\dotsc, \pr_1(m,k^{(m)}_{v_m}),\pr_1(m,k^{(m)}_{v_m}) \right]}\right|\\
       &
        \left[\sigma_{\pr_1(1,j^{(1)}_1)}\cdots \sigma_{\pr_1(m,j^{(m)}_{w_m})} \sigma_{\pr_1(1,k^{(1)}_1)}^2\cdots \sigma_{\pr_1(m,k^{(m)}_{v_m}) }^2\right]\Bigg)^2
    \end{aligned}
\end{equation}
The product over $\sigma$ can be simply written as
\begin{equation}
    \sigma_{\pr_1(1,j^{(1)}_1)}\cdots \sigma_{\pr_1(m,j^{(m)}_{w_m})} \sigma_{\pr_1(1,k^{(1)}_1)}^2\cdots \sigma_{\pr_1(m,k^{(m)}_{v_m}) }^2 = \prod_{i=1}^m \sigma_i^{w_i + 2v_i}.
\end{equation}
To write the coefficients $T_{\pi[\cdots]}$ more compactly, we introduce the following unary encoding $\bm{\eta}$
\begin{equation}
    \bm{\eta}(\bm{a},\bm{b}) := (\underbrace{1,\dotsc,1}_{a_1\text{ times}}, \underbrace{2,\dotsc,2}_{a_2\text{ times}}, \dotsc, \underbrace{m,\dotsc,m}_{a_m\text{ times}},\underbrace{1,\dotsc,1}_{b_1\text{ times}}, \underbrace{2,\dotsc,2}_{b_2\text{ times}}, \dotsc, \underbrace{m,\dotsc,m}_{b_m\text{ times}}),
\end{equation}
for vectors $\bm{a},\bm{b}\in \{0,\dotsc, g\}^m.$
Using this, we can write the coefficients $T_{\pi[\cdots]}$ as
\begin{equation}
    T_{\pi\left[\pr_1(1,j^{(1)}_1),\dotsc, \pr_1(m,j^{(m)}_{w_m}),\pr_1(1,k^{(1)}_1),\pr_1(1,k^{(1)}_1),\dotsc, \pr_1(m,k^{(m)}_{v_m}),\pr_1(m,k^{(m)}_{v_m})\right]} = T_{\pi\left[\bm{\eta}(\bm{w},  2\bm{v})\right]}.
\end{equation}
Plugging this back into the right-hand side of equation \eqref{eq: sumFSp1} yields
\begin{equation}
    \begin{aligned}
        \sum_{S}C_p^{|S|}\opnorm{[F]_S}_{*,p}^2 \leq  g^{g}C_p^g\sum_{\bm{w}}&\sum_{\bm{j}^{(1)},\dotsc, \bm{j}^{(m)}}\frac{1}{N^{|S|}}\frac{1}{N^{g-|S|}} \opnorm{I}_{*,p}^2 \\
        &\Bigg( \sum_{\bm{v}}\sum_{\bm{k}^{(1)},\dotsc, \bm{k}^{(m)}}\sum_{\pi \in \mathcal{S}_g} \left|T_{\pi\left[\bm{\eta}(\bm{w},  2\bm{v})\right]}\right|
        \prod_{i=1}^m \sigma_i^{w_i + 2v_i}\Bigg)^2.
    \end{aligned}
\end{equation}
Notice that the product over $\sigma$ is independent of the summation over $\bm{k}^{(1)},\dotsc, \bm{k}^{(m)}$. Since the size of each slice complement $S^c_i$ is always less than $N$, we have
\begin{equation}
    \sum_{\bm{k}^{(1)},\dotsc, \bm{k}^{(m)}} 1 \leq \sum_{k^{(1)}_1,\dotsc, k^{(1)}_{v_1} = 1}^N\cdots \sum_{k^{(m)}_1,\dotsc, k^{(m)}_{v_m} = 1}^N 1 = N^{|\bm{v}|}.
\end{equation}
By construction, $|\bm{v}| = \frac{g - |\bm{w}|}{2}$. Hence,
\begin{equation}
\label{eq: sumFSp2}
    \begin{aligned}
        \sum_{S}C_p^{|S|}\opnorm{[F]_S}_{*,p}^2 \leq  g^{g}C_p^g\sum_{\bm{w}}&\sum_{\bm{j}^{(1)},\dotsc, \bm{j}^{(m)}}\frac{1}{N^{|\bm{w}|}}\frac{1}{N^{g-|\bm{w}|}} \opnorm{I}_{*,p}^2 \\
        &\Bigg( \sum_{\bm{v}}N^{\frac{g-|\bm{w}|}{2}}\sum_{\pi \in \mathcal{S}_g}
        \left|T_{\pi\left[\bm{\eta}(\bm{w},  2\bm{v})\right]}\right|\prod_{i=1}^m \sigma_i^{w_i + 2v_i}\Bigg)^2.
    \end{aligned}
\end{equation}
By moving the $N^{\frac{g - |\bm{w}|}{2}}$ factor out of the parentheses in equation \eqref{eq: sumFSp2}, we see that it exactly cancels the factor $\frac{1}{N^{g-|\bm{w}|}}$, leaving us with
\begin{equation}
\label{eq: sumFSp3}
    \begin{aligned}
        \sum_{S}C_p^{|S|}\opnorm{[F]_S}_{*,p}^2 \leq  g^{g}C_p^g\sum_{\bm{w}}&\sum_{\bm{j}^{(1)},\dotsc, \bm{j}^{(m)}}\frac{1}{N^{|\bm{w}|}}\opnorm{I}_{*,p}^2 \\
        &\Bigg( \sum_{\bm{v}}\sum_{\pi \in \mathcal{S}_g}\left|T_{\pi\left[\bm{\eta}(\bm{w},  2\bm{v})\right]}\right|
        \prod_{i=1}^m \sigma_i^{w_i + 2v_i}\Bigg)^2.
    \end{aligned}
\end{equation}
We can do a similar treatment to the summation over $\bm{j}^{(1)},\dotsc, \bm{j}^{(m)}$, and notice that
\begin{equation}
    \sum_{\bm{j}^{(1)},\dotsc, \bm{j}^{(m)}} 1 = N^{|\bm{w}|},
\end{equation}
which cancels out the $\frac{1}{N^{|\bm{w}|}}$ factor in equation \eqref{eq: sumFSp3}. In the end, we obtain
\begin{equation}
    \begin{aligned}
        \sum_{S}C_p^{|S|}\opnorm{[F]_S}_{*,p}^2 \leq  g^{g}C_p^g\sum_{\bm{w}}
        \Bigg( \sum_{\bm{v}}\sum_{\pi \in \mathcal{S}_g}\left|T_{\pi\left[\bm{\eta}(\bm{w},  2\bm{v})\right]}\right|
        \prod_{i=1}^m \sigma_i^{w_i + 2v_i}\Bigg)^2\opnorm{I}_{*,p}^2.
    \end{aligned}
\end{equation}
Together with Theorem \ref{thm:matpol}, we have proved the desired statement.

\end{proof}

\section{Evaluating \texorpdfstring{$G_w$}{Gw} for general Hamiltonians}\label{apx:Gw}

\lemmaGwH*

\begin{proof}
Using H\"older's inequality with $p = 1$ and $q = \infty$,
\begin{align}
    G_w(H) \leq &\left(\sum_{\bm{w}: |\bm{w}| = w}\sum_{\pi\in \mathcal{S}_g}\sum_{ \bm{v}: |\bm{w}|+2|\bm{v}|=g}\ind({\pi[\bm{\eta}(\bm{w}, 2\bm{v})]}) \right)\\
    \cdot &\left(\max_{\bm{w}: |\bm{w}| = w} \sum_{\pi\in \mathcal{S}_g}\sum_{ \bm{v}: |\bm{w}|+2|\bm{v}|=g}\ind({\pi[\bm{\eta}(\bm{w}, 2\bm{v})]}) \right). \nonumber
\end{align}
By using the definition of $\bm{\eta}$ from equation \eqref{eq:eta}, we have 
\begin{align}
\label{eq:Gw_apx}
    G_w(H) \leq &\left( \sum_{\pi\in \mathcal{S}_g} \sum_{i_w,\dotsc,i_1 = 1}^{\Gamma}\sum_{j_v,\dotsc,j_1=1}^{\Gamma} \ind(\pi(i_w,\dotsc,i_1,j_v,j_v, \dotsc, j_1,j_1))\right)\\
    \cdot &\left(\max_{i_w,\dotsc,i_1} \sum_{\pi\in \mathcal{S}_g}\sum_{j_v,\dotsc,j_1}\ind(\pi(i_w,\dotsc,i_1,j_v,j_v, \dotsc, j_1,j_1))\right). \nonumber
\end{align}
where $v = (g-w)/2.$ Instead of summing over the vectors $\bm{w}$ and $\bm{v}$, we are now directly summing over the local terms in the Hamiltonian, where indices $i_1,\dotsc,i_w$ come from $\bm{w}$ and indices $j_1,\dotsc,j_v$ come from $\bm{v}$. Because of this, each $j_i$ occurs in pairs in the indicator $\ind$. The permutation $\pi$ is now a reordering of $(i_w,\dotsc,i_1,j_v,j_v, \dotsc, j_1,j_1).$ 

We are going to evaluate the two terms in equation \eqref{eq:Gw_apx} separately. The first term is 
\begin{equation}
\label{eq:firstterm}
   L_1 =  \sum_{\pi\in \mathcal{S}_g} \sum_{i_w,\dotsc,i_1 = 1}^{\Gamma}\sum_{j_v,\dotsc,j_1=1}^{\Gamma} \ind(\pi(i_w,\dotsc,i_1,j_v,j_v, \dotsc, j_1,j_1)),
\end{equation}
where $g=w+2v \geq 2$.

Fix a permutation $\pi\in \mathcal{S}_g.$ 
The sum
\begin{equation}
   S_1 = \sum_{i_w,\dotsc,i_1 = 1}^{\Gamma}\sum_{j_v,\dotsc,j_1=1}^{\Gamma} \ind(\pi(i_w,\dotsc,i_1,j_v,j_v, \dotsc, j_1,j_1))
\end{equation}
is effectively counting the number of ways we can make a commutator chain of length $g$ by selecting $w+v$ local terms of $H$ while duplicating $v$ selected terms. We shall first provide a concrete example of evaluating this sum before proceeding to the general argument. 

Take $g=6$, $w=2$ and $v=2$. Let $\pi(i_2,i_1,j_2,j_2,j_1,j_1) = (j_1, i_2, j_2, j_2,j_1,i_1)$. The sum $S_1$ corresponds to counting the number of commutator chains of the form
\begin{equation}
    [\notate[X]{H_{\gamma(j_1)}}{6}{\text{6th layer}},[\notate[X]{H_{\gamma(i_2)}}{5}{\text{5th layer}},[\notate[X]{H_{\gamma(j_2)}}{4}{\text{4th layer}},[\notate[X]{H_{
    \gamma(j_2)}}{3}{\text{3rd layer}}, [\notate[X]{H_{\gamma(j_1)}}{2}{\text{2nd layer}},\notate[X]{H_{\gamma(i_1)}}{1}{\text{1st layer}}]]]]]
\end{equation}
such that it is non-zero. We refer to the position of local terms $H_{\gamma(i)}$ in the commutator chain as ``layer''. In this case, from the innermost to the outermost commutator, $H_{\gamma(i_1)}$ is on the first layer, $H_{\gamma(j_1)}$ the second layer, etc.. We are now going to count how many such non-zero commutators there are. Starting from the first layer, there should always be up to $m$ choices. The local term on the second layer must anti-commute with the first one. Therefore, there should be up to $Q_{\max}(H)$ choices. The local term on the third layer must anti-commute with at least one of the local term of the previous layers, i.e. at least one of $H_{\gamma(j_1)}$ or $H_{\gamma(i_1)}$. This gives us up to $2Q_{\max}(H)$ choices. The fourth layer must be a duplicate of the third layer, leaving us with up to 1 choice. The fifth layer must also anti-commute with at least on of the local terms on the previous layers, giving us up to $3Q_{\max}(H)$ choices. Lastly, the sixth layer is a duplicate of the second layer, leaving us again with up to 1 choice. Thus, the number of such non-zero commutators are less or equal to $m \cdot Q_{\max}(H)\cdot 2Q_{\max}(H)\cdot 1\cdot 3Q_{\max}(H)\cdot 1 = 3! \cdot mQ_{\max}(H)^3.$ Observe that this bound is uniform for all permutations $\pi$. The only thing that changes is where the duplicate layers are in the commutator chain, but one could simply repeat the reasoning above to end up with the same upper bound. For instance, consider $\pi'(i_2,i_1,j_2,j_2,j_1,j_1) = (j_1,j_2,i_2,j_2,j_1,i_1)$ by swapping $i_2$ and $j_2$, the fourth layer now has up to $3Q_{\max}(H)$ choices, while the fifth and the sixth layer only have up to 1 choice being duplicate layers. Hence, under this permutation, the bound stays the same. Even in the extreme case such as $\pi''(i_2,i_1,j_2,j_2,j_1,j_1) = (j_1,i_1,j_1,i_2,j_2,j_2)$, the innermost commutator automatically vanishes and the number of non-zero commutator chains of this form is 0, which is still bounded by $3! \cdot mQ_{\max}(H).$ 

Now, consider general commutator chain of length $g$ 
\begin{equation}
    [\notate[X]{H_{\gamma(a_{g-1})}}{5}{\text{ $g$th layer}},[\notate[X]{\cdots}{4}{\cdots}, [\notate[X]{H_{\gamma(a_2)}}{3}{\text{ 3rd layer}}, [\notate[X]{H_{\gamma(a_1)}}{2}{\text{ 2nd layer}},\notate[X]{H_{\gamma(a_0)}}{1}{\text{ 1st layer}}]]]\cdots ]
\end{equation}
with $v$ duplicate layers distributed somewhere in the chain. When selecting local terms from the first layer to the last layer, we expect three constraints on each layer
\begin{enumerate}
    \item it is on the first layer, or
    \item it must anti-commute with at least one of the local terms on previous layers, or
    \item it is a duplicate of one of the previous layers.
\end{enumerate}
The first constraint leaves us with $m$ choices, and the last constraint only up to 1 choice. Since there are $w+v$ non-duplicate layers, the second constraint leaves us with up to  $(w+v-1)! Q_{\max}(H)^{w+v-1}$ choices. Thus, 

\begin{equation}
    S_1 \leq g^{w+v-1} mQ_{\max}(H)^{w+v-1}.
\end{equation}
 Because the above upper bound does not depend on the particular choice of $\pi\in \mathcal{S}_g$, we have
 \begin{equation}
     L_1 \leq g! g^{w+v-1} mQ_{\max}(H)^{w+v-1}.
 \end{equation}

Next, we move on to the second term:
\begin{equation}
    L_2 =  \max_{i_w,\dotsc,i_1} \sum_{\pi\in \mathcal{S}_g}\sum_{j_v,\dotsc,j_1}\ind(\pi(i_w,\dotsc,i_1,j_v,j_v, \dotsc, j_1,j_1))
\end{equation}
Fix any configuration $(i_w,\dotsc,i_1)$ and permutation $\pi\in \mathcal{S}_g$, the sum 
\begin{equation}
    S_2 = \sum_{j_v,\dotsc,j_1}\ind(\pi(i_w,\dotsc,i_1,j_v,j_v, \dotsc, j_1,j_1))
\end{equation}
is effectively counting the number of ways of making a commutator chain of length $g$ out of terms in $H$ such that $H_{\gamma(i_w)},\dotsc, H_{\gamma(i_1)}$ must be present in the commutator chain and the rest of the terms must be duplicates. Let 
\begin{equation}
    [H_{\gamma(b_{g-1})},[\cdots, [H_{\gamma(b_2)}, [H_{\gamma(b_1)},H_{\gamma(b_0)}]]]\cdots ] \neq 0
\end{equation}
be a commutator chain that satisfies the constraints in $S_2$. Now, when we select local terms from the first to last layer, we expect four constraints: 
\begin{enumerate}
    \item it is on the first layer, or
    \item it must anti-commute with at least one of the local terms on previous layers, or
    \item it is a duplicate of one of the previous layers, or
    \item it is fixed for being one of $H_{\gamma(i_w)},\dotsc, H_{\gamma(i_1)}$.
\end{enumerate}
The last two cases only leave us with up to 1 choice, but the first and second cases leaves us with a total of up to $(v-1)! \cdot m Q_{\max}(H)^{v-1}$ choices. Therefore,
\begin{equation}
    S_2 \leq g^{v-1}mQ_{max}(H)^{v-1}
\end{equation}
and because this upper bound is independent of the choice of $(i_w,\dotsc,i_1)$ and $\pi\in \mathcal{S}_g$, 
\begin{equation}
    L_2 \leq g!g^{v-1}mQ_{max}(H)^{v-1}.
\end{equation}

Putting $L_1$ and $L_2$ together and using $g! < g^g$ gives us
\begin{equation}
    \begin{aligned}
        G_w(H) &\leq L_1 \cdot L_2\\
        &\leq g^{3g-2} m^2 Q_{max}(H)^{g-2}.
    \end{aligned}
\end{equation}
This concludes the proof.
\end{proof}

\section{Average \texorpdfstring{$G_w$}{Gw} for sparse Hamiltonians}
\label{apx:avGw}
\lemmaavGwH*

\begin{proof}
Consider a Hamiltonian of the form 
\begin{equation}
    H(\bm{B}) = \sum_{i=1}^m B_i H_{\gamma(i)}
\end{equation}
where the $B_i$'s are independent Bernoulli variables such that $\Prob(B_i = 1) = p_B$ for all $1\leq i\leq m$. We use the bolded text $\bm{B}=(B_1,\dotsc,B_m)$ to denote the dependency of these Bernoulli variables. Let $\bm{b}\in \{0,1\}^m$ be a sample of the Bernoulli variables. When the Bernoulli variables take the values of this sample, we get a configuration of the sparse Hamiltonian where we can treat $\bm{b}$ as coefficients
\begin{equation}
    H(\bm{B} = \bm{b}) = \sum_{i=1}^m b_i H_{\gamma(i)}.
\end{equation}

Similar to the $G_w$ for general Hamiltonians, we now consider the configuration specific sum:
\begin{equation}
    G_w(H(\bm{B}=\bm{b})) = \sum_{\bm{w}:|\bm{w}|=w}\prod_{i\in\Supp(\bm{w})}b_i\left(\sum_{\pi\in \mathcal{S}_g}\sum_{\bm{v}:|\bm{w}| + 2|\bm{v}| = g}\prod_{j\in\Supp(\bm{v})\backslash \Supp(\bm{w})}b_j\ind(\pi[\bm{\eta}(\bm{w},  2\bm{v})])\right)^2
\end{equation}
which looks slightly more complicated than the original $G_w$ due to the extra coefficients $b_i$. The other summations over vectors $\bm{w},\bm{v} \in \Z^m_{\geq 0}$ work the same as in $G_w,$ where $|\bm{w}| + 2|\bm{v}| = g.$ The indicator $\ind$ and unary encoding $\bm{\eta}$ work precisely the same way as before. We are now interested in computing the average of this object over all samples of the Bernoulli variables, denoted by 
\begin{equation}
    \Langle G_w(H(\bm{B}))\Rangle = \left\langle \sum_{\bm{w}:|\bm{w}|=w}\prod_{i\in\Supp(\bm{w})}B_i\left(\sum_{\pi\in \mathcal{S}_g}\sum_{\bm{v}:|\bm{w}| + 2|\bm{v}| = g}\prod_{j\in\Supp(\bm{v})\backslash \Supp(\bm{w})}B_j\ind(\pi[\bm{\eta}(\bm{w}, 2\bm{v})])\right)^2\right\rangle.
\end{equation}
By construction, the product over $B_i$ is completely independent from $B_j$, we have\footnote{$B_i$ is independent from $B_j$ if $i\neq j$.}
\begin{equation}
    \Langle G_w(H(\bm{B}))\Rangle =  \sum_{\bm{w}:|\bm{w}|=w}\prod_{i\in\Supp(\bm{w})}\Langle B_i\Rangle\Langle S_{\bm{w}}(\bm{B})\Rangle,
\end{equation}
where 
\begin{equation}
    S_{\bm{w}}(\bm{B}) := \left(\sum_{\pi\in \mathcal{S}_g}\sum_{\bm{v}:|\bm{w}| + 2|\bm{v}| = g}\prod_{j\in\Supp(\bm{v})\backslash \Supp(\bm{w}))}B_j\ind(\pi[\bm{\eta}(\bm{w},  2\bm{v})])\right)^2. 
\end{equation}

Because the product of Bernoulli variables considers the support of $\bm{w}$, it is reasonable to reorganize the summation over $\bm{w}$ by their support sizes, i.e. 
\begin{equation}
    \sum_{\bm{w}:|\bm{w}| = w} \mapsto \sum_{s=0}^w \sum_{\substack{\bm{w}:|\bm{w}|=w;\\ |\Supp(\bm{w})| = s}}.
\end{equation}
Since $\Langle B_i\Rangle = p_B$ for all $i\in \{1,\dotsc,m\}$, and the product over the support of $\bm{w}$ only contains distinct $B_i$'s, we have 
\begin{equation}
\label{eq: GwHB1}
    \Langle G_w(H(\bm{B}))\Rangle =  \sum_{s=1}^w \sum_{\substack{\bm{w}:|\bm{w}|=w;\\ |\Supp(\bm{w})| = s}} p_B^s\Langle S_{\bm{w}}(\bm{B})\Rangle.
\end{equation}

Next, we want to evaluate $\Langle S_{\bm{w}}(\bm{B})\Rangle.$ By first expanding the squared term, we get 
\begin{equation}
    \Langle S_{\bm{w}}(\bm{B})\Rangle = \sum_{\pi,\pi'}\sum_{\bm{v},\bm{v}'}\left\langle \prod_{j,j'}B_jB_{j'}\right\rangle \ind(\pi[\bm{\eta}(\bm{w}, 2\bm{v})])\ind(\pi'[\bm{\eta}(\bm{w},2\bm{v}')])
\end{equation}
where we hide the summation constraints $\pi,\pi' \in \mathcal{S}_g$, $\bm{v},\bm{v}'\in \Z^g_{\geq 0}$ such that $|\bm{w}| + 2|\bm{v}| = |\bm{w}| + 2|\bm{v}'| = g$, for notational convenience. The product over $j$ and $j'$ might overlap, since the supports of $\bm{v}$ and $\bm{v}'$ might overlap. At the same time, the two vectors could also overlap with $\bm{w}$. Because the product over $B_j$ and $B_{j'}$ is taken over the non-overlapping support of $\bm{v}$ and $\bm{v}'$ with $\bm{w}$, we should first split these two vectors accordingly. We consider:
\begin{equation}
    \bm{v} = \bm{x} + \bm{y} \text{ and }\bm{v'} = \bm{x}' + \bm{y}'
\end{equation}
such that $\Supp(\bm{x}), \Supp(\bm{x}') \subseteq \Supp(\bm{w})$ and $\Supp(\bm{x})\cap \Supp(\bm{y}) = \Supp(\bm{x}')\cap \Supp(\bm{y}') = \emptyset$. In this way, we have\footnote{We extend the definition of the unary encoding to 
\begin{equation}
    \bm{\eta}(\bm{a},\bm{b},\bm{c}) = (\underbrace{1,\dotsc, 1}_{\text{$a_1$ times}},\dotsc,\underbrace{m,\dotsc, m}_{\text{$a_m$ times}},\underbrace{1,\dotsc, 1}_{\text{$b_1$ times}},\dotsc,\underbrace{m,\dotsc, m}_{\text{$b_m$ times}},\underbrace{1,\dotsc, 1}_{\text{$c_1$ times}},\dotsc,\underbrace{m,\dotsc, m}_{\text{$c_m$ times}}),
\end{equation} 
for vectors $\bm{a},\bm{b},\bm{c}\in \{0,\dotsc, g\}^m$.}
\begin{equation}
    \Langle S_{\bm{w}}(\bm{B})\Rangle = \sum_{\pi,\pi'}\sum_{\bm{x},\bm{x}'}\sum_{\bm{y},\bm{y}'}\left\langle \prod_{\substack{j\in \Supp(\bm{y}) \\ j'\in \Supp(\bm{y}')}}B_jB_{j'}\right\rangle \ind(\pi[\bm{\eta}(\bm{w}, 2\bm{x},  2\bm{y})])\ind(\pi'[\bm{\eta}(\bm{w}, 2\bm{x}', 2\bm{y}')]).
\end{equation}
Notice that now the indices over which the product of $B_j$ and $B_{j'}$ are taken, match exactly with the support of $\bm{y}$ and $\bm{y}'.$ Since we are taking this product over the supports, and its value only depends on the size of the supports and their overlap, we split the summations even further into
\begin{equation}
    \sum_{\bm{x},\bm{x}'}\sum_{\bm{y},\bm{y}'} \mapsto \sum_{a=0}^{w/2} \sum_{a'=0}^{w/2} \sum_{b=0}^{(g-w)/2}\sum_{b'=0}^{(g-w)/2}\sum_{c=0}^{\min(b,b')}\sum_{\bm{x},\bm{x}',\bm{y},\bm{y}':*}
\end{equation}
where the $*$-condition is given by 
\begin{equation}
    *:\begin{cases}|\bm{x}| = a \\ |\bm{x}'| = a'\\
    |\Supp(\bm{y})| = b\\
    |\Supp(\bm{y}')| = b'\\ 
    |\Supp(\bm{y})\cap \Supp(\bm{y}')| = c\\
    |\bm{w}|+ 2|\bm{x}| + 2|\bm{y}| = g\end{cases}.
\end{equation}
In this way, given any $b$, $b'$ and $c$, the average product over $B_j$ and $B_{j'}$ is always 
\begin{equation}
    \left\langle \prod_{\substack{j\in \Supp(\bm{y}) \\ j'\in \Supp(\bm{y}')}}B_jB_{j'}\right\rangle = p_B^{b+b'-c}
\end{equation}
by construction, and we get 
\begin{equation}
    \Langle S_{\bm{w}}(\bm{B})\Rangle = \sum_{\pi,\pi'} \sum_{a,a'=0}^{w/2}  \sum_{b,b'=0}^{(g-w)/2}\sum_{c=0}^{\min(b,b')}\sum_{\bm{x},\bm{x}',\bm{y},\bm{y}':*} p_B^{b+b'-c}\ind(\pi[\bm{\eta}(\bm{w}, 2\bm{x},  2\bm{y})])\ind(\pi'[\bm{\eta}(\bm{w}, 2\bm{x}', 2\bm{y}')]).
\end{equation}
However, the $*$-condition might not be achievable for some combinations of $a$, $a'$, $b$, $b'$ and $c$. In those cases, the contribution of the sum over $\bm{x}$, $\bm{x}'$, $\bm{y}$ and $\bm{y}'$ is simply zero. 

Given a vector $\bm{w}$ and achievable $a$, $a'$, $b$, $b'$, and $c$, let us focus on the following sum of indicators:
\begin{equation}
   \sum_{\bm{x},\bm{x}',\bm{y},\bm{y}':*}\ind(\pi[\bm{\eta}(\bm{w}, 2\bm{x}  ,2\bm{y})])\ind(\pi'[\bm{\eta}(\bm{w}, 2\bm{x}', 2\bm{y}')]).
\end{equation}
By expanding the unary encoding, we get 
\begin{equation}
\begin{aligned}
    &\sum_{\phi: \mathrm{surj}}\sum_{k_1,\dotsc,k_a\in \{i_w,\dotsc,i_1\}}\sum_{k'_1,\dotsc,k'_{a'}\in \{i_w,\dotsc,i_1\}}\sum_{\mathrm{distinct}\;j_1,\dotsc, j_{b + b' - c}=1}^m  \\
    &\quad \ind\left(\pi\left(i_w,\dotsc, i_1, k_a,k_a, \dotsc, k_1,k_1,j_{\phi\left(\frac{g-w}{2} - a\right)},j_{\phi\left(\frac{g-w}{2}-a\right)},\dotsc, j_{\phi(1)},j_{\phi(1)}\right)\right)\\ 
    \cdot &\quad \ind\left(\pi'\left(i_w,\dotsc,i_1,k'_{a'},k'_{a'},\dotsc, k'_1,k'_1, j_{\phi(g-w-a-a')},j_{\phi(g-w-a-a')},\dotsc, j_{\phi\left(\frac{g-w}{2}+1-a\right)},j_{\phi\left(\frac{g-w}{2}+1-a\right)}\right)\right),
\end{aligned}
\end{equation}
where the summations over $k_i$ and $k'_i$ come from $\bm{x}$ and $\bm{x}'$ whose support overlaps with the support of $\bm{w}$, and the summations over $j_i$ and $j'_i$ come from $\bm{y}$ and $\bm{y}'$ whose supports are disjoint with the support of $\bm{w}$. For the latter summations, we have introduced surjective functions $\phi: \{1,\dotsc, g-w-a-a'\} \to \{1,\dotsc, b+b'-c\}$ to control the subscripts of $j$, since we are organizing the summation over vectors $\bm{y},\bm{y'}$ by their support sizes. These functions distribute the summation indices $j_1,\dotsc, j_{b + b' - c}=1$ over the $j$-part of the two indicators. To capture the overlap between the supports, we require these surjective functions to satisfy the condition such that $\{\phi(1),\dotsc, \phi\left(\frac{g-w}{2}-a\right)\}$ contains $b$ distinct elements, while $\{\phi\left(\frac{g-w}{2} + 1-a\right),\dotsc, \phi(g-w-a-a')\}$ contains $b'$ distinct elements. In this way, since we are summing over $b+b'-c$ distinct elements, there is guaranteed to be an overlap of $c$ elements. Note that by introducing the constraints above for $\phi$, we are only summing over a subset of all possible surjective functions. As our aim is to deliver an upper bound, there is no need to count the number of $\phi$'s exactly. We can safely ignore these constraints and upper bound our sum using the set of all surjective functions from $\{1,\dotsc, g-w-a-a'\}$ to $\{1,\dotsc, b+b'-c\}$. 

By doing the same trick for the summation over $\bm{w}$ in equation \eqref{eq: GwHB1}, we can now upper bound $\Langle G_w(H(\bm{B}))\Rangle$ by
\begin{equation}
    \begin{aligned}
         &\sum_{s=0}^w \sum_{\psi: \mathrm{surj}}\sum_{\mathrm{distinct}\;i_1,\dotsc,i_s=1}^m \sum_{\pi,\pi'}\sum_{a,a'=0}^{w/2}  \sum_{b,b'=0}^{(g-w)/2}\sum_{c=0}^{\min(b,b')} p_B^{s+b+b' -c} \\
         & \sum_{\phi: \mathrm{surj}}\sum_{k_1,\dotsc,k_a\in \Supp(\bm{w})}
         \sum_{k'_1,\dotsc,k'_{a'}\in \Supp(\bm{w})}\sum_{\mathrm{distinct}\;j_1,\dotsc, j_{b + b' - c}=1}^m \\ 
         &\ind\left(\pi\left(i_{\psi(w)},\dotsc, i_{\psi(1)}, k_a,k_a, \dotsc, k_1,k_1,j_{\phi\left(\frac{g-w}{2} - a\right)},j_{\phi\left(\frac{g-w}{2}-a\right)},\dotsc, j_{\phi(1)},j_{\phi(1)}\right)\right)\\ 
    \cdot &\ind\left(\pi'\left(i_{\psi(w)},\dotsc, i_{\psi(1)},k'_{a'},k'_{a'},\dotsc, k'_1,k'_1, j_{\phi(g-w-a-a')},j_{\phi(g-w-a-a')},\dotsc, j_{\phi\left(\frac{g-w}{2}+1-a\right)},j_{\phi\left(\frac{g-w}{2}+1-a\right)}\right)\right)
    \end{aligned}
\end{equation}
where $\psi:\{1,\dotsc,w\} \to \{1,\dotsc,s\}$ are surjective functions to organize the subscripts of $i$'s based on support size of $\bm{w}$ (no need to impose extra conditions here), and the summations over $i_1,\dotsc,i_s$ and $j_1,\dotsc,j_{b+b'-c}$ should be interpreted as summation over distinct elements. Based on these summations, we can reformulate the summation over $k_1,\dotsc, k_a$ and $k'_1,\dotsc,k'_{a'}$ more nicely so that it does not explicitly depend on the $i$-indices, but as repetitions of these indices. 
\begin{equation}
    \begin{aligned}
         &\sum_{s=0}^w \sum_{\psi: \mathrm{surj}} \sum_{\pi,\pi'}\sum_{a,a'=0}^{w/2}  \sum_{b,b'=0}^{(g-w)/2}\sum_{c=0}^{\min(b,b')} p_B^{s+b+b' -c} \sum_{\phi: \mathrm{surj}}\sum_{\sigma,\sigma':\mathrm{func}}\\
         & \sum_{\mathrm{distinct}\;i_1,\dotsc,i_s=1}^m\sum_{\mathrm{distinct}\;j_1,\dotsc, j_{b + b' - c}=1}^m \\ 
         &\ind\left(\pi\left(i_{\psi(w)},\dotsc, i_{\psi(1)}, i_{\sigma(a)},i_{\sigma(a)}, \dotsc, i_{\sigma(1)},i_{\sigma(1)},j_{\phi\left(\frac{g-w}{2} - a\right)},j_{\phi\left(\frac{g-w}{2}-a\right)},\dotsc, j_{\phi(1)},j_{\phi(1)}\right)\right)\\ 
    \cdot &\ind\left(\pi'\left(i_{\psi(w)},\dotsc, i_{\psi(1)},i_{\sigma'(a')},i_{\sigma'(a')}, \dotsc, i_{\sigma'(1)},i_{\sigma'(1)}, j_{\phi(g-w-a-a')},j_{\phi(g-w-a-a')},\dotsc, j_{\phi\left(\frac{g-w}{2}+1-a\right)},j_{\phi\left(\frac{g-w}{2}+1-a\right)}\right)\right)
    \end{aligned}
\end{equation}
where we use arbitrary well-defined functions $\sigma: \{1,\dotsc,a\}\to \{1,\dotsc, s\}$ and $\sigma': \{1,\dotsc,a\}\to \{1,\dotsc, s\}$ to capture the repetition of indices $i_1,\dotsc,i_s.$ Similar to $\psi$ and $\phi$, these functions only care about assigning the summation indices as labels rather than their actual value. We have also moved the summation over the $i$-indices down to the second line, as they commute with the summation over $\phi$ and $\sigma$.

We can now apply the same argument of counting commutator chains as we did previously for general $G_w$. Let us first evaluate the summation over the second indicator by fixing $s$, $\psi$, $i$-indices $i_1,\dotsc, i_s$, $\phi$, $\sigma,$ $\sigma'$ and the $j$-indices $j_{\phi\left(\frac{g-w}{2} - a\right)},\dotsc,j_{\phi(1)}$ in the first indicator. For the second indicator, we only sum over the $j$-indices that do not overlap with the $j$-indices in the first indicator. By construction, there are $c$ $j$-indices in the second indicator that is already fixed due to the overlap. Thus, when only summing over the second indicator, we can expect the following constraints for selecting local terms on each layer of the corresponding commutator chain:
\begin{enumerate}
    \item it is the first layer, or
    \item it must anti-commute with one of the previous layers, or
    \item it is fixed for being one of $H_{\gamma(i_1)},\dotsc, H_{\gamma(i_s)}$, or
    \item it is a duplicate layer, or
    \item it is fixed for being in the overlap with the $j$-indices in the first indicator.
\end{enumerate} 
Constraints $3,4$ and $5$ leave us with only 1 choice, while the first and second constraints leave us with up to $g^{b'-c-1} m Q_{\max}(H)^{b'-c-1}.$
Then, for the first indicator, summing $i_1,\dotsc, i_s$ and the previously fixed $j$-indices $j_{\phi\left(\frac{g-w}{2} - a\right)},\dotsc,j_{\phi(1)}$ while keeping every other summation index still fixed, we expect the following constraints for selecting local terms on each layer of the corresponding commutator chain:
\begin{enumerate}
    \item it is the first layer, or
    \item it must anti-commute with one of the previous layers, or
    \item it is a duplicate layer.
\end{enumerate} 
The last constraint leaves us with only up to 1 option. But the first and second constraint leaves us with up to $g^{s + b-1} m Q_{\max}(H)^{s+b-1}$ options because we are selecting $s$ distinct elements in the $i$-indices and $b$ distinct element sin the $j$-indices. Putting everything together, we can upper bound $\Langle G_w(H(\bm{B}))\Rangle$ by 
\begin{equation}
\label{eq: GwHB2}
    \begin{aligned}
         \sum_{s=0}^w \sum_{\psi: \mathrm{surj}} \sum_{\pi,\pi'}\sum_{a,a'=0}^{w/2}  \sum_{b,b'=0}^{(g-w)/2}\sum_{c=0}^{\min(b,b')} p_B^{s+b+b' -c} \sum_{\phi: \mathrm{surj}}\sum_{\sigma,\sigma':\mathrm{func}} g^{s+b+b'-c-1}m^2 Q_{\max}(H)^{s+b+b'-c-1}
    \end{aligned}
\end{equation}
Since the summand now is independent of label-distributing functions $\psi,\phi,\sigma,\sigma'$ and permutations $\pi,\pi'$, we can contract these summations into the following factors. For the permutations, recall that $|\mathcal{S}_g| = g!\leq g^g$. Hence, $\sum_{\pi,\pi'}1 \leq g^{2g}$. For the label-distributing functions, we use the fact that, in general, the number of arbitrary well-defined functions from $[N] \to [M]$ is upper bounded by $M^N$. Hence, 
\begin{equation}
    \sum_{\psi: \mathrm{surj}} 1 \leq s^w, \quad \sum_{\phi:\mathrm{surj}} 1\leq (b+b'-c)^{g-w-a-a'},\quad \sum_{\sigma,\sigma':\mathrm{func}}1 \leq s^{a + a'}.
\end{equation}
By construction $s\leq g$ and $(b+b'-c)\leq g$. Therefore, 
\begin{equation}
    \sum_{\psi: \mathrm{surj}}\sum_{\phi:\mathrm{surj}} \sum_{\sigma,\sigma':\mathrm{func}}1\leq g^{w + g-w-a-a' + a+a'} = g^g.
\end{equation}
Thus, $\Langle G_w(H(\bm{B}))\Rangle$ is upper bounded by 
\begin{equation}
    \begin{aligned}
        g^{3g} \sum_{s=0}^w \sum_{a,a'=0}^{w/2}  \sum_{b,b'=0}^{(g-w)/2}\sum_{c=0}^{\min(b,b')} p_B^{s+b+b' -c} g^{s+b+b'-c-2}m^2 Q_{\max}(H)^{s+b+b'-c-2}
    \end{aligned}
\end{equation}
Contract $\sum_{a,a'}1 \leq w^2\leq g^2$, note that by construction $s+b+b'-c  \leq g$, and reorganize the summands in equation \eqref{eq: GwHB2}, we get the following upper bound
\begin{equation}
    \begin{aligned}
        \Langle G_w(H(\bm{B}))\Rangle \leq g^{4g} m^2 Q_{\max}(H)^{-2}\sum_{s=0}^w \sum_{b,b'=0}^{(g-w)/2}\sum_{c=0}^{\min(b,b')} (p_B Q_{\max}(H))^{s+b+b'-c}.
    \end{aligned}
\end{equation}
The bound in the theorem uses a different labeling where $q,q'$ takes over $b,b'$. This completes the proof. 

\end{proof}

\end{document}